\documentclass[11pt,letterpaper]{article}
\usepackage[margin=1in]{geometry}
\usepackage{amsmath}
\usepackage{xcolor}
\usepackage{amssymb}
\usepackage{amsthm}
\usepackage{url}
\usepackage[hyperfootnotes=false]{hyperref}
\usepackage{setspace}
\usepackage{upquote}
\usepackage{verbatim}
\usepackage{contour}
\usepackage{mathrsfs}
\usepackage{stmaryrd}
\usepackage{esvect}
\usepackage{tabularx}

\usepackage{enumitem}

\usepackage{tikz}
\usetikzlibrary{arrows,arrows.meta}
\usetikzlibrary{patterns,calc}
\usetikzlibrary{shapes}
\usetikzlibrary{shapes.geometric}

\usepackage{enumitem}

\newcommand{\warningsign}{\tikz[baseline=-.55ex] \node[shape=regular polygon, rounded corners=1mm, regular polygon sides=3, inner sep=0pt, draw, text=black, fill=lightgray, thick] {\textbf{!}};}
\newcommand{\protectwarningsign}{\protect\tikz[baseline=-.55ex] \protect\node[shape=regular polygon, rounded corners=1mm, regular polygon sides=3, inner sep=0pt, draw, text=black, fill=lightgray, thick] {\textbf{!}};}

\newcommand{\textsu}[1]{\textup{\textsf{#1}}}

\newcommand{\SI}{\textup{{\ensuremath{\mathcal S\!I}}}}
\newcommand{\inst}{\textup{\textsu{inst}}}
\newcommand{\perc}{\textup{\textsu{perc}}}

\newcommand{\bij}{\textup{\textsu{bij}}}

\newcommand{\la}{\lambda}
\newcommand{\ComCla}[1]{\textup{\textsu{#1}}}
\newcommand{\problem}[1]{\mbox{{\normalfont\textsc{#1}}}} 
\newcommand{\problemp}[1]{\mbox{\normalfont\textsc{$(\problem{#1})$}}}
\newcommand{\sharpP}{\ComCla{\#P}}
\newcommand{\SP}{\ComCla{\#P}}

\newcommand{\GapP}{\ComCla{GapP}}
\newcommand{\GapPP}{\ComCla{GapP}_{\ge 0}}
\newcommand{\PolynP}{\ComCla{PolynP}}
\newcommand{\PolynPP}{\ComCla{PolynP}_{\ge 0}}

\newcommand{\PPAD}{\ComCla{PPAD}}
\newcommand{\PPADS}{\ComCla{PPADS}}
\newcommand{\PPA}{\ComCla{PPA}}
\newcommand{\PPP}{\ComCla{PPP}}
\newcommand{\CLS}{\ComCla{CLS}}
\newcommand{\PLS}{\ComCla{PLS}}
\newcommand{\sharpCOUNTALL}[1][]{\ComCla{\#COUNTALL#1}}
\newcommand{\sharpCOUNTGAP}[1][]{\ComCla{\#COUNTGAP#1}}
\newcommand{\TFNP}{\ComCla{TFNP}}
\newcommand{\PTFNP}{\ComCla{PTFNP}}

\newcommand{\UP}{\ComCla{UP}}
\newcommand{\coUP}{\ComCla{coUP}}
\newcommand{\NP}{\ComCla{NP}}
\newcommand{\coNP}{\ComCla{coNP}}
\renewcommand{\P}{\ComCla{P}}
\newcommand{\CeqP}{\ComCla{C$_=$P}}

\newcommand{\PH}{\ComCla{PH}}
\newcommand{\FP}{\ComCla{FP}}

\newcommand{\acc}{\textup{\textsu{acc}}}

\newcommand{\GL}{\textup{GL}}

\newcommand{\IC}{\mathbb{C}}
\newcommand{\IQ}{\mathbb{Q}}
\newcommand{\IR}{\mathbb{R}}
\newcommand{\IN}{\mathbb{N}}
\newcommand{\IO}{\mathbb{O}}
\newcommand{\IOp}{{\mathbb{O}_+}}

\newcommand{\IZ}{\mathbb{Z}}

\newcommand{\aS}{\mathfrak{S}}
\newcommand{\sgn}{\text{sgn}}
\newcommand{\tensor}{{\textstyle\ensuremath{\bigotimes}}}
\newcommand{\HWV}{\textsu{HWV}}
\newcommand{\mult}{\textup{mult}}

\newcommand{\eq}{\textup{\textsu{eq}}}

\newcommand{\diag}{\text{diag}}
\newcommand{\suppb}{\textup{suppb}}
\swapnumbers
\numberwithin{equation}{subsection}

\newtheorem{proposition}[equation]{Proposition}
\newtheorem{claim}[equation]{Claim}
\newtheorem{lemma}[equation]{Lemma}
\newtheorem{remark}[equation]{Remark}
\newtheorem{corollary}[equation]{Corollary}
\newtheorem{theorem}[equation]{Theorem}
\newtheorem*{theorem*}{Theorem}
\newtheorem*{proposition*}{Proposition}
\newtheorem{definition}[equation]{Definition}
\theoremstyle{definition}
\newtheorem{example}[equation]{Example}

\newtheorem{conjecture}[equation]{Conjecture}

\makeatletter
\def\Ddots{\mathinner{\mkern1mu\raise\p@
\vbox{\kern7\p@\hbox{.}}\mkern2mu
\raise4\p@\hbox{.}\mkern2mu\raise7\p@\hbox{.}\mkern1mu}}
\makeatother

\makeatletter
\newcommand{\doublewidetilde}[1]{{%
  \mathpalette\double@widetilde{#1}%
}}
\newcommand{\double@widetilde}[2]{%
  \sbox\z@{$\m@th#1\widetilde{#2}$}%
  \ht\z@=.9\ht\z@
  \widetilde{\box\z@}%
}
\makeatother

\def\.{\hskip.06cm}
\def\ts{\hskip.03cm}
\def\nin{\noindent}

\usepackage[colorinlistoftodos]{todonotes}

\RequirePackage{cleveref}
\usepackage{hypcap}
\hypersetup{colorlinks=true, citecolor=darkblue, linkcolor=darkblue}
\definecolor{darkblue}{rgb}{0.0,0,0.7}
\newcommand{\darkblue}{\color{darkblue}}

\definecolor{darkred}{rgb}{0.68,0,0}
\newcommand{\darkred}{\color{darkred}}

\definecolor{darkgreen}{rgb}{0,.38,0}
\newcommand{\darkgreen}{\color{darkgreen}}

\newcommand{\defn}[1]{\emph{\darkblue #1}}
\newcommand{\defna}[1]{\emph{\darkred #1}}

\newcommand{\defng}[1]{\emph{\darkgreen #1}}

\newcommand{\per}{\mathrm{per}}

\def\bu{\bullet}
\def\emp{\varnothing}

\def\zz{\mathbb Z}
\def\nn{\mathbb N}
\def\cc{\mathbb C}
\def\rr{\mathbb R}
\def\qqq{\mathbb Q}

\def\pp{\mathbb P}

\def\sm{\smallsetminus}

\def\De{\Delta}

\def\la{\lambda}
\def\ga{\gamma}

\def\de{\delta}

\def\al{\alpha}
\def\be{\beta}

\def\ka{\ell}
\def\io{j}

\def\vr{\varrho}
\def\vp{\varphi}

\def\cB{\mathcal B}

\def\cA{\mathcal A}

\def\<{\langle}
\def\>{\rangle}
\def\La{\Lambda}
\def\Lam{\la}

\def\RR{ {\text {\rm R} } }

\def\0{{\mathbf 0}}

\def\tre{\trianglerighteq}
\def\EE{{\mathcal E}}

\def\bbc{\textbf{\textit{c}}}
\def\bbd{\textbf{\textit{d}}}

\def\bbx{\textbf{\textit{x}}}
\def\bby{\textbf{\textit{y}}}

\def\ba{\textbf{\textit{a}}}

\def\bcup{{\large{\cup}}}
\def\bcap{{\large{\cap}}}

\def\tilde{\widetilde}
\def\mP{{\mathscr{B}}}
\def\BB{\mathscr{L}}

\makeatletter
\def\@fnsymbol#1{\ensuremath{\ifcase#1\or *\or **\or \ddagger\or
   \mathsection\or \mathparagraph\or \|\or \dagger\or \dagger\dagger
   \or \ddagger\ddagger \else\@ctrerr\fi}}
\makeatother

\author{Christian Ikenmeyer\thanks{University of Liverpool, United Kingdom, christian.ikenmeyer@liverpool.ac.uk} \and Igor Pak\thanks{University of California, Los Angeles, USA, pak@math.ucla.edu}}

\title{
What is in $\sharpP$ and what is not?
}
\begin{document}
\sloppy
\maketitle
\thispagestyle{empty}

\begin{abstract}
For several classical nonnegative integer functions,
we investigate if they are members of the counting complexity class \ts $\sharpP$ \ts or \ts not.
We prove \ts $\sharpP$ \ts membership in surprising cases,
and in other cases we prove non-membership, relying on standard complexity assumptions or on oracle separations.

We initiate the study of the polynomial closure properties of \ts $\sharpP$
on affine varieties, i.e.,
if all problem instances satisfy algebraic constraints.
This is directly linked to classical combinatorial proofs of algebraic identities and inequalities.
We investigate \ts $\#\TFNP$ \ts and
obtain oracle separations that prove the strict inclusion of \ts $\sharpP$ \ts in all standard syntactic subclasses of \ts
$\#\TFNP-1$.
\end{abstract}

\bigskip
\bigskip

{\footnotesize


\bigskip
\bigskip

\noindent ACM computing classification system:

\smallskip

\noindent $\bullet$ Mathematics of computing $\sim$ Discrete mathematics $\sim$ Combinatorics

\noindent $\bullet$ Theory of computation $\sim$ Computational complexity and cryptography $\sim$ Complexity classes

\bigskip
\bigskip


\smallskip

}

\newpage
\clearpage

\setcounter{tocdepth}{2}
{\footnotesize
\tableofcontents
}
\newpage
\clearpage

\section{Introduction}\label{s:intro}

\subsection{Foreword}\label{ss:intro-foreword}
Finding a \defna{combinatorial interpretation} is an everlasting problem in Combinatorics.
Having combinatorial objects assigned to numbers brings them depth and structure,
makes them alive, sheds light on them, and allows them to be studied in a way
that would not be possible otherwise.  Once combinatorial objects are found,
they can be related to other objects via bijections, while the numbers'
positivity and asymptotics can then be analyzed.

Historically, this approach was pioneered by J.J.~Sylvester in
his ``\emph{constructive theory of partitions}''~\cite{Syl82}.
There, Sylvester was able to rederive a host of old
partition identities and prove many new ones by interpreting the
coefficients on both sides as the numbers of certain \emph{Ferrers shapes} (now called
\emph{Young diagrams}), and relating two sides to each other. G.H.~Hardy marveled
at such proofs, calling them ``striking'' and ``unlike any other''~\cite{Har40},
see also~\cite{Pak06}.

Since the 1960s, this approach became a staple in \defng{Enumerative Combinatorics},
reaching as far as undergraduate textbooks~\cite{SW86}, monographs~\cite{Loe11}
and multimedia compendia~\cite{Vie16}.  In \defng{Algebraic Combinatorics}, even one
combinatorial interpretation can introduce revolutionary changes.  Notably, a Young tableau
interpretation of the \emph{Littlewood--Richardson {\rm (LR-)} coefficients} \ts $c^\la_{\mu\ts\nu}$
was discovered in~\cite{LR34}.  These numbers describe the structure constants of the
\emph{Schur functions} multiplication \cite{Mac95,Sta12}.
Over the last few decades, this result led to an avalanche of developments,
culminating with a complete resolution of the
\emph{Horn problem}~\cite{Kly98} (see also~\cite{Ful98}), proof of the
\emph{saturation conjecture}~\cite{KT99}, and polynomial time algorithms
for the vanishing of the LR--coefficients~\cite{BI13,MNS12,Ike16}.


When a combinatorial interpretation exists it is a modern wonder, a starting
point of a combinatorial investigation.  \defna{But what if none is known?}
Such examples in Enumerative Combinatorics are too numerous
to be listed, see e.g.~\cite[$\S$4]{Pak18}.  In Algebraic Combinatorics,
the following are the top three ``most wanted''
combinatorial interpretations, all from \defng{Stanley's list}~\cite{Sta00}:

\smallskip

\nin
$\bullet$ \ \defn{Kronecker coefficients}  \ts $g(\la,\mu,\nu)$ \ts which
generalize LR--coefficients and give structure constants of tensor
products of $\aS_n$-modules.  This celebrated problem goes back to
Murnaghan~\cite{Mur38} and plays a crucial role in \defng{Geometric
Complexity Theory} (GCT), see~\cite{Mul09}. See
\cite{BDO15,IMW17,PP17,PPY19} and~$\S$\ref{ss:open-orbits},
for some recent combinatorial and complexity work on the subject.

\smallskip

\nin
$\bullet$ \ \defn{plethysm coefficients} \ts $p_\la(\mu,\nu)$ \ts which
describe decompositions of Schur functors of $\aS_n$-modules, and is the
main subject of \ts GCT7~\cite{Mul07}, see also \cite{BIP19,FI20,IP17}.
They also appear in connection to the \defng{Foulkes conjecture} in
Representation Theory, see \cite{Bri93,CIM17,Lan15}.

\smallskip

\nin
$\bullet$ \ \defn{Schubert coefficients} \ts $c(u,v,w)$ \ts which give
structure constants of the product of Schubert polynomials,
defined by Lascoux and Sch\"utzenberger~\cite{LS82} in the context
of \defng{cohomology of the Grassmannian}, see \cite{Mac91,Man01}.
We refer to \cite{Knu16,KZ20,MPP14} for examples of positive results.

\smallskip

In all three cases, there is a widespread belief that these
coefficients must have a combinatorial interpretation. A positive
resolution of either problem would be a major breakthrough culminating
decades long study. In the context of GCT,
Mulmuley conjectured~\cite{Mul09} that both Kronecker and plethysm
coefficients are in~$\sharpP$ (see \cite{Val79}), as a step towards proving that \ts $\P\ne\NP$.
Note that all three functions are in
\ts $\GapP_{\geq 0}$\ts, suggesting commonality of the obstacles.

Now, \defna{what if the community is wrong}, and these functions are
not in~$\sharpP$?  Such a possibility has only been raised recently
\cite{Pak19,Spe11}.
 According to Popper's philosophy, a belief needs to be \emph{disprovable} in order to be \emph{scientific}~\cite{Pop62}.
Until now there has been little effort towards proving that some
natural combinatorial functions are not in~$\sharpP$ (see below).
With this paper we initiate a systematic study of this problem.

We show that many natural combinatorial functions are not in~$\sharpP$
under various complexity assumptions.  In a positive direction, we
prove that many functions \emph{are} in~$\sharpP$, some strikingly close
to those that are not.

\smallskip

\subsection{Motivational examples of {\normalfont$\SP$} functions}\label{ss:intro-basics-SP}
\newcounter{motivationalcounter}
Let \ts $\GapPP$ \ts be the class of nonnegative functions in \ts $\GapP \. := \. \{f_1 - f_2 \mid f_1,f_2\in \sharpP\}$.\footnote{The closure $\GapP=\sharpP-\sharpP$ of $\sharpP$ under subtraction was introduced in \cite{fenner1994gap} and indep.\ in~\cite{Gup95}.}
More generally, we consider the class
\ts $\PolynP \. := \. \bigl\{\varphi(f_1,\ldots,f_k) \mid \varphi\in\IQ[x_1,\ldots,x_k], f_i \in \sharpP\bigr\}$,
and study the class \ts $\PolynPP$ \ts of nonnegative
functions in \ts $\PolynP$.  The place to start is to look for
natural
integer
functions in these classes and ask if they lie in~$\SP$.
For the three functions as above the problem remains open, but what is known in other
cases?  Consider the following motivating examples:

\smallskip

\refstepcounter{motivationalcounter} \label{motivational:LE}
\nin
$(\themotivationalcounter)$ \.
Let \ts $e: P\to \nn$ \ts be the number of linear extensions of~$P$, where \ts $P=(X,\prec)$ \ts
is a poset with $n$ elements.  Recall that \ts $e(P)\ge 1$, so \ts $e'(P): = e(P)-1\in \GapPP$.
Now observe that \ts $e'\in \SP$ \ts simply because finding the lex-smallest
linear extension $L$ can be done in polynomial time (see e.g.~\cite{CW95}),
so \ts $e'(P)$ \ts counts linear extensions of~$P$ that are different from~$L$.
Note aside that since \ts $e$ \ts is $\sharpP$-complete~\cite{BW91}, then so is~$\ts e'$.

\smallskip

\refstepcounter{motivationalcounter}%
\label{motivational:sperner}%
\nin
$(\themotivationalcounter)$ \.
Recall \defn{Sperner's lemma} which states that for every
$\{1,2,3\}$-coloring $\chi$ of interior vertices in a side-length $n$-triangle region~$\De_n$ of the plane
whose sides are colored $1$, $2$ and~$3$, respectively, there is a
\defng{rainbow} (123) triangle.  We trust the reader is familiar with
the setting, see e.g.~\cite{Pap94a} and \cite[$\S$6.7]{MM11}.  Here \ts $n$ \ts
is given in binary and $\chi$ is given by a polynomial size circuit.
Denote by \ts $t(\chi)$ \ts the number of rainbow triangles,
so that \ts $t(\chi)-1 \in \GapPP$.

Since the typical proof of Sperner's lemma involves tracing down the
path of non-rainbow triangles until a rainbow triangle is reached,
it may come as a surprise that \ts $t(\chi)-1 \in \SP$.  Indeed,
simply observe that \ts $t(\chi) -1 = 2 \ts t_-(\chi)$,
where \ts $t_\pm(\chi)$ \ts denotes the number of rainbow triangles
with positive/negative orientation.  This follows from
\ts $t(\chi)=t_+(\chi)+t_-(\chi)$ \ts and \ts $t_+(\chi)-t_-(\chi)=1$
\ts equations, see e.g.~\cite[$\S$8]{Pak03}.

\smallskip

\nin
\refstepcounter{motivationalcounter}%
\label{motivational:threecoloring}%
\nin
$(\themotivationalcounter)$ \.
Let \ts $G$ \ts be a simple graph with at least one edge, and let
\ts $f(G)$ \ts be the number of proper $3$-colorings of~$G$.
Then \ts $f(G)/6$ is an integer valued function in $\PolynPP$ \ts by taking into account permutations of colors.
Of the six possible $3$-colorings corresponding to a given $3$-coloring one
can easily choose the lex-smallest, implying that \ts
$f(G)/6 \in \SP$.\footnote{It is important to emphasize that while $f(G)$ is
$\SP$-complete, it is completely irrelevant to the conclusion.  Crucially,
the \defng{lex-smallest test} \ts is in~$\P$ in both this and the previous example.
In non$\ts$--$\ts\SP$ examples of this kind, the lex-smallest test is $\NP$-hard
(see below).} Such solution is not always possible in other problems,
see $\S$\ref{ss:intro-basics-SP-not}\eqref{motivnon:smith},
and algorithmic approaches to equivalence problems have been studied in
\cite{BG83, BG84, FG11}.

\smallskip

\nin
\refstepcounter{motivationalcounter}%
\nin
$(\themotivationalcounter)$ \label{motivational:matchings} \.
Let \. $\delta(k,G) \. := \. m_k(G)^2 \. - \. m_{k-1}(G)\. m_{k+1}(G)$,
where \ts $m_k(G)$ \ts is the number of $k$-matchings in graph~$G$.
The function \ts $\de\in \GapP$ \ts by definition.  By the celebrated
result of Heilmann and Lieb~\cite{HL72}, the sequence \. $m_1(G), m_2(G),\ldots$
\. is \defng{log-concave}, implying that \. $\delta \in \GapPP$\ts. This result is
a starting point of many combinatorial investigations~\cite{God93}, including notably
the ``interlacing families'' series~\cite{MSS13}.  While all signs point to \ts $\de$ \ts
being ``difficult to handle'', it was observed in~\cite{Pak19} that a beautiful
proof in~\cite{Kra96} easily implies that \ts $\de\in \SP$.

\smallskip

\nin
\refstepcounter{motivationalcounter}\label{motivational:fermat}%
\nin
$(\themotivationalcounter)$ \.
Recall \defn{Fermat's little theorem}: \ts For every prime~$\ts p$ \ts and \ts $a \in \nn$, we have:
\[
a^p \. \equiv \. a \. \pmod{p}.
\]
This is one of the most basic and most celebrated results in Number Theory,
see e.g.~\cite[$\S$3.4]{IR82}, and is the starting point of the
\defng{Miller--Rabin primality test}, see e.g.~\cite[$\S$10.8.2]{MM11}.
The theorem can be rephrased as: \ts for all \ts $a \in \IN$,  we have
\[
\varphi(a) \ := \ \tfrac{1}{p}(a^p -a) \in \IN\ts.
\]
It is readily converted into a \ts $\PolynP$ \ts function by substituting \ts
$a \leftarrow N(\phi)$ \ts as follows:

\[
\tfrac{1}{p}\big(N(\phi)^p \. - \. N(\phi)\big) \ \in \ \varphi(\sharpP) \ \subseteq \ \PolynP,
\]
where \ts $N(\phi)$ \ts is the number of satisfying assignments of a Boolean formula~$\phi$.
It was shown by Peterson~\cite{Pet72} (see also~\cite{Gol56})
that this function is actually in~$\sharpP$ (see Proposition~\ref{pro:fermat}),
by giving a combinatorial interpretation for \ts $\varphi(a)$, and in this way reproving Fermat's little theorem.
In other words, we have \ts $\varphi(\sharpP)\subseteq\sharpP$, i.e., the class \ts $\sharpP$ \ts is closed under the \defng{Frobenius map}~$\varphi$.
At the heart of the proof of Proposition~\ref{pro:fermat} is a polynomial-time algorithm for identifying lex-smallest elements as in
$\S$\ref{ss:intro-basics-SP}\eqref{motivational:threecoloring}, but here in a $\IZ/p\IZ$ orbit.

\smallskip

\nin
\refstepcounter{motivationalcounter}%
\nin
$(\themotivationalcounter)$ \.
Consider the following inequality
by Grimmett~\cite{Gri76}:
$$
\tau(G) \, \le \, \frac{1}{n} \left(\frac{2\ts m}{n-1}\right)^{n-1}
$$
for the \emph{number of spanning trees}~$\tau(G)$ in a simple graph \ts $G=(V,E)$ \ts with \ts
$|V|=n$ \ts vertices and \ts $|E|=m$ \ts edges.  One can turn this into a \ts $\GapPP$ \ts
function as follows:
$$f(G) \, := \, (2\ts m)^{n-1} \. - \. n \ts (n-1)^{n-1} \ts \tau(G)\ts.
$$
On the other hand, \emph{given that the inequality holds}, the claim \ts $f\in \SP$ \ts
is trivial since \ts $\tau \in \FP$.  Indeed, since $f(G)$ can be computed in
polynomial time by the matrix-tree theorem, we conclude that \ts $f(G)$ \ts counts
the set of $n$-bit binary strings from $0$ to~$f(G)-1$.\footnote{Combinatorialists
would argue that a combinatorial interpretation should explain \emph{why}
the inequality holds in the first place.  In fact, there are several schools
of thought on this issue (see a discussion in~\cite[$\S$4]{Pak18}).   We believe
that the computational complexity approach is both
the least restrictive and the most formal way to address this.}\,\footnote{In
 the context of GCT, motivated by the work on LR--coefficients, Mulmuley
 asks if Kronecker and plethysm coefficients count the number of integer points in a
 polytope defined by the inequalities with polynomial description~\cite{Mul09}.
 We do not work with this narrower notion in this paper. See, however,~\cite{KM18}.}
This is why it is important in the examples above that our functions are not
obviously in~$\FP$ (e.g., being $\SP$-hard is a good indication), since
otherwise the problem becomes trivial.

\smallskip

\subsection{Motivational non-examples}\label{ss:intro-basics-SP-not}
\newcounter{motivationalnonexcounter}
It may come as a surprise that the non-example
comes from the simplest of the inequalities.

\smallskip

\nin
\refstepcounter{motivationalnonexcounter}%
\nin
$(\themotivationalnonexcounter)$ \.
\defn{Cauchy--Schwartz inequality}:

\begin{equation}\label{eq:AMGM}
a^2 \. + \. b^2 \, \ge \, 2 \ts a \ts b \quad \text{where \ $a,\ts b \ts \in \ts \rr$\ts.}
\end{equation}
Now take \ts $a,b$ \ts to be counting functions.  Formally, for two Boolean formulas \ts
$\phi$ \ts and \ts $\psi$, let
\begin{equation}\label{eq:AMGM-c}
h(\phi,\psi) \, := \, N(\phi)^2 \. + \. N(\psi)^2 \. - \. 2\ts N(\phi) \ts N(\psi) \, = \,
\bigl(N(\phi) \. - \. N(\psi)\bigr)^2.
\end{equation}
By definition, the function \ts $h\in \GapPP$\ts.  Note, however,
that if \ts $h\in \SP$, then we get a polytime witness for \ts $N(\phi) \ne N(\psi)$.
This is unlikely, as it would imply the collapse of polynomial hierarchy to the second level:
\ts $\PH = \Sigma_2^\textup{\textsu{p}}$
(see Proposition~\ref{p:GapP-squared-PH}).  Colloquially, this says that under the natural
complexity assumption \ts $\PH \ne \Sigma_2^\textup{\textsu{p}}$, the Cauchy--Schwartz
inequality~\eqref{eq:AMGM} does not have a combinatorial interpretation in full generality.

\smallskip

\refstepcounter{motivationalnonexcounter}\label{notivationalnon:hadamard}%
\nin
$(\themotivationalnonexcounter)$ \.
The \defn{Hadamard inequality} for real $d\times d$  matrices states:
\begin{equation}\label{eq:Had-3x3}
\det\begin{pmatrix}
      a_{11} & \cdots & a_{1d}\\
      \vdots & \ddots & \vdots \\
      a_{d1} & \cdots & a_{dd}
    \end{pmatrix}^2 \ \le \ \prod_{i=1}^{d}\.\bigl(a_{i1}^2\ts + \ts \ldots \ts + \ts a_{id}^2\bigr)\ts.
\end{equation}
Geometrically, it says that the volume of a parallelepiped in~$\rr^d$ is at most the
product of its basis edge lengths, with equality when these edges are orthogonal.
Note that standard proofs of~\eqref{eq:Had-3x3} involve the eigenvalues of \ts $A=(a_{ij})$,
see e.g.\ \cite[$\S$2.13]{HLP52} and~\cite[$\S$2.11]{BB61},
suggesting that translation into combinatorial language would be difficult.

Substitute all \ts $a_{ij} \gets N(\phi_{ij})$ \ts in~\eqref{eq:Had-3x3}, where \ts $\phi_{ij}$ \ts
are Boolean formulas. 
Denote by \ts $H_d$ \ts the resulting counting function written
in the style of~\eqref{eq:AMGM-c}, i.e, $H_d$ is the difference of the right-hand side and the left-hand side of~\eqref{eq:Had-3x3}.  It is easy to see that \ts $H_2\in \SP$,
see~$\S$\ref{ss:main-notation}.
For \. $d\ge 3$, we prove that \ts $H_d\notin\SP$ \ts under an assumption
that we call the \defna{univariate binomial basis conjecture} (Conjecture~\ref{conj:binomialbasisuniv}).
This is a general conjecture about the structure of \ts $\sharpP$.
Formally, we show the existence of an oracle \ts $A \subseteq\{0,1\}^*$ \ts
with \ts $H_3(\vv\sharpP^A)\not\subseteq\sharpP^A$, see Theorem~\ref{p:Had}.

\smallskip

\nin
\refstepcounter{motivationalnonexcounter}\label{motivational:ahlswede}%
\nin
$(\themotivationalnonexcounter)$ \.
For a simple graph $G$ on $n$ vertices, denote by \.
$\bbd(G)=(d_1,\ldots,d_n)$ \. the degree sequence.
Consider the following natural inequality:
\begin{equation}\label{eq:planar-AD}
\pp[\text{\. $G$ \. is planar\.}] \, \le \, \pp[\text{\. $G$ \. is planar\.} \, | \, \bbd(G) \leqslant \bbc\.]\.,
\end{equation}
where $\bbc=(c_1,\ldots,c_n)$ \. is a given sequence, the inequality \ts $\bbd(G)\leqslant \bbc$ \ts
is coordinate-wise: \ts $d_i\le c_i $ \ts for all \ts $1\le i \le n$, and where
the probability is over uniform random graphs on \ts $[n]=\{1,\ldots,n\}$.
This says that being planar correlates with having small degrees.\footnote{Note aside
that the number of labeled planar graphs on~$n$ vertices can be computed in time
polynomial in~$n$ using Tutte's generating function formulas~\cite{Tut63},
see also~\cite{Noy14,Sch15}.
On the other hand, the number of labeled graphs with a given upper bound on the degrees is
likely not in~$\FP$, cf.~\cite{Wor18}.}

We can convert~\eqref{eq:planar-AD} it into a \ts $\GapPP$ \ts function
as follows:
$$
\aligned
\vr(\bbc) \, := \ \, & 2^{\binom{n}{2}} \. \#\bigl\{\text{planar graphs $G$ on $[n]$ with $\bbd(G)\leqslant \bbc$}\bigr\} \ -  \\
&\qquad - \,\. \#\bigl\{\text{planar graphs on $[n]$}\bigr\} \. \cdot \.  \#\bigl\{\text{graphs $G$ on $[n]$ with $\bbd(G)\leqslant \bbc$}\bigr\}\ts.
\endaligned
$$
This inequality is a simple special case of the \defng{Kleitman inequality}~\cite{Kle66},
which is a corollary of the \defng{Ahlswede--Daykin inequality}~\cite{AD78} (Theorem~\ref{t:AD}).
In Proposition~\ref{pro:ahlswededaykin}, we show that the polynomial inequality implied by
the Ahlswede--Daykin inequality is not in~$\SP$, again under the
univariate binomial basis conjecture.

\smallskip

\refstepcounter{motivationalnonexcounter}\label{motivnon:smith}%
\nin
$(\themotivationalnonexcounter)$ \.
Recall the following \defn{Smith's theorem}~\cite{Tut46}.
Let $e=(v,w)$ be an edge in a cubic graph~$G$. Then the number $N_e(G)$ of
Hamiltonian cycles in $G$ containing~$e$ is always even.  Denote
\ts $f(G,e):=N_e(G)/2$ \ts and observe that \ts $f\in \PolynPP$.
Is \ts $f\in \SP$?  We don't know.  This seems unlikely and remains
out of reach with existing technology.
But let us discuss the context behind this problem.

Tutte's original proof in~\cite{Tut46} uses a double counting argument.
The \defng{Price--Thomason algorithm} \ts for finding another Hamiltonian
cycle in a cubic graph~\cite{Pri77,Tho78} gives a more direct
combinatorial proof of Smith's theorem and implies that this search
problem is in \ts $\PPA$, the class defined by the polynomial parity argument.
In fact, \problem{AnotherHamiltonianCycle} \ts
is a motivational problem for \ts $\PPA$, while \ts \problem{Sperner}, see~$\S$\ref{ss:intro-basics-SP}\eqref{motivational:sperner},
is a motivational problem for \ts $\PPAD$~\cite{Pap94a}\footnote{Several versions of \problem{Sperner} on non-orientable manifolds are $\PPA$-complete \cite{Gri01,DEFLQX21}, as well as e.g.\ the problems Consensus-Halving/Necklace Splitting \cite{FG18,FHSZ20,DFHM22}, and integer factoring (assuming the GRH) \cite{Jer16}.
Main $\PPAD$-complete problems include \defng{Nash equilibrium} \cite{DGP09,CD09} and \defng{hairy ball}~\cite{GH21}.}.

Note that the Price--Thomason algorithm partitions the set of all Hamiltonian cycles into pairs,
but this pairing algorithm is known to require an exponential number of steps in the worst case, see \cite{Cam01,krawczyk1999complexity}.
A polynomial-time algorithm instead would
allow us to search for Hamiltonian cycles and only count the ones that are lexicographically smaller than their pairing partner, which would show that $N_e(G)/2 \in \sharpP$, and $(\problem{AllOtherHamiltonianCyclesThroughEdge}-1)/2\in\sharpP$.
Note that such a pairing algorithm (not for the symmetric group $\aS_2$, but for $\aS_3$) is the reason why $f(G)/6 \in \sharpP$ in $\S$\ref{ss:intro-basics-SP}\eqref{motivational:threecoloring}.

We study the basic search problem \problem{Leaf}\footnote{Search problems are often of the type \problem{AnotherSolution}, but the name does not suggest that. \problem{Leaf} for example could reasonably be called \problem{AnotherLeaf}. We adapt the search problem notation and drop the \problem{Another} prefix and mean the corresponding problem of \emph{counting} all but the given leaf.
The problem of counting \emph{all} leaves when we are \emph{not} given one is called \problem{AllLeaves}. Since all our problems are counting problems, we drop the customary $\#$ in front of the problem name, also to avoid having two $\#$ in the class names, see \S\ref{sec:countingclassesandTFNPINTRO}.
} that is used to define $\PPA$, and that arises directly from \problem{Sperner} by a parsimonious reduction from the $\PPAD$-complete problem \problem{SourceOrSink} and removing the edge directions.
We show that for the corresponding counting problem we have an oracle separation that shows $\problem{AllLeaves}^A/2 \notin\sharpP^A$.
In fact, for the counting version of \problem{Leaf}, where we are given one leaf and count all others, we show that $\problem{Leaf}^A-1 \notin\sharpP^A$.
This has to be contrasted to \problem{Sperner}, where the membership $\problem{Sperner}-1 \in \sharpP$ relativizes, i.e., holds with respect to all oracles.
The oracle instances are significantly more complicated as for the \problem{Hadamard} problem, see \S\ref{ss:intro-basics-SP-not}\eqref{notivationalnon:hadamard}.

\smallskip

\refstepcounter{motivationalnonexcounter}\label{motivnon:chess}%
\nin
$(\themotivationalnonexcounter)$ \.
We have seen that \. $\problem{Sperner}(\chi)-1 = 2t_-(\chi)$, hence \.
$(\problem{Sperner}-1)/2 \in \sharpP$.
It is easy to see that the reverse inclusion holds:
The counting class \ts $\#\PPAD\problemp{Sperner}$ \ts defined by the \problem{Sperner} problem contains \ts
$2\sharpP\ts + \ts 1$, or, in other words, \ts $\sharpP = (\#\PPAD\problemp{Sperner}-1)/2$.
For the other classes in $\TFNP$ we similarly get
\begin{eqnarray*}
\sharpP &=&  \big(\#\PPAD\problemp{Sperner}\ts - \ts 1\big)/\ts{}2 \\
&=& \#\PPADS\problemp{Sink}\. - \. 1 \\
&=& \, \#\CLS\problemp{EitherSolution(Sperner,\ts{}Iter)}\. -\. 1
\end{eqnarray*}
and these equalities relativize.
But for the more complex classes we get oracle separations:
\ts $(\#\PPA\problemp{Leaf}-1)/2$, \ts
$\#\PPP\problemp{Pigeon}-1$ \ts and \ts $\#\PLS\problemp{Iter}-1$ \ts strictly contain $\sharpP$ with respect to an oracle.

But this does not give the complete picture, since non-parsimonious reductions between complete problems give different counting classes. For example if instead of leaves in a graph we count the nodes that are adjacent to leaves (which we call preleaves), then this does not change the complexity of the search problem, but it changes the counting class from $\#\PPA\problemp{Leaf}$ to the class $\#\PPA\problemp{Preleaf}$ (note that the functions in $\#\PPA\problemp{Leaf}$ always attain odd values, while the functions in $\#\PPA\problemp{Preleaf}$ do not have this restriction).
The underlying argument is the \emph{chessplayer algorithm}, see e.g.\ \cite{Pap90, BCEIP}, which results in non-parsimonious reductions, which then give rise to a complexity class inclusion diagram where we have
an oracle with respect to which we have
a strict inclusion of \ts $\sharpP$ \ts in \emph{all} the classes \. $\#\PPAD-1$, \ts $\#\PPADS-1$, \ts
$\#\CLS-1$, \ts $\#\PPA-1$, \ts $\#\PPP-1$ \ts and \ts $\#\PLS-1$. The full class inclusion diagram of our results can be found in Figure~\ref{fig:inclusions}. The definitions of the classes and problems can be found in \S\ref{ss:TFNP-Background}.

\smallskip

These problems are syntactically guaranteed to be nonnegative, but in contrast
to the \problem{Hadamard} problem (for example), here the oracle separations
are much more delicate, as we have to fool the Turing machine while producing
instances of the correct cardinality (which is easier if the problem is a
polynomial evaluated at arbitrary \ts $\sharpP$ \ts functions).
To overpass these obstacles, we introduce the notion of a set-instantiator in Definition~\ref{def:setinstantiator}.
We will also treat cases where we have a nonnegativity guarantee, but no further information about the reason. This requires extra care, see Propositions~\ref{prop:log-concavity} and~\ref{pro:karamatadiagonalization}.

\medskip

\section{Definitions, notations and first steps}\label{s:first}
We start in $\S$\ref{ss:main-closure} with the concept of polynomial closure properties of $\sharpP$.
We then prove some simple separations in $\S$\ref{ss:main-monotone} and~$\S$\ref{ss:main-nonmonotone},
as a warmup before our main results in the next section.

\subsection{Basic notation}\label{ss:main-notation}
Let \ts $\nn = \{0,1,2,\ldots\}$, \ts $\qqq_+=\{x\in \qqq, \ts x \ge 0\}$.
For \ts $i \in \IN$ \ts and \ts $x \in \IR$,  we write \.
$$
\binom xi  \, = \, \frac 1 {i!} \, x \ts (x-1) \. \cdots \. (x-i+1)\ts.
$$
In particular, \ts $\binom{i}{0}=\binom{i}{i}=1$. We think of \. $\binom xi$ \.
as a rational polynomial of degree~$i$.
Note that for \, $0\le x<i$, \. $x \in \IN$, we have \. $\binom{x}{i}=0$.
For a vector \ts $(a_1,\ldots,a_n)\in \rr^n$, we use both \ts $\vv{a}$ \ts and
\ts $\ba$ \ts to denote this vector.

We are assuming the reader is familiar with basic complexity theory and
standard complexity classes: \. $\P$, \. $\NP$, \. $\UP$, \. $\PH$,  \. $\FP$, \. $\SP$, \.
$\GapP$, \. $\PPA$ \. and \. $\PPAD$.  We refer to \cite{AB,MM11,Pap94} for
the definitions and standard results, and to \cite{Aar16,Wig19} for further
background.

\smallskip

\subsection{Closure properties}\label{ss:main-closure}
We say that a map \ts $\vp: \nn^k\to \qqq$ \ts is \defn{integer-valued} \ts if
it only attains integer values.  Similarly, map \ts $\vp$ \ts is \defn{nonnegative},
write \ts $\vp \geqslant 0$, if it only attains nonnegative values.

We  say that
\ts $\vp$ \ts is a \defn{closure property} of~$\SP$, if for all \ts $f_1,\ldots,f_k\in \SP$ \ts
we have \ts $\vp(f_1,\ldots,f_k) \in \SP$.  More concisely, we also write:
\[
\varphi\big(\vv{\sharpP}\big)\, \subseteq\, \sharpP\ts.
\]

This is a generalization of the notation \ts $\GapP=\sharpP-\sharpP$ \ts
from \cite{fenner1994gap}.\footnote{When defining $\vv\sharpP$, two different definitions
are equivalent (in the same way as for $\GapP$).  First, one can  define \ts $\vv\sharpP$ \ts
via $k$ many nondeterministic polynomial time Turing machines and consider the $k$-vector
of their number of accepting paths as the output.  Alternatively, one can define it via one
nondeterministic polynomial time Turing machine that has $k$ many different states
of acceptance and one reject state
(these states of acceptance are usually labeled with $+1$ and $-1$ in \ts $\GapP$).
This is a complexity class of multi-output functions, as, for example, considered
in~\cite{Val76}.}
Let \ts $S\subseteq \nn^k$ \ts be a fixed subset.  We say that
\ts $\vp$ \ts is a \defn{closure property} of~$\ts \SP\ts$ \defn{restricted} \ts to~$S$ \ts (or \defn{on}~$S$),
if for all \ts $f_1,\ldots,f_k\in \SP$ \ts which satisfy \. $\big(f_1(w),\ldots,f_k(w)\big) \in S$ \. for all \ts $w \in\{0,1\}^*$,
we have \ts $\vp(f_1,\ldots,f_k) \in \SP$.

Note that we evaluate these \ts $\SP$ \ts functions on the same input.
For example, in the notation of~$\S$\ref{ss:intro-basics-SP}\eqref{motivational:sperner}, the map
\ts $\vp(t_-,t_+) := t(\chi) = t_- + t_-$ \ts is restricted to \ts $S=\{(t_-,t_+) \mid t_+-t_-=1\}$.
Similarly, in the notation of~$\S$\ref{ss:intro-basics-SP}\eqref{motivational:threecoloring}, we have \ts $S=6\ts\IN$.

We write the restriction to~$S$ as a subscript, usually denoted $\vv\sharpP_{\in S}$, but the property ``$\in S$'' is sometimes notationally replaced by other properties such as ``$\geq\!\!1$'' (in which case $S=\IN_{\geq 1}$) or ``even'' (in which case $S=2\IN$).  For example, in notation
of~$\S$\ref{ss:intro-basics-SP}\eqref{motivational:LE},  we have \ts $e(P)\in \SP_{\ge 1}$.  Similarly,
in the notation of~$\S$\ref{ss:intro-basics-SP-not}\eqref{motivnon:smith}, we have \.
$N_e(G) \in \SP_{\textup{even}}$.  This allows us to write statements such as
$$
\sharpP_{\geq 1}+1 \. \subseteq \. \sharpP\., \quad \text{and} \quad
\sharpP^A_{\textup{even}}/2\. \not\subseteq\. \sharpP^A
$$
for the oracle~$A$ separation.  More generally, in the multivariate case we write
\[
\varphi\big(\vv{\sharpP}_{\in S}\big) \. \subseteq \. \sharpP
\]
for the closure property of~$\SP$ restricted to~$S$.
\cite{HR00}
study the univariate case and call such a restriction a counting property.
These univariate restrictions also play a role in \cite{C+89}
and are the main focus of \cite{GW87}.
The most famous example is probably $\UP=\sharpP_{\in\{0,1\}}$ (if one identifies languages with their characteristic functions, which we do), see \cite{Val76, GS88, Ko85, HT03}. In some contexts it is natural to consider a promise version of $\UP$, see \cite{VV85}, but that is \emph{different} from what we consider here. To make connections to $\TFNP$ more visible, we define $\#\TFNP := \sharpP_{\geq 1}$.

Let \ts $\vp,\psi \in \qqq[x_1,\ldots,x_k]$ \ts be rational polynomials.
We write \. $\vp\geqslant_{\#} \psi$ \.
if \.
$$\vp(f_1,\ldots,f_k) \. - \. \psi(f_1,\ldots,f_k) \. \in \ts \SP \quad \text{for all} \quad f_1,\ldots,f_k\in \SP\ts,
$$
or, equivalently, \. $(\varphi-\psi)(\vv\sharpP)\subseteq\sharpP$.
For example, \. $x^2+3 x\geqslant_\#0$.  Less obviously, \. $x^2\geqslant_\# x$, since \ts $x^2-x= 2\binom{x}{2}$ \ts
which counts unordered pairs \ts $(i,j)$, where \ts $1\le i < j \le x$.  For the Hadamard inequality~\eqref{eq:Had-3x3},
we easily have \ts $H_2\geqslant_\# 0$, since
$$\aligned
\det\begin{pmatrix}
      a & b\\
      c & d
    \end{pmatrix}^2  \, & = \, (ad\ts -\ts bc)^2 \, = \, a^2d^2 \. - \. 2 \ts abcd \. + \. b^2c^2 \\
    & \leqslant_\# \,
     a^2c^2 \. + \. a^2d^2 \. + \. b^2c^2 \. + \. b^2d^2 \, = \, (a^2\ts +\ts b^2)(c^2\ts + \ts d^2)\ts.
     \endaligned
$$
We emphasize again that over the reals this is \emph{not} a valid proof of the Hadamard inequality for $2\times 2$ matrices
since the \ts $2\ts abcd$ \ts term can be negative.  The inequality \ts $H_2(a,b,c,d)\geq 0$ \ts over the reals
follows from the Cauchy--Schwartz inequality in this case.

\smallskip

\subsection{Complete squares}\label{ss:main-monotone}
As in the introduction, we have \. $\GapP \. = \. \sharpP\. - \. \sharpP \. = \. \{f_1 - f_2 \mid f_1,f_2\in \sharpP\}$.
We use the notation \. $[\ComCla{C}=0]$ \. to denote the class of languages \. $L\subseteq\{0,1\}^*$ \. for which there
exists a function \. $f \in \ComCla{C}$ \. with: \ts $w \in L$ \ts if and only if \ts $f(w)=0$.
For example, \ts $[\sharpP=0]=\coNP$ and $[\GapP=0]=\CeqP$.  The following proposition about $k$-th powers of $\GapP$ functions is well known:

\begin{proposition}\label{p:GapP-squared-PH}
If \. $\GapP^k \subseteq \sharpP$ \. for some even~$k$, then \. $\PH = \Sigma_2^\textup{\textsu{p}}$\ts.
\end{proposition}
\begin{proof}
Recall that \. $\PH \subseteq \NP^{\CeqP}$, which can be found for example in \cite{Tar91}, \cite{Gre93}
and~\cite{Cur16},
which follows from Toda's $\PH\subseteq\P^{\sharpP}$ theorem (see \cite{Tod91, KVVY93, For97, For09}) as follows:
$
\PH = \NP^\PH \subseteq \NP^{\P^\sharpP} = \NP^\sharpP = \NP^\GapP \ \textup{=\footnotemark}\  \NP^\CeqP.
$
\footnotetext{$\NP^\GapP \subseteq \NP^\CeqP$ holds because instead of calling the oracle for a function $g\in\GapP$
we can nondeterministically guess its
return value $i = g(w)$ and call the $\CeqP$ oracle $[g - i = 0]$ on the input $w$ to check for correctness (continue the
computation if the guess was correct; reject the computation if the guess was wrong).}
We now observe:\\
$\PH \subseteq \NP^{\CeqP}=\NP^{[\GapP=0]} = \NP^{[\GapP^k=0]} \subseteq \NP^{[\sharpP=0]} = \NP^{\coNP} = \Sigma_2^\textup{\textsu{p}}$.
\end{proof}

\begin{corollary}[{\em \defn{Cauchy--Schwartz inequality}}{}]\label{c:AMGM} \. $a^2+b^2 \ts \not\geqslant_\# 2ab$ \. unless \. $\PH = \Sigma_2^\textup{\textsu{p}}$\ts.
\end{corollary}

This innocent looking corollary has immediate negative consequences on the existence of combinatorial proofs for inequalities (in the sense of combinatorial interpretations of the difference of both sides of the inequality), for example the Cauchy inequality or the Alexandrov--Fenchel inequality, see \S\ref{sec:completesquaresandnonmonotone} for the details.

Given the success in our \defng{matching polynomial}
example $\S$\ref{ss:intro-basics-SP}\eqref{motivational:matchings}, one can ask if this example
is generalizable to other log-concave properties.  Formally,
is it true that \ts $g^2 \ts \geqslant_\# \ts f\ts h$ \ts when
functions \ts $(f,g,h)$ \ts are restricted to \ts
$S = \bigl\{(f,g,h)\in\nn^3 \mid g^2-fh \geq 0\bigr\}$?
We give a negative answer to this question, suggesting that
many log-concavity results and open problems (see~$\S$\ref{ss:open-log-concavity})
are unlikely to have a direct combinatorial proof.

\begin{proposition}[{\em \defn{Log-concavity}}{}]\label{prop:log-concavity}
Let \. $\vp(f,g,h) := g^2-fh$,  and let \.
$S := \{(f,g,h)\in\IN^3 \mid g^2-fh \geq 0\}$.
Then \. $\varphi\big(\sharpP^{\times 3}_{\in S}\big) \ts \not \subseteq \ts \sharpP$ \. unless \.  $\PH=\Sigma_2^{\textsu{p}}$.
\end{proposition}

\begin{proof}
Let \ts $f:=1$, \ts $g:=(x+y)$ \ts and \ts $h:=4\ts xy$. Observe that \ts
$g^2-fh = (x-y)^2\geq 0$ \ts for all \ts $x,y\in \rr$.
The resulting complete square allows us to use Corollary~\ref{c:AMGM} and prove the result.
We now formalize this approach in the notation above.

Let \ts $\vv\gamma:\IN^2\to\IN^3$ \ts defined by \ts $(x,y)\mapsto\big(1,(x+y),4xy\big)$.
Then \ts $\vv\gamma\big(\sharpP^{\times 2}\big) \subseteq \sharpP^{\times 3}_{\in S}$.  Note that
on the left-hand side we have no index anymore, as the image is guaranteed to lie in $S$.
If we have \. $\varphi\big(\sharpP^{\times 3}_{\in S}\big) \subseteq \sharpP$, then it follows \.
$\varphi\big(\vv\gamma\big(\sharpP^{\times 2}\big)\big) \subseteq \sharpP$. But we have \. $\varphi\big(\vv\gamma\big(\sharpP^{\times 2}\big)\big) = \GapP^2$.
We conclude: \ts if \. $\varphi\big(\sharpP^{\times 3}_{\in S}\big) \subseteq \sharpP$ \. then \.$\GapP^2 \subseteq \sharpP$.
Hence, by Proposition~\ref{p:GapP-squared-PH}, we have \. $\PH=\Sigma_2^{\textsu{p}}$.
\end{proof}

\smallskip

\subsection{Non-monotone closure properties}\label{ss:main-nonmonotone}
A map \ts $\vp: \nn^k\to \qqq$ \ts is called \defn{monotone} if \ts
$\vp(a_1,\ldots,a_k)\le \vp(a_1',\ldots,a_k')$ \ts for all integer \ts $a_1\le a_1'$,
\ldots, $a_k\le a_k'$. For example, polynomials \ts $x/2$, \ts $x-1$ \ts and \ts
$x+y$ \ts are monotone, but \ts $x^2-2x$ \ts and \ts $(x-y)^2$ \ts are not.

\smallskip

\begin{proposition}[Non-monotone closure properties]\label{p:non-monotone}
Fix \ts $k\ge 1$. If \ts $\varphi : \IN^k\to\IN$ \ts is a
non-monotone closure property of \ts $\sharpP$, then \ts $\UP=\coUP$.
\end{proposition}
\begin{proof}
Let \ts $\varphi$ \ts be a $k$-variate non-monotone closure property of \ts
$\sharpP$.  Then there exists \. $\vv c \in\IN^k$ \. and \. $i \in [k]$ \.
with \. $\varphi(\vv c)>\varphi(\vv c+\vv{e_i})$, where \ts $\vv{e_i}$ \ts
is the $i$-th standard basis vector.  Let \. $D:=\varphi(\vv c)$,
and let \. $d:=\varphi\big(\vv c+\vv{e_i}\big)$.
Note that
$$
\psi \, : \, f \. \mapsto \. \binom{\varphi(f\ts\cdot\ts \vv{e_i} \. + \. \vv c)}{D}
$$
is a univariate closure property of \ts $\sharpP$.

Now let \. $f \in \UP = \sharpP_{\in\{0,1\}}$ \. be arbitrary.
Let $\beta = f(w)$ for an arbitrary $w \in \{0,1\}^*$.
We have \ts $\beta=0$ \ts if and only if \. $\beta\cdot \vv{e_i}+\vv c=\vv c$, and
if and only if \. $\varphi(\beta \cdot \vv{e_i}+\vv c)=D$.
Similarly, we have \ts $\beta=1$ \ts if and only if \. $\beta\cdot \vv{e_i}+\vv c=\vv c+\vv{e_i}$, and
if and only if \. $\varphi(\beta\cdot \vv{e_i}+\vv c)=d$.
Therefore, \. $\psi(\beta)=1-\beta$.
Hence, we have seen that \ts $\psi(f)=1-f$ \ts and that
\ts $\psi(f) \in \UP$. It follows that \. $f \in 1-\UP = \coUP$.
\end{proof}

A similar use of binomial coefficients can also be found in \cite{BG92}.
Curiously, \ts $x(x-1)^2\geqslant_\# 0$ \ts since \ts $x(x-1)^2 = 6\binom{x}{3} + 2\binom{x}{2}$,
yet by Proposition~\ref{p:non-monotone} we have:

\begin{corollary}\label{c:square-weak} \.
$(x-1)^2 \ts \not\geqslant_\# 0$ \. unless \. $\UP=\coUP$.
\end{corollary}

Note that \ts $a^2+b^2\geqslant ab$ \ts holds over $\nn$ (but not over~$\rr$),
and is halfway between \ts $a^2+b^2\geqslant 2ab$ \ts
and \ts $a^2+b^2\geqslant 0$.  So one can ask if \ts $a^2+b^2\geqslant_\# ab$.
Observe that \ts $\vp(a,b):=a^2-ab+b^2$ \ts is non-monotone: \ts $\vp(0,2)=4$ \ts
and \ts $\vp(1,2)=3$.  Proposition~\ref{p:non-monotone} then gives:

\begin{corollary}\label{c:AMGM-weak} \.
$a^2+b^2 \ts \not\geqslant_\# \ts ab$ \. unless \. $\UP=\coUP$.
\end{corollary}

Recall the \defn{Motzkin polynomial} \ts $M(x,y):=x^2y^4 + x^4y^2-3x^2y^2+1$.
It follows from the AM-GM inequality applied to positive terms,
that \ts $M(x,y)\ge 0$ \ts for all $x,y\in \rr$.
On the other hand, this polynomial is famously not a \defng{sum of squares}, and is a
fundamental example in Semidefinite Optimization,
see e.g.~\cite{Ble13,Mar08}.
Now, observe that \ts $M(x,y)$ \ts is not monotone: \ts $M(0,1)=1$ \ts and \ts $M(1,1) =0$.
Proposition~\ref{p:non-monotone} then gives:

\begin{corollary}\label{c:Motzkin} \.
$M(x,y) \ts \not\geqslant_\# \ts 0$ \. unless \. $\UP=\coUP$.
\end{corollary}

\medskip

\subsection{The binomial basis theorem}\label{ss:main-binomial}
In this section we recall a classical result,
describing all relativizing polynomial closure properties of $\sharpP$ and $\GapP$. Note that we
considered only non-monotone examples in~$\S$\ref{ss:main-monotone} and~$\S$\ref{ss:main-nonmonotone}, while many
natural polynomials are monotone.
Clearly, every polynomial with integer coefficients is a closure property
of~$\GapP$, but might not be a closure property of~$\sharpP$.
If all coefficients of $\varphi$ are nonnegative integers, then $\varphi$
is clearly a closure property of $\sharpP$, but we have seen
that there are more, e.g.\ \ts $\frac 1 2 x^2- \frac 1 2  x = \binom{x}{2} \geqslant_\# 0$.

The main tool to shed light onto these issues is the binomial basis for
the polynomial ring $\IQ[x_1,\ldots,x_k]$, which is given by the polynomials \ts
$\beta_\ba \in \IQ[x_1,\ldots,x_k]$, $\ba=(a_1,\ldots,a_k) \in \IN^k$, via
\[
\beta_\ba(x_1,\ldots,x_k) := \binom{x_1}{a_1}\cdots \binom{x_k}{a_k}.
\]
Every polynomial has a unique expression of finite support in this basis.
The univariate version is well-known under the name of classical numerical polynomials.
The study of the multivariate version goes back to Nagell~\cite{Nag19}.
This basis explains the behavior we observe, as stated in the following fundamental theorem, for which the proof is split up into the $\sharpP$ part, see \cite[Thm.~3.13]{HVW95}, and the $\GapP$ part, see \cite[Thm.~6]{Bei97} (see also the bibliographic notes in \cite[\S5.6]{HO02}).
The $\GapP$ part can be obtained as a direct consequence of the algebraic properties of the binomial basis, see \S\ref{subsec:binomialbasisandivf}.
We will reprove the $\sharpP$ part as a direct corollary of our much more general Diagonalization Theorem~\ref{thm:diagonalization}.

\smallskip

\begin{theorem*}[{{\em\defna{Binomial basis theorem}}, {\em \defn{Thm.~\ref{thm:closure}}}}{}]\ The following four properties for a multivariate polynomial $\varphi$ over $\IQ$ are equivalent:\\
\begin{tabularx}{\textwidth}{Xl}
$\bullet$\ $\varphi$ is a closure property of \ts $\GapP$
&
$\bullet$\ $\varphi$ is a relativizing closure property of \ts $\GapP$
\\
\begin{minipage}{9cm}
$\bullet$\ $\varphi$ is integer-valued\\\mbox{~}
\end{minipage}
&
\begin{minipage}{9cm}
$\bullet$\ the expression of \ts $\varphi$ \ts over the binomial basis has only integer coefficients.
\end{minipage}
\end{tabularx}\\
\noindent Moreover, the following are equivalent:\\
\begin{tabularx}{\textwidth}{Xl}
$\bullet$\ $\varphi$ is a closure property of \ts $\GapP_{\geq 0}$
&
$\bullet$\ $\varphi$ is a relativizing closure property of \ts $\GapP_{\geq 0}$
\\
\begin{minipage}{6cm}
$\bullet$\ $\varphi$ is integer-valued and attains only nonnegative integers\\\mbox{~}
\end{minipage}
&
\begin{minipage}{9cm}
$\bullet$\ the expression of $\varphi$ over the binomial basis has integer coefficients and $\varphi$ attains only nonnegative integers if evaluated at integer points in the nonnegative cone.
\end{minipage}
\end{tabularx}\\
\noindent Moreover, the following are equivalent:
\begin{itemize}
\item $\varphi$ is a relativizing closure property of \ts $\sharpP$,
\item the expression of \ts $\varphi$ \ts over the binomial basis has only nonnegative integer coefficients.
\end{itemize}
\end{theorem*}

\smallskip

\noindent Note that even though \ts  $\sharpP$ \ts  and  \ts $\GapP_{\geq 0}$ \ts
have different relativizing closure properties,
this does not unconditionally separate these two classes.
Note also that the theorem implies that \emph{all polynomial closure properties of
\ts $\GapP$  \ts and \ts  $\GapP_{\geq 0}$  \ts relativize}.
We conjecture that this is true for  \ts $\sharpP$ \ts  as well, which is Conjecture~\ref{conj:binomialbasis}.
Proving this, however, would imply  \ts $\sharpP \ne \sharpP^\NP$  \ts (and hence  \ts $\P\neq\NP$),
even just for the univariate \ts $\varphi=\binom{x-1}2$, see Theorem~\ref{thm:relatPNP}.
We get the following sequence of implications:
\[
\P = \NP \ \ \ \Longrightarrow \ \ \ \sharpP = \sharpP^\NP \ \ \ \stackrel{\textup{Thm.}~\ref{thm:relatPNP}}{\Longrightarrow} \ \ \  \binom{\sharpP-1}{2}\subseteq \sharpP \ \ \ \stackrel{\textup{Prop.}~\ref{p:non-monotone}}{\Longrightarrow} \ \ \ \UP = \coUP.
\]
We call polynomials $\varphi$ whose expression over the binomial basis has only nonnegative integer coefficients \defn{binomial-good}, all others are called \defn{binomial-bad}.

The fact that $H_d$ in \S\ref{ss:intro-basics-SP-not}\eqref{notivationalnon:hadamard} is binomial-bad gives us the described separation, see Proposition~\ref{p:Had}.
One famous instance of a binomial-good polynomial is the Frobenius map from~$\S$\ref{ss:intro-basics-SP}\eqref{motivational:fermat}.
There, binomial-goodness can be interpreted as a combinatorial proof of Fermat's little theorem, see Proposition~\ref{pro:fermat}
and its proof by Peterson.

Theorem~\ref{thm:closure} is
proved in \cite[Thm.~3.13]{HVW95} together with \cite[Thm.~6]{Bei97},
which is in fact an extension of an argument of \cite[Thm~3.1.1]{C+89} and \cite[p.~310]{OH93} about the weakness of $\sharpP$ machines in the presence of oracles.
We prove it as a corollary of our Diagonalization Theorem~\ref{thm:diagonalization} (see \S\ref{sec:diagthmINTRO}), which greatly extends Theorem~\ref{thm:closure}.

\subsubsection{The Ahlswede--Daykin inequality}
More advanced problems, where the set $S$ is given as a semialgebraic set, are also possible, for example for the Ahlswede--Daykin inequality,
see~$\S$\ref{ss:intro-basics-SP}\eqref{motivational:ahlswede} and~$\S$\ref{sec:ahlswededaykinkleitman}.

\smallskip

\begin{proposition}[{\em \defng{Ahlswede--Daykin inequality}}{}]
\label{pro:ahlswededaykin}
Let
$$
S \. :=
\left\{\bigl(\alpha_0,\alpha_1,\beta_0,\beta_1,\gamma_0,\gamma_1,\delta_0,\delta_1,h_1,h_2,h_3,h_4\bigr)\in\IN^{12} \ \Biggl| \
\aligned
& \ts \alpha_0\beta_0+h_1 = \gamma_0\delta_0\ts, \. \alpha_0\beta_1+h_2 = \gamma_0\delta_1\\
& \ts \alpha_1\beta_0+h_3 = \gamma_0\delta_1\ts, \.
\ts \alpha_1\beta_1+h_4 = \gamma_1\delta_1
\endaligned
\right\}
$$
and let
$$
\varphi \, := \, \bigl(\gamma_0+\gamma_1\bigr)\big(\delta_0+\delta_1\big)
\. - \. \big(\alpha_0+\alpha_1\big)\big(\beta_0+\beta_1\big).
$$
Then, under the Univariate Binomial Basis Conjecture~\ref{conj:binomialbasisuniv}, we have \.
$\varphi(\vv \sharpP_{\in S}) \not\subseteq \sharpP$.
\end{proposition}

\begin{proof}
Define \ts $\vv\gamma:\IN\to\IN^{12}$ \ts via \.
$\gamma(x)=\bigl(1,1,x,x,x,1,1,x,0,2\binom{x}{2},2\binom{x}{2},0\bigr)$.
Then \ts $\vv\gamma(\sharpP) \subseteq \sharpP^{\times 12}_{\in S}$\ts. Note that
on the left-hand side we have no index anymore, as the image is guaranteed to lie in $S$.
Assume for the sake of contradiction that we have an inclusion \.
$\varphi(\sharpP^{\times 12}_{\in S}) \subseteq \sharpP$. Then it follows that we have an inclusion \.
$\varphi\big(\vv\gamma(\sharpP)\big) \subseteq \sharpP$.
Conjecture~\ref{conj:binomialbasisuniv} implies that this inclusion relativizes.
Therefore, by Theorem~\ref{thm:closure} the univariate polynomial \ts $\varphi\circ \vv\gamma$ \ts is binomial-good.
But we have \. $\varphi\big(\vv\gamma(f)\big) = f^2-2f+1=2\binom{f}{2}-f+1$, which is binomial-bad, a contradiction.
\end{proof}

\section{Main results}\label{s:main}
In this section we state our main results.
It $\S$\ref{sec:diagthmINTRO} we state the Diagonalization Theorem and we give Karamata's inequality as an involved example for its application. In $\S$\ref{sec:countingclassesandTFNPINTRO} we lift these techniques to
handle $\TFNP$ and its subclasses.
We obtain several oracle separations from $\sharpP$ in this way, see Figure~\ref{fig:inclusions}.

\subsection{The diagonalization theorem}
\label{sec:diagthmINTRO}
In Proposition~\ref{pro:ahlswededaykin} the set $S$ lies on an affine algebraic variety, and the proof goes by embedding a curve given by binomial-good polynomials. This is a way of finding separations, but it remains unclear if such curves always exist or how we can find them.
In general, if $S$ lies on an affine variety $Z$ with vanishing ideal $I$, then we know that if there exists a polynomial $\xi \in I$ such that $\varphi+\xi$ is binomial-good, then $\varphi$ is a polynomial closure property of $\sharpP$ on $S$. This is exactly the insight that gives \ts $\problem{Sperner}-1 \in \sharpP$, where all instances lie on the variety \ts $\{(t_+, t_-) \in \IN^2 \. : \. t_+ - t_- -1 = 0\}$.

The reverse is true in the important case of graph varieties (all our examples fall in this category), as we show in the following Diagonalization Theorem.  Formally, assume that there exist \ts
$\ka \in \{0,\ldots,k\}$, and
polynomial maps \ts $\zeta_b:\IQ^\ka\to\IQ$, where \ts $b\in\{\ka+1,\ldots,k\}$, such that~$Z$ is the image \. $\big(f_1,\ldots,f_\ka,\zeta_{\ka+1}(f_1,\ldots,f_\ka),\ldots,\zeta_k(f_1,\ldots,f_\ka)\big)$.
In this case the vanishing ideal $I$ is generated by the \ts $\zeta_b-f_b$, see~$\S$\ref{sec:graphvarieties}.
We call a coset $\varphi+I$ binomial-good, if it contains a binomial-good polynomial, otherwise $\varphi+I$ is binomial-bad.

\smallskip

\begin{theorem*}[{\em \defna{Diagonalization Theorem}, informal version, see \defn{Thm~\ref{thm:diagonalization}}}{}]
Fix $k$ and \ts $0 \leq \ka \leq k$.
Let \ts $\varphi \in \IQ[f_1,\ldots,f_k]$.
Fix functions \ts $\zeta_b \in \IQ[f_1,\ldots,f_\ka]$.
Set $I$ to be the ideal generated by the \ts $\zeta_b(f_1,\ldots,f_\ka)-f_b$, where \ts $\ka+1 \leq b \leq k$.
Fix a function \. $\problem{Multiplicities}:\{0,1\}^*\to\IN^k$ and $\vv t \in S$.
For $A = \bigcup_{j\geq 0}A_j$, $A_j \subseteq \{0,1\}^j$, $\tilde A_{j}$ being the set of length $j-1$ suffixes of the strings in $A_{j}$ that start with a 1, define
\[
p_{A}(w) := \varphi(\problem{Multiplicities}(\tilde A_{|w|})) \ \ \. \textup{ if } \ \ \. A_{|w|}(0^{|w|}) = 1,
\quad
\textup{and} \ \ \. p_{A}(w) := \varphi(\vv t) \ \  \textup{ otherwise}.
\]
Assume further technical conditions, e.g., the existence of set-instantiators for \problem{Multiplicities}.
If \ts $\varphi+I$ \ts is binomial-bad, then there exists \ts $A\subseteq\{0,1\}^*$, such that
for every nondeterministic polytime Turing machine $M$ there exists~$j$, such that \ts
$p_A(0^j)\neq \#\acc_{M^A}(0^j)$;
and whenever \ts $A(0^j)=1$, we have \ts $\problem{Multiplicities}(\tilde A_j)\in S$.
\end{theorem*}

The Diagonalization Theorem~\ref{thm:diagonalization} is the technical heart of this paper.
It is stated in high generality,
and we apply it to a large set of examples of very different flavor, such as for example the Hadamard inequality or $\#\PPA-1$.
Its proof relies on the Witness Theorem~\ref{thm:findcounterexample},
whose proof uses methods from several areas of mathematics including algebraic geometry and Ramsey theory.

As an illustration
we now apply the Diagonalization Theorem to Karamata's inequality, see~$\S$\ref{sec:karamata} for the full details.
In the Karamata setting we are given $f_i, g_i \in \sharpP$, $1 \leq i \leq n$,
such that the following functions $h_i$, $1 \leq i < n$, are also all in $\sharpP$~:
$$
h_i \, := \, f_1 \ts +\ts \ldots \ts +\ts  f_i \. - \. g_1 \ts -\ts \ldots \ts -\ts  g_i,
$$
and we are also guaranteed that
\begin{equation}\label{eq:karamatalastINTRO}
f_1 \ts +\ts \ldots \ts +\ts f_n \. - \. g_1 \ts -\ts \ldots \ts -\ts  g_n \, = \, 0.
\end{equation}
This assumption is called
\defn{majorization}, see~$\S$\ref{sec:karamata}.
Moreover, the functions
\begin{equation}
\label{eq:karamatadeINTRO}
\textup{
$d_i \. := \. f_i \ts - \ts f_{i+1}$ \ \ and \ \ $e_i \. := \. g_i \ts - \ts g_{i+1}$
}
\end{equation}
are also in \ts $\sharpP$ \ts for all \ts $1 \leq i < n$.
Let $Z\subseteq \IQ^{5n-3}$ denote the variety of points that satisfy the
constraints~\eqref{eq:karamatalastINTRO} and \eqref{eq:karamatadeINTRO}, and let \ts
$S = Z\cap \IN^{5n-3}$.
Let $\gamma\in\GapPP$ be any monotone integer-valued function.
Define the Karamata function as
$$
K_{n,\gamma}(\vv f,\vv g) \ := \ \sum_{i=1}^{n} \gamma(f_i) \, - \, \sum_{i=1}^{n} \gamma(g_i).
$$
Clearly \ts $K_{n,\gamma}(\vv \sharpP_{\in S}) \subseteq \GapP$.
In fact, even \ts $K_{n,\gamma}(\vv \GapP) \subseteq \GapP$.
Karamata's inequality implies that the answer is always nonnegative for inputs from~$S$.
Hence, $K_{n,\gamma}(\vv{\sharpP}_{\in S})\subseteq\GapP_{\geq 0}$.
For which $\gamma$ do we have \ts
$K_{n,\gamma}(\vv{\sharpP}_{\in S})\subseteq\sharpP$?
\begin{itemize}
\item For affine linear $\gamma$ we clearly have $K_{n,\gamma}=0\in\sharpP$.
\item For $\gamma(t)=t^2$, we have $K_{2,\gamma}(f_1,f_2,g_1,g_2) = (d_1 + e_1 )h_1$ on $S$. This can be seen by plugging in $d_1 = f_1-f_2$, $e_1 = g_1-g_2$, and $g_2=f_1+f_2-g_1$. Clearly $(d_1 + e_1 )h_1 \in \sharpP$.
This has several proofs, for example instead of $(d_1 + e_1 )h_1$ we could have taken $2 h_1 + 2 e_1 h_1 + 4\binom{h_1}{2}$ with the same argument.
\item For $\gamma(t)=t^2$, we have $K_{3,\gamma}(f_1,f_2,f_3,g_1,g_2,g_3) = (d_1 +e_1) h_1 + (d_2+e_2) h_2 \in \sharpP$
on $S$.
\item For $\gamma(t)=\binom{t}{2}$, we have $K_{2,\gamma}(f_1,f_2,g_1,g_2) = (e_1+1) h_1 + 2\binom{h_1}{2} \in \sharpP$
on $S$.
\item For $\gamma(t)=\binom{t}{2}$, we observe that for the \emph{double} we have $2K_{3,\gamma}(\vv\sharpP_{\in S})\subseteq\sharpP$ via the observation that $2K_{3,\binom{t}{2}} = K_{3,t^2}$ on~$S$ (the affine linear parts cancel out).
\end{itemize}
All inclusions $K_{n,\gamma}\subseteq\sharpP$ in this section so far relativize.
The next proposition shows that the doubling we just used was in fact necessary, because otherwise we obtain an oracle separation.

\smallskip

\begin{proposition*}[see Proposition~\ref{pro:karamatadiagonalization}]
There exists \ts $A\subseteq\{0,1\}^*$ \ts such that \.
$K_{3,\binom{t}{2}} \big(\vv{\sharpP}^A_{\in S}\big) \not\subseteq \sharpP^A$.
\end{proposition*}

\begin{proof}[Proof sketch]
We use the Diagonalization Theorem~\ref{thm:diagonalization}.
We have $5n-3 = 12$, so $S = Z\cap \IN^{12}$.
We fix an arbitrary order of the 12 variables: \ts $\big(f_1,f_2,f_3,g_1,g_2,g_3,d_1,d_2,e_1,e_2,h_1,h_2\big)$.
The variety $Z$ is then given as the kernel of the linear map given by the following left matrix:
\begin{tikzpicture}
\node at (0,0) {
$\left(\begin{smallmatrix}
  1& -1&  0&  0&  0&  0& -1&  0&  0&  0&  0&  0\\
  0&  1& -1&  0&  0&  0&  0& -1&  0&  0&  0&  0\\
  0&  0&  0&  1& -1&  0&  0&  0& -1&  0&  0&  0\\
  0&  0&  0&  0&  1& -1&  0&  0&  0& -1&  0&  0\\
  1&  0&  0& -1&  0&  0&  0&  0&  0&  0& -1&  0\\
  1&  1&  0& -1& -1&  0&  0&  0&  0&  0&  0& -1\\
  1&  1&  1& -1& -1& -1&  0&  0&  0&  0&  0&  0
\end{smallmatrix}\right)$
};
\node at (5,0) {\resizebox{2cm}{0.3cm}{$\rightsquigarrow$}};
\node at (10,0) {
$\left(\begin{smallmatrix}
  1& -1&  0&  0&  0&  0& -1&  0&  0&  0&  0&  0\\
  0&  1& -1&  0&  0&  0&  0& -1&  0&  0&  0&  0\\
  0&  0&  0&  1& -1&  0&  0&  0& -1&  0&  0&  0\\
  0&  0&  0&  0&  1& -1&  0&  0&  0& -1&  0&  0\\
  1&  0&  0& -1&  0&  0&  0&  0&  0&  0& -1&  0\\
  1&  1&  0& -1& -1&  0&  0&  0&  0&  0&  0& -1\\
  1&  1&  1& -1& -1& -1&  0&  0&  0&  0&  0&  0
\end{smallmatrix}\right)$
};

\end{tikzpicture}
To obtain the necessary parametrization $\zeta$ of $Z$,
we convert it to row echelon form (the fact that the entries are integer is convenient, but not necessary for our techniques to apply), which is the right matrix.
We set $\ka=5$, $k=12$.
Permuting the order of the columns to $(6,9,10,11,12,1,2,3,4,5,7,8)$,
we obtain affine linear functions $\zeta_8,\ldots,\zeta_{12}$, each $\zeta_b$ depends linearly on the first 5 variables.
We set $\problem{Multiplicities} = \problem{OccurrenceMulti}_{12}$,
which is the function defined as follows: on input $w\in\{0,1\}^*$ we split $w$ into $12$ parts of roughly the same size, and the output is the vector that specifies how many 1s are in the first part, how many 1s are in the second part, and so on.
We set
\[
\varphi\big(f_1,f_2,f_3,g_1,g_2,g_3,d_1,d_2,e_1,e_2,h_1,h_2\big)\ := \ f_1^2+f_2^2+f_3^2-g_1^2-g_2^2-g_3^2\,.
\]
After verifying the technical assumptions,
to apply the theorem it remains to show that $\varphi+I$ is binomial-bad.
Since all $\zeta_b$ are affine linear, this claim boils down to checking that polyhedron \.
$\mathcal P_{\varphi,\zeta}$ \. does not have integer points.
We formalize and generalize this implication in the Polyhedron Theorem~\ref{thm:polyhedron}.
Here we have a polyhedron in $\IQ^{91}$ given by $21$ linear equations intersected with the nonnegative orthant.
We use a computer to set up the polyhedron. Indeed, it contains the half-integer point that shows that $2K_{3,\binom{t}{2}}(\vv\sharpP_{\in S})\subseteq\sharpP$,
but it \emph{does not} contain an integer point, which
gives $A$ such that $\varphi(p_A)\notin\sharpP^A$, but $\varphi(p_A) \in K_{3,\binom{t}{2}} (\vv{\sharpP}^A_{\in S})$.
\end{proof}

The Polyhedron Theorem \ref{thm:polyhedron} is also what we use for studying counting classes from $\TFNP$ (see \S\ref{sec:countingclassesandTFNPINTRO}), but the corresponding polyhedra there are very simple, see e.g.~Theorem~\ref{thm:decrementationseparation} and Theorem~\ref{thm:unbalancedflowseparation}.
The polyhedron for \problem{Sperner} for example has the integer point $(2,0)$ to represent that $\varphi = 2t_- + 0t_+$ on the variety.
The technical difficulty in those cases is not the polyhedron, but the existence of set-instantiators.
If we just study closure properties of $\sharpP$, then trivial set-instantiators can be used.

\subsection{Counting classes and {\normalfont\TFNP}}
\label{sec:countingclassesandTFNPINTRO}
In this section define the counting classes for which we claimed in \eqref{motivnon:smith} in $\S$\ref{ss:intro-basics-SP-not} that many of them coincide with $\sharpP$, while others are strictly stronger w.r.t.\ an oracle.
In order to do so, we attach oracles to the syntactic subclasses of $\TFNP$.\footnote{We consider $\CLS$, $\PLS$, $\PPAD$, $\PPADS$, $\PPA$, and $\PPP$ here, see e.g.~\cite{GP17}. The instances are exponentially large (di)graphs given succinctly by circuits or lists of circuits. For the sake of simplicity, we will assume in this discussion that finite lists of circuits are merged into a single circuit with additional input bits.\label{footnote:singlecircuit}}

Consider for example the relation \ts $\problem{rLonely}$.  This is for \ts $\PPA$, as the other classes are handled analogously.
Let \. $(C,x) \in \problem{rLonely}$ \. if and only if \.
$$x \neq 0 \ \wedge \ \big(C'(x)=x \ \vee \ C'(C'(x))\neq x\big),
$$
where $C$ is the description of a polynomially-sized multi-output Boolean circuit that describes the partner function on an exponentially large graph,
and $C'$ is the syntactic modification to $C$ that ensures that $C'(0)=0$.

Now, let \ts $\ComCla{r}\PPA$ \ts be the set of polynomially balanced relations $R$ for which a pair \ts $(\alpha,\beta)$ \ts of polytime computable maps exists with \ts $(C,\beta(x)) \in R$ \ts if and only if \. $(\alpha(C),x)\in \problem{rLonely}$.
These are the relations that correspond to search problems in $\PPA$.
Let $\ComCla{rP}$ denote the set of polynomially balanced relations that can be evaluated in polynomial time.
By definition, \ts $\ComCla{r}\PPA\subseteq\ComCla{rP}$.

For a language \ts $A \subseteq \{0,1\}^*$ \ts we define analogously \.
$(C,x) \in \problem{rLonely}^A$ \. if and only if \. \big($C(x)=x$ \. $\vee$ \. $C(C(x))\neq x$\big),
but now we allow the circuit $C$ to have arbitrary arity oracle gates that query the oracle $A$.
Let $\ComCla{r}\PPA^A$ be the set of polynomially balanced relations $R$ for which a pair $(\alpha,\beta)$ of polynomial-time computable maps exists with \. $(C,\beta(x)) \in R$ \. if and only if \. $(\alpha(C),x)\in \problem{rLonely}^A$. Note here that the only difference is that $\alpha(C)$ can have oracle gates.
Let $\ComCla{rP}^A$ denote the set of polynomially balanced relations that can be evaluated in polynomial time with access to~$A$.
By definition, we have \.
$\ComCla{r}\PPA^A\subseteq\ComCla{rP}^A$.\footnote{We use the \textsf{r}-prefix to avoid notational issues similar to the ones discussed in \cite{HV95, BS01}. We do not claim to have found a particularly good notation, but a suggestive one.}

We define the corresponding counting class \. $\#\PPA^A$ \. via
\[\textstyle
f \in \#\PPA^A \ \ \Longleftrightarrow \ \ \exists R\in\ComCla{r}\PPA^A \ : \ f(w) = \sum_{y\in\{0,1\}^*} R(w,y).
\]
Recall that

\vspace{-0.7cm}

\[\textstyle
f \in \sharpP^A \ \ \Longleftrightarrow \ \ \exists R\in\ComCla{rP}^A \ : \ f(w) = \sum_{y\in\{0,1\}^{*}} R(w,y).
\]
Hence, clearly \. $\#\PPA^A\subseteq\sharpP^A$, and in fact \. $\#\PPA^A\subseteq\sharpP^A_{\geq 1}$ \.
for all languages~$A$.\footnote{In fact, $\sharpP_{\geq 1}$ can be thought of as the counting analogue of $\TFNP$, i.e., it is reasonable to define \. $\#\TFNP:=\sharpP_{\geq 1}$.}\,\footnote{We choose this approach, which is different from the type-2 complexity approach in \cite{BCEIP,BM04}, because we
want to compare our counting classes to $\sharpP^A$.}

For the study of whether or not a problem is in $\sharpP$ we need the finer viewpoint that is obtained when insisting on $(\alpha,\beta)$ being a \emph{parsimonious reduction}, i.e.,
\big($(C,\beta(x))\in R$ and $(C,\beta(y))\in R$\big) implies $x=y$.
Since not all $\PPA$-complete problems are equivalent to each other via parsimonious reductions, this gives rise to different counting complexity classes,
depending on the $\PPA$-complete problem.
We write $\#\PPA\problemp{P}$ to indicate that we mean the counting class defined by problem $\problem{P}$ under parsimonious reductions.\footnote{Note that the prefix $\PPA$ is actually superfluous in this case, but we keep it for clarity.} For example, observe that all functions in $\#\PPA\problemp{Lonely}$ output odd integers on every input, while the class $\#\PPA$ contains more functions than that (as we will see when discussing \problem{Preleaf}).
In fact, for that reason it makes sense to study the class $(\#\PPA\problemp{Lonely}+1)/2$ and $(\#\PPA\problemp{Lonely}-1)/2$ and ask if they are subsets of $\sharpP$.

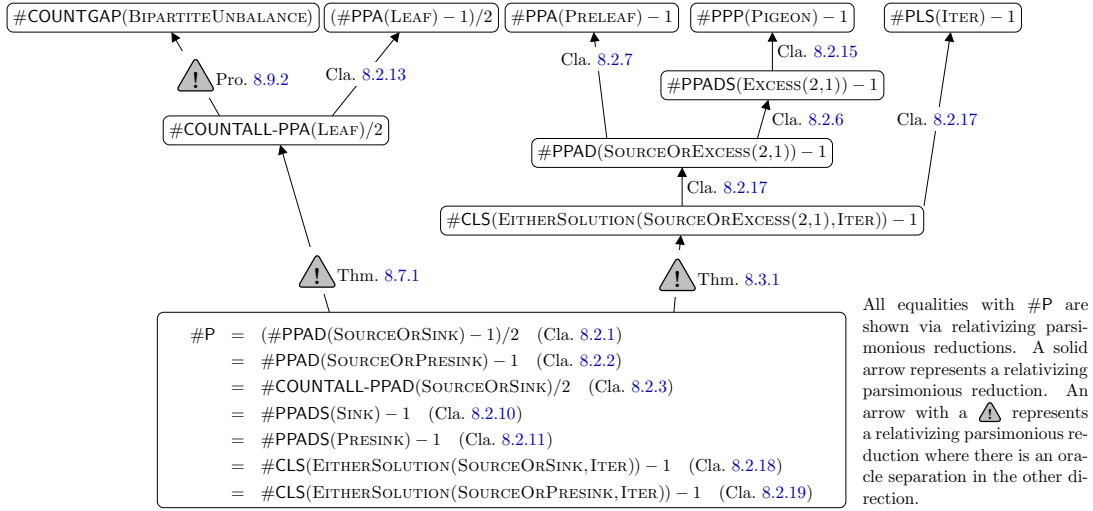
\begin{figure}
\quad\scalebox{0.6}{
\begin{tikzpicture}
\node[draw, rounded corners] (sharpP) at (0,-0.7) {
\begin{minipage}{15cm}
\vspace{-1em}\begin{eqnarray*}
\sharpP
&=& (\#\PPAD\problemp{SourceOrSink}-1)/2 \quad (\textup{Cla.~\ref{cla:sourceorsinkequalssharpP})}\\
&=& \#\PPAD\problemp{SourceOrPresink}-1  \quad (\textup{Cla.~\ref{cla:sourceorpresinkequalssharpP})}\\
&=& \sharpCOUNTALL[-PPAD]\problemp{SourceOrSink}/2 \quad (\textup{Cla.~\ref{cla:countallsourceorpresinkhalfequalssharpP}}) \\
&=& \#\PPADS\problemp{Sink}-1 \quad (\textup{Cla.~\ref{cla:sinkequalssharpP}})\\
&=& \#\PPADS\problemp{Presink}-1 \quad (\textup{Cla.~\ref{cla:presinkequalssharpP}})\\
&=& \#\CLS\problemp{EitherSolution(SourceOrSink,\ts{}Iter)}-1 \quad (\textup{Cla.~\ref{cla:CLSssiequalssharpP}}) \\
&=& \#\CLS\problemp{EitherSolution(SourceOrPresink,\ts{}Iter)}-1 \quad (\textup{Cla.~\ref{cla:CLSspiequalssharpP}})
\end{eqnarray*}
\end{minipage}
};
\node[draw, rounded corners] (CLS) at (4,3.5) {$\#\CLS\problemp{EitherSolution(SourceOrExcess(2,1),\ts{}Iter)}-1$};
\node[draw, rounded corners] (PLS) at (10,8) {$\#\PLS\problemp{Iter}-1$};
\node[draw, rounded corners] (PPAD) at (4,5) {$\#\PPAD\problemp{SourceOrExcess(2,1)}-1$};
\node[draw, rounded corners] (PPAPRELEAF) at (2,8) {$\#\PPA\problemp{Preleaf}-1$};
\node[draw, rounded corners] (PPAPRELEAFHALF) at (-2,8) {$(\#\PPA\problemp{Leaf}-1)/2$};
\node[draw, rounded corners] (PPADS) at (6,6.5) {$\#\PPADS\problemp{Excess(2,1)}-1$};
\node[draw, rounded corners] (PPP) at (6,8) {$\#\PPP\problemp{Pigeon}-1$};
\node[draw, rounded corners] (COUNTALLLEAFHALF) at (-5,5.5) {$\sharpCOUNTALL[-PPA]\problemp{Leaf}/2$};
\node[draw, rounded corners] (BIPARTITEUNBALANCE) at (-7.5,8) {$\sharpCOUNTGAP\problemp{BipartiteUnbalance}$};
\draw[-{Triangle[length=2mm, width=2mm]}]
($(sharpP.north east)!0.5!(sharpP.north)$)
-- (CLS) node [midway, fill=white] {\phantom{Thm.~\ref{thm:decrementationseparation}}{\Large$\warningsign$}Thm.~\ref{thm:decrementationseparation}};
\draw[-{Triangle[length=2mm, width=2mm]}] (CLS.north east) -- (PLS) node [midway, fill=white] {Cla.~\ref{cla:CLSinPLSandPPAD}};
\draw[-{Triangle[length=2mm, width=2mm]}] (CLS) -- (PPAD) node [midway, xshift=1cm] {Cla.~\ref{cla:CLSinPLSandPPAD}};
\draw[-{Triangle[length=2mm, width=2mm]}]
($(PPAD.north east)!0.5!(PPAD.north)$)
-- (PPADS) node [midway, xshift=1cm] {Cla.~\ref{cla:PPADinPPADS}};
\draw[-{Triangle[length=2mm, width=2mm]}] (PPADS) -- (PPP) node [midway, xshift=1cm] {Cla.~\ref{cla:PPADSinPPP}};
\draw[-{Triangle[length=2mm, width=2mm]}]
($(PPAD.north west)!0.5!(PPAD.north)$)
-- (PPAPRELEAF) node [near end, fill=white] {Cla.~\ref{cla:PPADinPPA}};
\draw[-{Triangle[length=2mm, width=2mm]}]
($(COUNTALLLEAFHALF.north east)!0.5!(COUNTALLLEAFHALF.north)$)
-- (PPAPRELEAFHALF) node [midway, fill=white] {Cla.~\ref{cla:PPAhalfinPPA}};
\draw[-{Triangle[length=2mm, width=2mm]}]
($(sharpP.north west)!0.5!(sharpP.north)$)
-- (COUNTALLLEAFHALF) node [near start, fill=white] {\phantom{Thm.~\ref{thm:halvingseparation}}{\Large$\warningsign$}Thm.~\ref{thm:halvingseparation}};
\draw[-{Triangle[length=2mm, width=2mm]}]
($(COUNTALLLEAFHALF.north west)!0.5!(COUNTALLLEAFHALF.north)$)
-- (BIPARTITEUNBALANCE) node [midway, fill=white] {\phantom{Pro.~\ref{pro:unbalancedflowseparation}}{\Large$\warningsign$}Pro.~\ref{pro:unbalancedflowseparation}};
\node at (10.5,-0.5){\begin{minipage}{5cm}
All equalities with $\sharpP$ are shown via relativizing parsimonious reductions. A solid arrow represents a relativizing parsimonious reduction. An arrow with a {\footnotesize {\protect$\protectwarningsign$}} represents a relativizing parsimonious reduction where there is an oracle separation in the other direction.
\end{minipage}};
\end{tikzpicture}}
\caption{The relativizing equalities and inclusions; and the oracle separations.}
 \label{fig:inclusions}
\end{figure}

Aside from the classical problems we study slightly adjusted problems that are not equivalent via parsimonious reductions (see the detailed list in~$\S$\ref{ss:TFNP-Background}).
Each class is defined via parsimonious reductions to a complete problem. The naming prefixes \. $\sharpCOUNTALL[-\PPA]$ \.
and \. $\sharpCOUNTGAP$ \. are essentially flavor.
By definition we have \.
$\#\PPA\problemp{P} \subseteq \#\PPA$ \. for all search problems \. $\problem{P}\in\PPA$, and analogously for all other search problems.
The relativizing inclusions and oracle separations that we find are depicted in Figure~\ref{fig:inclusions}.
All equalities with $\sharpP$ are shown via relativizing parsimonious reductions and they are proved in $\S$\ref{sec:syntacticcountingsubclasses}.
A solid arrow represents a relativizing parsimonious reduction. An arrow with a \scalebox{0.8}{$\warningsign$} represents a relativizing parsimonious reduction where there is an oracle separation in the other direction.
This means that all classes that are above $\sharpP$ in the figure strictly contain $\sharpP$ (with respect to an oracle).

We find a surprisingly large number of counting problem classes that, if adjusted properly with ``$-1$'' and ``$/2$'' are \emph{equal to $\sharpP$}.
This includes the canonical counting versions of \ts $\PPAD$, \ts $\PPADS$ \ts and \ts $\CLS$.
Only after making slight changes to the problems via non-parsimonious polytime equivalences (similar to the \emph{chessplayer algorithm}, see e.g.\ \cite{Pap90, BCEIP}), we obtain that the non-parsimonious counting classes strictly contain \ts $\sharpP$, which puts the new problems outside of~$\ts\sharpP$.

We identify two main ``reasons'' (i.e., oracle separations), one for ``$-1$'' (the decrementation separation, see~$\S$\ref{sec:decrementationseparation}), and one for ``$/2$'' (the halving separation, see $\S$\ref{sec:halvingseparation}).
There are two versions of $\#\PPA$ at the top of the diagram, one for each of the two reasons, and they are not easily comparable.\footnote{
It is not even clear if $\#\PPA\problemp{Leaf}-1$ is contained in $(\#\PPA\problemp{Leaf}-1)/2$.
One would want to just double a \problem{Leaf} instance, but that will end up creating 2 leaves too many. In other words, the class $\#\PPA\problemp{Leaf}-1$ seems to not be closed under the operation of scaling a function by 2.}
It is also noteworthy that we know of no counting version of \ts $\PPA$, \ts $\PLS$ \ts or \ts $\PPP$ \ts that equals~$\ts\sharpP$.\footnote{Note that since $\PTFNP$ (see \cite{GP18}) contains $\CLS$, the decrementation separation also shows $\#\PTFNP^A-1\not\subseteq\sharpP^A$.}

Since the polynomials $x-1$, $\frac x 2$, and $\frac {x-1}{2}$ are monotone, the tool to prove the separations is the Diagonalization Theorem~\ref{thm:diagonalization}.
The main difference from all separations so far is that now the instances are much more involved.
In the halving separation we have to ``hide'' the partner vertex from the $\sharpP$ machine, and in the decrementation separation we have to ``hide'' which of the solutions is connected to the zero vertex.
This is especially difficult for $\#\PLS\problemp{Iter}-1$ (and hence for $\#\CLS-1$).\footnote{Notably, $\PLS$ is also missing from the oracle separations in~\cite{BCEIP}.}
We formalize our approach in the definition of a set-instantiator in Definition~\ref{def:setinstantiator} and the necessary set-instantiators are created
in Section~\ref{sec:TFNP}.

Even though we are mainly interested in
membership and non-membership in $\sharpP$,
with only a little extra work our techniques directly
give
us another oracle separation
$$\sharpCOUNTALL[-PPA]\problemp{Leaf}^A/2 \ \subsetneq \ \sharpCOUNTGAP\problemp{BipartiteUnbalance}^A\..$$
This is because after doubling we have \.
$\sharpCOUNTALL[-PPA]\problemp{Leaf}^A \subseteq \sharpP^A$, while still
$$\sharpP^A \ \subsetneq \ 2\. \sharpCOUNTGAP\problemp{BipartiteUnbalance}^A
$$
(see Proposition~\ref{pro:unbalancedflowseparation}).

\subsection{Structure of the paper}
Section~\ref{sec:binomialbasistechnical} gives a high-level introduction of the proof ideas that lead to oracle separations and the binomial basis theorem, which is a first step towards our Diagonalization Theorem~\ref{thm:diagonalization}.
Section~\ref{sec:witnesstheorem} proves the Witness Theorem~\ref{thm:findcounterexample}, which is a generalization of the naive oracle separation proof approach to all nontrivial graph varieties. It is a key ingredient of the Diagonalization Theorem, which is introduced in
Section~\ref{sec:diagonalizationtheorem}, which uses the notation of set-instantiators, a formal way to treat not only the closure properties of $\sharpP$, but syntactic subclasses of $\TFNP$.

Section~\ref{sec:combinineq} uses the theory that we developed in the earlier chapters to handle the details of the Cauchy inequality, the Alexandrov--Fenchel inequality, the Hadamard inequality, Fermat's little theorem (in $\sharpP$), the Ahlswede--Daykin inequality, and the Karamata inequality.
Section~\ref{sec:TFNP} establishes all necessary set-instantiators and proves all oracle separations from Figure~\ref{fig:inclusions}.
We conclude with final remarks and open problems in Section~\ref{s:open}.
\medskip

\section{The binomial basis}
\label{sec:binomialbasistechnical}

\subsection{Oracle separations and the binomial basis: An informal high-level view}
\label{sec:proofideas}

\nin
$(1)$~{\bf Fooling polytime Turing machines with oracles.}
The initial idea is classical.
Consider the following problem \problem{ConnectedComponents}. A problem instance is similar to an instance to the $\PPA$-complete problem $\problem{Sink}$,
i.e., we have an exponentially large undirected graph whose nodes have degree at most two, but in the \problem{ConnectedComponents} problem we are not guaranteed a sink at node zero. The graph is given as a pair of two circuits $C_1$ and $C_2$; an edge from $x$ to $y$ is present in $G$ if and only if \. $\big($($C_1(x)=y$ \. or \. $C_2(x)=y$) \. and \. ($C_1(y)=x$ \. or \. $C_2(y)=x$)$\big)$. The function \problem{ConnectedComponents} counts the number of connected components in $G$, ignoring isolated nodes.

We can quickly see that if we replace the circuits with black boxes (i.e., oracles),
then this problem is not in $\sharpP$. Assume there exists a nondeterministic Turing machine~$M$
for which the number of accepting paths
$\#\acc_M(G)=\problem{ConnectedComponents}(G)$ for all $G$.
Let $G$ be an instance with 2 connected components of superpolynomial size each. Then $M$ has 2 accepting paths, but every path can only access a polynomial number of oracle positions. Therefore we can take two edges $\{u,v\}$ and $\{x,y\}$ of $G$ that lie in different connected components
and that are not queried by $M$,
we remove them, and replace them with the two edges $\{u,x\}$, $\{v,y\}$ to obtain a graph $G'$ with only 1 connected component. But since the two accepting paths of $M$ do not query the oracle at these positions, $M$ will have at least two accepting paths on the instance $G'$, which is a contradiction.

\smallskip
\nin
(2)~{\bf Combinatorial Diagonalization.}
This simple idea described above can be used in more sophisticated ways, see \cite[Thm~3.1.1]{C+89}.
We take the counting problem \problem{CircuitSat} in which we are given a Boolean circuit and the number of accepting paths is supposed to be the number of inputs for which the circuit outputs True.
Clearly $\problem{CircuitSat}\in\sharpP$.
We modify the problem by defining a function $\varphi:\IN\to\IN$ with $\varphi(0)=0$, $\varphi(1)=1$, $\varphi(2)=1$, and the other values are not relevant. Replace the circuit with an oracle, so instead of counting input strings for which the circuit outputs True, we now search for strings for which the oracle query returns~1. We now argue that \. $\varphi(\problem{CircuitSat}) \notin \sharpP$ \. with respect to that oracle (if the oracle is chosen correctly).

To see this, assume the existence of a nondeterministic Turing machine~$M$
with \. $\#\acc_M(C)=\varphi(\problem{CircuitSat}(C))$.
If a computation path of $M$ does not query any~1 at any oracle position, then it cannot accept. Indeed, if it does accept, then we can change the oracle to all zeros to obtain a contradiction.  If there is exactly one oracle position with a 1 (say, at position $a$), then $M$ must have exactly one accepting path and that path queries the oracle at~$a$.
Now, if we change the oracle by setting it to zero at~$a$ and to~1 at~$b$, where $b$ was not queried by the accepting path of $M$ before, then we get another accepting path. Let us call this accepting path~$\tau$.
What happens if we take the oracle that has a~1 at position~$a$ and a~1 at position $b$ at the same time? If $\tau$ does not query the oracle at position $a$, then \emph{both} accepting paths will accept on this instance, which is a contradiction, because $\varphi(2)=1<2$. The key point is that this happens \emph{with high probability} if the positions are chosen uniformly at random.

\smallskip
\nin
(3)~{\bf The binomial basis.}
The binomial basis theorem classifies completely the set of polynomials $\varphi$ for which this argument works.
If we pick a polynomial $\varphi$ in the discussion above, it could have been \.
$\varphi(f) = \frac{1}{6} f^3 - f^2 + \frac{11}{6} f$ \. for example.
This is a monotone function whose values are nonnegative integers, so we have no immediate way rule out membership in~$\sharpP$.
The key insight that relates this polynomial to what we observed is when we write it in its \emph{binomial basis}, the basis of the vector space of polynomials that is spanned by the binomial coefficients: \.
$\varphi(f) = f-\binom{f}{2}+\binom{f}{3}$.
The problematic issue is the negative coefficient $-1$ for $\binom{f}{2}$. This $-1$ is exactly the amount by which we overshoot with our number of accepting paths when the oracle has two 1s.
We call a polynomial binomial-good if all its coefficients in the binomial basis are nonnegative integers, otherwise it is binomial-bad.

This idea of univariate polynomial closure properties of $\sharpP$ generalizes to the \emph{multivariate case}, i.e., for \ts $\varphi$ \ts a multivariate polynomial. Here we again have a (multivariate) binomial basis, and say that a multivariate polynomial is binomial-good if all its coefficients in this basis are nonnegative integers.

\smallskip
\nin
(4)~{\bf The multivariate binomial basis.}
The multivariate case is of high interest when studying inequalities in combinatorics, but it is also already important when studying concrete instantiations such as $\#\PPAD-1$.
For a $\problem{SourceOrSink}$ instance (which is parsimoniously equivalent to \problem{Sperner}) we want to count the number number of nonzero sources $f$ plus the number of sinks~$g$ minus~1, i.e. $\varphi(f,g)=f+g-1$.
We know that in the $\problem{SourceOrSink}$ problem, for all instances we have $f-g+1=0$. This makes it possible to calculate $f+g-1 = 2f$, which implies $\#\PPAD\problemp{SourceOrSink}-1 \subseteq \sharpP$.
It is useful to think of the two functions $f+g-1$ and $2f$ as being the same function in the quotient ring $\IQ[f,g]/\langle f-g+1\rangle$. This ring is known as the \emph{coordinate ring} $\IQ[Z]$ of the affine variety $Z=\{(f,g)\mid f-g+1=0\}$.

We will employ this viewpoint in Section~\ref{sec:witnesstheorem} to prove the Witness Theorem~\ref{thm:findcounterexample},
which says that in many situations our notion of ``binomial-good on an affine variety'' is consistent with our intuition.
The main challenge in proving this theorem is that some nonnegative integer points on the variety $Z$ might actually \emph{not} correspond to a problem instance. For $\problem{SourceOrSink}$ they all do, but for example for $\problem{Iter}$ we have $Z=\IQ$ and there is no instance with 0 solutions;
or for $\problem{AllLeaves}$ we also have $Z=\IQ$, but there is no instance with an odd number of solutions.
The saving grace is that in all our cases the variety $Z$ is a graph variety and hence we can study asymptotic behavior via using a variant of \defng{Ramsey's theorem}.

The case $\problem{SourceOrSink}$ is a hyperplane, i.e., given by a single equation, but in general we have more equations, for example
for \. $\#\CLS\problemp{EitherSolution(SourceOrExcess(2,1),\ts{}Iter)}$, or
when treating Karamata's inequality in~$\S$\ref{sec:karamata}.
In the case where all constraints are affine linear (and $\varphi$ is arbitrary) we can rephrase the definition of binomial-goodnes of $\varphi$ on a graph variety given by functions $\zeta$ as follows: $\varphi$ is binomial-good on the graph variety if and only if the polyhedron $\mathcal P_{\varphi,\zeta}$ contains an integer point (see the Polyhedron Theorem~\ref{thm:polyhedron}). These integer points correspond directly to $\sharpP$ algorithms, for example for $\#\PPAD\problemp{SourceOrSink}-1$ there is an integer point $(2,0)$ to represent that $\varphi = 2f + 0g$ on the variety.

\smallskip

\nin
(5)~{\bf Set-instantiators.}
To make these ideas work for actual problem instances such as \problem{SourceOrSink} or \problem{Iter} instead of just for oracles for which we count the 1s, we introduce the notion of a set-instantiator and we construct the set-instantiators for all relevant cases. Intuitively, a set-instantiator mimics the behavior of the oracle idea presented above. For the existence of a set-instantiator one defines how to set up random instances of specified cardinalities that sit nicely in each other so that the Turing machine $M$ is fooled and
behaves in the same way as for a \problem{CircuitSat} instance. This becomes very challenging when the problem does not allow permuting the instance around, which is the case for \problem{Iter}, where the problem is always directed in a fixed direction. The \#\PLS\problemp{Iter} case is the most challenging (which carries over to $\#\CLS$), and we use a refined treatment.

\smallskip

\subsection{The binomial basis and integer-valued functions}
\label{subsec:binomialbasisandivf}
This section covers some basic properties about the binomial basis.
Fix~$k$.
Let \. $S \subseteq \IN^k$ \. and let \. $D := \{\vv b \in \IN^k \. \mid \.\exists \vv a \leqslant \vv b : \vv a \in S\}$ \.
denote the downwards closure of~$S$.
It is instructive to think of $D = \IN^k$, although we will need it in greater
generality in Section~\ref{sec:witnesstheorem}, where we work with functions
that are not defined on all of \ts $\IN^k$.

We are interested in functions \ts $D\to\IN$ \ts and functions \. $\IZ^k\to\IZ$, and also in function $S_\leqslant\to\IN$.
We focus on \ts $D\to\IN$ \ts first.
This set of functions lies in the $\IQ$-vector space of all functions \ts $D\to\IQ$.
For \ts $\vv b \in D$ \ts we define the binomial function \. $\beta_{\vv b} : D\to\IQ$ \. via
\[
\beta_{\vv b}(x_1,\ldots,x_k) \ := \ \binom{x_1}{b_1}\cdots \binom{x_k}{b_k} \ = \ \binom{\vv x}{\vv b}.
\]
The \defn{binomial basis} is defined as the set of all binomial functions \.
$\{\beta_{\vv b}\mid \vv b \in D\}$.
Recall that a function \. $D \to \IQ$ \. or \. $\IZ^k \to \IQ$ \. is \defn{integer-valued} \ts
if it only attains integer values.

\smallskip

\begin{proposition}[{\rm cf.~\cite{Nag19}}]\label{pro:basis}
Every function \. $\varphi: D\to\IQ$ \. can be expressed as a $($possibly infinite$)$ \ts $\IQ$-linear combination of elements from the binomial basis. This expression is unique.
The function $\varphi$ is integer-valued \ts if and only if \ts its coefficients in the binomial basis are all integer.
\end{proposition}

\smallskip

If we only consider polynomials $\varphi$, then this result is classical and
slightly easier to prove.  The following is a
variation on the original argument by Nagell~\cite{Nag19}
(see also~\cite[p.~13]{Nar95}).

\smallskip

\begin{proof}
We first show uniqueness. We assume for the sake of contradiction that a function has two distinct expressions in the binomial basis. Then the difference of them is a nontrivial expression of the zero function:
\begin{equation}\tag{$\dagger$}\label{eq:nontrivialexpressionofzero}
0 \,= \,\sum_{\vv b \in \IN^k} c_{\vv b} \beta_{\vv b}\,, \ \ \.  \text{where} \ \ c_I \in \IQ \ \ \text{and not all \. $c_{\vv b}$ \. are zero}.
\end{equation}
Let $m \in \IN$ denote the smallest number for which $c_{\vv b} \neq 0$, $|\vv b|=m$.
Every \ts $\beta_{\vv b}$ \ts has a unique monomial of largest total degree,
which is \. $x^{\vv b} := x_1^{b_1} \cdots x_k^{b_k}$ \. of degree~$|\vv b|$.
Moreover,
\begin{equation}\tag{$\ast$}\label{eq:crucialvanishingproperty}
\text{if} \ \ \beta_{\vv a}(x_1,\ldots,x_k)\neq 0\., \ \  \text{then} \ \ \vv x \. \geqslant \. \vv a.
\end{equation}
Therefore \eqref{eq:nontrivialexpressionofzero} induces a nontrivial linear combination of 0 in the homogeneous degree $m$ part:
\[
0 \, = \ \sum_{\vv b \in D, \. |\vv b|=m} \. \frac{c_{\vv b}}{b_1!\cdots b_k!} \, x^{\vv b}\,, \ \ \.  \text{where} \ \ c_{\vv b} \in \IQ \ \ \text{and not all \. $c_{\vv b}$ \. are zero}.
\]
Note that this is a finite linear combination.
But the monomials $x^{\vv b}$ are linearly independent functions on $D$, so all $c_{\vv b}=0$, which is a contradiction.
Hence the uniqueness is proved.

It remains to show that every function $\varphi:D\to\IQ$ can be expressed over the binomial basis (not necessarily with finite support).
The coefficients $c_{\vv b}$ for expressing $\varphi$ over the binomial basis are constructed using the following recursive property.
We set
\begin{equation}\label{eq:defcI}
c_{0,\ldots,0} \ := \ \varphi(0,\ldots,0) \quad\quad \text{ and } \quad\quad
c_{\vv b} \ := \ \varphi(\vv b) \, - \, \sum_{\vv a \leqslant \vv b, \. \vv a \neq \vv b} \. c_{\vv a}\ts \beta_{\vv a}(\vv b)
\end{equation}
This defines all \ts $c_{\vv b}$.
We have \. $\varphi = \sum_{\vv b \in \IN^k} c_{\vv b} \beta_{\vv b}$, because \ts $\beta_{\vv b}(\vv b)=1$, and
\eqref{eq:crucialvanishingproperty} implies that \. $\beta_{\vv a}(\vv b)=0$ \. for all other \ts $\vv a$ \ts over which is not summed.

Since the \ts $\beta_{\vv b}$ \ts are integer-valued functions we have that if the coefficients in the binomial basis are integers, then the function is integer-valued. By construction in the proof above (see how \ts $c_{\vv b}$ \ts is defined) this works in the other direction as well:
\ts if \ts $\varphi$ \ts is integer-valued, then the coefficients in the binomial basis are all integer.
\end{proof}

\smallskip

For \ts $D=\IN^k$,
we call polynomials $\varphi$ whose expression over the binomial basis has only nonnegative integer coefficients \defn{binomial-good}, all others are called \defn{binomial-bad}.
For \ts $D \subseteq \IN^k$ \ts we call \ts $\varphi$ \. \defn{$D$-good} \.  if the coefficients over the binomial basis are nonnegative integers;  otherwise \ts $\varphi$ \ts is \. \defn{$D$-bad}.

We remark that Proposition~\ref{pro:basis} generalizes from \. $D=\IN^k$ \.
to functions \. $\IZ^k \to \IQ$ \. as follows, with basically the same proof.
For \. $\vv b \in \IN^k$, we define the function \. $\tilde\beta_{\vv b} : \IZ^k\to\IQ$ \. via
\[
\tilde \beta_{\vv b}(x_1,\ldots,x_k) \ := \ \binom{x_1-\lfloor b_1\rfloor}{b_1}\cdots \binom{x_k-\lfloor b_k\rfloor}{b_k}.
\]
The \defn{shifted binomial basis} is defined as the set of all binomial functions \.
$\big\{\.\tilde\beta_{\vv b}\, :\, \vv b \in \IN^k\ts\big\}$.

\smallskip

\begin{proposition}\label{pro:shiftedbasis}
Every function \. $\varphi: \IZ^k\to\IQ$ \. can be expressed as a $($possibly infinite$)$ \ts $\IQ$-linear combination of elements from the shifted binomial basis. This expression is unique.
The function~$\ts\varphi$ \ts is integer-valued \ts if and only if \ts its coefficients in the shifted binomial basis are all integer.
\end{proposition}

For polynomials it does not matter which of the two bases we use, so we use the simpler (unshifted) one.
We remark that if \ts $\varphi$ \ts is a multivariate polynomial (nonnegative or not) of degree~$d$, then its expression over the binomial basis has finite support. This just follows from the fact that the set of all \ts $\beta_{\vv b}$ \ts with \ts $|\vv b|\leq d$ \ts is linearly independent, and hence (by counting the dimension) is a basis of the space of polynomials of degree \ts $\leq d$.

\smallskip

\subsection{The binomial basis theorem}
We start with the following general definition formalizing and generalizing
definitions in~$\S$\ref{ss:main-closure}.

\smallskip

\begin{definition}
A function $\varphi : \IN^k \to \IN$ is a \ts \defn{closure property of $\sharpP$} \. if the following holds: \ts if \. $f_1,\ldots,f_k \in \sharpP$, then the function \ts $\varphi(f_1,\ldots,f_k)$ \ts defined via
\[
\bigl[\varphi(f_1,\ldots,f_k)\bigr](x_1,\ldots,x_k) \ := \ \varphi\bigl(f_1(x_1),\ldots,f_k(x_k)\bigr)
\]
is in \ts $\sharpP$. In other words, \. $\varphi\big(\vv \sharpP\big) \subseteq \sharpP$.
We analogously define closure properties of \ts $\sharpP^A$ for any language~$A$ as follows.

We say that a closure property $\varphi$ of \ts $\sharpP$ \ts \defn{relativizes} \.
if for every language $A$ the function $\varphi$ is a closure property of \ts $\sharpP^A$.
A function \. $\varphi : \IZ^k \to \IZ$ \. is a closure property of \ts $\GapP$ if \.
$\varphi\big(\vv \GapP\big)\subseteq\GapP$.
We say that a closure property $\varphi$ of $\GapP$ \. \defn{relativizes} \.
if for every language~$A$ the function $\varphi$ is a closure property of \ts $\GapP^A$.
\end{definition}

\smallskip

\begin{theorem}[{\em\defn{Binomial basis theorem}}, see {\cite[Thm.~3.13]{HVW95}} and {\cite[Thm.~6]{Bei97}}, stated above in~$\S$\ref{ss:main-binomial}{}]\label{thm:closure}  \ \\
The following properties for a multivariate polynomial $\varphi$ are equivalent:
\begin{itemize}
\item $\varphi$ \ts is a relativizing closure property of \ts $\GapP$,
\item $\varphi$ \ts  is a closure property of \ts $\GapP$,
\item $\varphi$ \ts  is integer-valued,
\item the expression of \ts $\varphi$ \ts over the binomial basis has only integer coefficients.
\end{itemize}
\noindent Moreover, the following are equivalent:
\begin{itemize}
\item $\varphi$ \ts  is a relativizing closure property of \ts $\GapP_{\geq 0}$\ts,
\item $\varphi$ \ts  is a closure property of \ts $\GapP_{\geq 0}$\ts,
\item $\varphi$ \ts  is integer-valued and attains only nonnegative integers,
\item the expression of $\varphi$ over the binomial basis has integer coefficients and \ts $\varphi$ \ts attains only nonnegative integers if evaluated at integer points in the nonnegative cone.
\end{itemize}
\noindent Moreover, the following are equivalent:
\begin{itemize}
\item $\varphi$ \ts is a relativizing closure property of \ts $\sharpP$,
\item the expression of \ts $\varphi$ \ts over the binomial basis has only nonnegative integer coefficients.
\end{itemize}
\end{theorem}

\begin{proof}
If \ts  $\varphi$ \ts  is a closure property of \ts $\GapP$, then clearly  \ts $\varphi$ \ts  is integer-valued.
If $\varphi$ \ts  is integer-valued, then its expression over the binomial basis has finite support and all coefficients are integers by Proposition~\ref{pro:basis}. In \cite{fenner1994gap}, see closure property~5, they present the proof that \ts $\GapP^A$ \ts is closed under taking binomial coefficients, addition, and product.  Thus,
every integer-valued  \ts $\varphi$  \ts whose its expression over the binomial basis has finite support,
and all coefficients are integers is a relativizing closure property of \ts $\GapP$.
Clearly, every relativizing closure property of \ts $\GapP$  \ts is a closure property of \ts $\GapP$,
which shows the equivalence of the first four items.
The same argument (with the obvious minimal modifications) shows the equivalence of the second four items.

If the expression of  \ts $\varphi$  \ts over the binomial basis has finite support and all coefficients are nonnegative integers, then $\varphi$ is a relativizing polynomial closure property of  \ts $\sharpP$, because the  \ts $\sharpP$  \ts closure proofs for addition, multiplication, and taking binomial coefficients relativize.
It remains to show the converse.
This is more technical and is a simple application of the much more general Diagonalization Theorem~\ref{thm:diagonalization}, and hence we postpone the proof to
Theorem~\ref{thm:relatclosure}.
\end{proof}

\smallskip

\subsection{The binomial basis conjecture}
Note that if the polynomial closure properties of $\sharpP$ all relativize, then Theorem~\ref{thm:closure} gives a complete classification of all polynomial closure properties of $\sharpP$: Those $\varphi$ whose expression over the binomial basis only nonnegative integers as coefficients. We conjecture that this is indeed the correct classification:
\begin{conjecture}[{\em \defna{Binomial basis conjecture}}{}]\label{conj:binomialbasis}
{\em The polynomial closure properties of \ts $\sharpP$ \ts all relativize.}
\end{conjecture}
Note that Theorem~\ref{thm:closure} says that Conjecture~\ref{conj:binomialbasis} is true if we replace \ts $\sharpP$ \ts by \ts $\GapP$ or by $\GapP_{\geq 0}$.
Note also that Conjecture~\ref{conj:binomialbasis} is equivalent to saying that the polynomial closure properties of \ts $\sharpP$ \ts are exactly the binomial-good polynomials.
Sometimes it is sufficient to use the weaker univariate version of the binomial basis conjecture:

\begin{conjecture}[{\em \defna{Univariate binomial basis conjecture}}{}]\label{conj:binomialbasisuniv}
{\em The univariate polynomial closure properties of \ts $\sharpP$ \ts all relativize.}
\end{conjecture}

We do not know if these two conjectures are equivalent.
Even the univariate version of the binomial basis conjecture implies \ts $\P \neq \NP$ \ts and will therefore be very difficult to prove.

\begin{theorem}\label{thm:relatPNP}
Conjecture~\ref{conj:binomialbasisuniv} implies $\sharpP\neq\sharpP^{\NP}$ \ts $($and hence, in particular, \ts $\P\neq\NP)$.
\end{theorem}
\begin{proof}
This is similar to the idea in \cite[Thm~3.12]{OH93}.
We claim that if $\sharpP = \sharpP^{\NP}$, then $\binom{x-1}{2}=\frac 1 2 x^2 - \frac 3 2 x + 1= \binom{x}{0}- \binom{x}{1}+ \binom{x}{2}$ is a polynomial closure property of $\sharpP$. This is shown as follows.

Given a function \ts $f \in \sharpP$ \ts
we construct the following  \ts $\sharpP^{\NP}$ \ts machine: \ts call the $\NP$ oracle to see if there is at least one witness.
If not, accept.
Otherwise, count pairs of distinct witnesses that are both not the smallest witness (an \ts $\NP$ \ts oracle call is used to see if there is a smaller witness).
The resulting function $g$ is in \ts $\sharpP^{\NP}$.
By construction, if $f(w)=0$, then $g(w)=1$. Moreover, if \ts $f(w)>0$, then \ts $g(w)=\binom{f(w)-1}2$. Since \ts $\binom{0-1}{2}=\frac{(-1)(-2)}{2}=1$, it follows that \ts $g=\binom{f-1}{2}$.
\end{proof}

\medskip

\section{{\normalfont$\sharpP$} on affine varieties: The witness theorem}
\label{sec:witnesstheorem}
In this section we prove the Witness Theorem~\ref{thm:findcounterexample}, which is the crucial theorem for dealing with different sets~$S$, in particular if~$S$ is a subset of an affine algebraic variety, such as for example for \. $\#\PPAD\problemp{SourceOrSink}$,
or if~$S$ has holes such as for \. $\#\PPA\problemp{Leaf}$.

\subsection{Notation for the witness theorem}
Denote \. $\llbracket d \rrbracket := \{0,1,2,\ldots,d\}$, which is not to be confused with \. $[d]=\{1,2,\ldots,d\}$.
Fix~$k$.  As before, we write \. $\vv f := (f_1,\ldots,f_k)$ \. for vectors of length~$k$.
Let \. $\IO := \IN^k$ \. be the set of integer points in the \defn{nonnegative orthant}.
As before, for \. $\vv f, \ts \vv g \in \IO$, we write \. $\vv g \leqslant \vv f$ \. if and only if \. $g_a \leq f_a$ \.
for all \. $1 \leq a \leq k$.
For \. $\vv g \leqslant \vv f$, we write
$$
\binom{\vv f}{\vv g} \ := \ \prod_{a=1}^k \binom{f_a}{g_a}.
$$

Let \ts $S\subseteq \IO$, and 
let \ts $\varphi\in\IQ[\vv f]$.
We want to determine whether or not \ts $\varphi$ \ts is a relativizing multivariate closure property of \ts $\sharpP$ \ts under the guarantee that the cardinalities of the instances only come from the set~$S$.
We will see that under certain assumptions on the Zariski closure of~$S$, this only depends on \ts $\varphi$ \ts and the Zariski closure of~$S$.
In fact, in many cases this can be characterized to exactly be the case when \ts $\varphi+I$ \ts contains a binomial-good polynomial, where~$I$ is the vanishing ideal of the set~$S$.

Let \. $I=I(S) \. := \. \big\{\varphi \in \IQ[\vv f] \, : \, \varphi(S) = \{0\}\big\}$ \. be the vanishing ideal of~$S$.
Let \. $Z=\overline S^{\.\textup{Zar}}$ be the Zariski closure, i.e., \. $Z = \big\{\vv f \in \IQ^k \, \big| \, \forall \varphi \in I(S): \varphi(\vv f) = 0\big\}$.

\begin{remark}{\rm
The Zariski closure over $\IQ$ is the same as the Zariski closure over $\IC$ intersected with $\IQ^k$.\footnote{See,
e.g., the proof in \ts \href{https://math.stackexchange.com/questions/279243}{math.stackexchange.com/questions/279243}\ts.}
We will work entirely over $\IQ$ and all polynomials have coefficients from~$\IQ$.}
\end{remark}

\smallskip

\subsection{Functions that grow slowly}
\label{sec:growslowly}

We will need to make an asymptotic growth behavior analysis, hence we define the vector space of functions that grow slower than $\varphi$ on $S$:
\[
\IQ[Z]_{\leq(\varphi,S)} \
:= \
\big\{\.\psi + I \in \IQ[Z] \ \, \big| \, \ \exists \ts \alpha \in \IQ \ \, \forall \ts \vv g \in S: \. \big|[\psi+I](\vv g)\big| \. \leq \. \alpha \. \big|[\varphi+I](\vv g)\big|\.\big\},
\]
where \ts $|a|$ \ts denotes the absolute value of \ts $a\in \qqq$.
Note that this is a linear subspace of \ts $\IQ[Z]$.
For some applications we are particularly interested in the case where \. $\dim\IQ[Z]_{\leq(\varphi,S)}<\infty$, for a fixed~$S$ and for all~$\varphi$,
while for others we also have a fixed~$\varphi$.
Clearly, if~$S$ is finite, then \ts $\IQ[Z]$ \ts itself is finite dimensional, and hence \. $\dim\IQ[Z]_{\leq(\varphi,S)}<\infty$.

\begin{example}
We will mostly study the cases where \. $\dim\IQ[Z]_{\leq(\varphi,S)}<\infty$,
but the cases where \. $\dim\IQ[Z]_{\leq(\varphi,S)}=\infty$ \. are also interesting,
as the following examples show.
Fix \ts $k=2$. Consider \.
 $S = \{(x,y) \in \IO\mid y = 2^x\}$, \. $I=0$ \. and \. $Z = \IQ^2$.
For \ts $\varphi=y-x$, we have
the infinite dimensional linear subspace \.
$\langle1,x,x^2,x^3,\ldots\rangle \subseteq \IQ[Z]_{\leq(\varphi,S)}$.
In fact, \. $y-x = 1+\binom x 2+\binom x 3+\binom x 4+\ldots$\ts, so \ts $y-x$ \ts
has only nonnegative integers in its binomial basis expansion.
For \ts $\varphi=y-2x$, we have the same infinite dimensional linear subspace, but \ts $y-2x$ \ts is not monotone.
Indeed, for \ts $x=0$ \ts we have \ts $y-2x = 1$, while for \ts $x=1$ \ts we have \ts $y-2x = 0$.
Thus if \ts $y-2x$ \ts is a closure property of \ts $\sharpP$ \ts on~$S$, then \ts $\UP=\coUP$ \ts
by Proposition~\ref{p:non-monotone}.
\end{example}

We write \ts $\textup{suppb}$ \ts for the support in the binomial basis.
The next claim shows that in many situations removing a finite amount of points from $S$ does not change whether the dimension of $\dim \IQ[Z]_{\leq(\varphi,S)}$ is finite or not.

\begin{claim}
If \ts $S'\subseteq S$ \ts with \ts $S\sm S' < \infty$
\ts and \ts $\overline S = \overline{S'} = Z$,
then for all \ts $\varphi\in\IQ[Z]$ \ts we have \.
$\dim \IQ[Z]_{\leq(\varphi,S)} < \infty$ \.
if and only if \.
$\dim \IQ[Z]_{\leq(\varphi,S')} < \infty$.
\end{claim}

\begin{proof}
Clearly \. $\dim \IQ[Z]_{\leq(\varphi,S)} \leq \dim \IQ[Z]_{\leq(\varphi,S')}$, which proves one direction.
It suffices to treat the case \ts $|S\sm S'|=1$. Let \ts $S\sm S' = \{\vv e\}$.
We make a case distinction based on whether or not \ts $\varphi(\vv e) = 0$.

If \ts $\varphi(\vv e) \neq 0$, then consider any arbitrary \. $\psi+I \in \IQ[Z]_{\leq(\varphi,S')}$.
We will see that \. $\psi+I \in \IQ[Z]_{\leq(\varphi,S)}$.
Let $\alpha$ be such that \. $\forall \ts \vv g \in S': \. \big|[\psi+I](\vv g)\big| \. \leq \. \alpha \. \big|[\varphi+I](\vv g)\big|$.
Similarly, let \. $\alpha' \ts := \ts \max\big\{\alpha,|\psi(\vv e)|/|\varphi(\vv e)|\big\}$.
It follows that \. $\forall \ts \vv g \in S: \. \big|[\psi+I](\vv g)\big| \. \leq \. \alpha' \. \big|[\varphi+I](\vv g)\big|$, hence \. $\psi+I\in\IQ[Z]_{\leq(\varphi,S)}$.

If \ts $\varphi(\vv e) = 0$, then we have \.
$\IQ[Z]_{\leq(\varphi,S)} = \big\{\psi \in \IQ[Z]_{\leq(\varphi,S')} \mid \psi(\vv e) = 0\big\}$. Since $\psi(\vv e) = 0$ is a homogeneous linear constraint, the dimension when imposing the constraint either stays the same or goes down by~1.
\end{proof}

\smallskip

Let \. $\xi_{\vv e} \ts := \ts \binom{\vv f}{\vv e} \ts \in \ts\IQ\big[\ts\vv f\ts\big]$ \. denote the binomial basis function to the exponent vector $\vv e \in \IO$.
A function \ts  $\xi_{\vv e}$ \ts is called \defn{small} \. if
\[
\forall \. \vv f \ts \in \ts S \, : \, \xi_{\vv e}(\vv f) \, \leqslant \, \varphi(\vv f).
\]
Otherwise, we call \ts $\xi_{\vv e}$ \ts \defn{large}.
Since functions in $I$ vanish on $S$, we see that \ts $\xi_{\vv e}$ \ts
is small if and only if every element of \ts $\xi_{\vv e}+I$ \ts is small.
In this case we say that \ts $\xi_{\vv e}+I$ \ts is a small element of \ts $\IQ[Z]$.

\begin{claim}\label{cla:existsxi}
Fix \. $\varphi \in \IQ[Z]$.
Suppose \. $\psi\in\IQ[\vv f]$ \. is binomial-good
and \.
$\psi + I \notin \IQ[Z]_{\leq(\varphi,S)}$\ts.
Then there exists \ts $\vv g \in S$ \ts and \.
$\vv e\in\textup{suppb}(\psi)$ \. with \. $\xi_{\vv e}(\vv g) > \varphi(\vv g)$.
\end{claim}

\begin{proof}
Decompose $\psi$ over the binomial basis
obtaining nonnegative integer coefficients. Set $\alpha$ to the sum of these coefficients.
By definition, there exists \ts $\vv g \in S$ \ts with \. $\psi(\vv g)>\alpha\ts\varphi(\vv g)$.
By the pigeonhole principle there exists \ts $\vv e$ \ts with \. $\xi_{\vv e}(\vv g)>\varphi(\vv g)$.
\end{proof}

\smallskip

\subsection{Graph varieties}
\label{sec:graphvarieties}
In this section we study an important class of varieties for which \.
$\dim\IQ[Z]_{\leq(\varphi,S)}<\infty$.
We write \. $\vv v \in \IQ^\ka$.
For \. $1 \leq b \leq \ka$, we let \ts $f_b := v_b$.
Given functions \. $\zeta_b \in \IQ[\vv v]$ \. and \.
$\ka+1 \leq b \leq k$,
we define \.
$
\eq_b := \zeta_b(\vv v) - f_{b} \in \IQ[\vv f]$.
We call each \ts $f_b$ \ts with \. $1 \leq b \leq \ka$, a \defn{parameter} \ts
and each \ts $f_b$ \ts  with \. $\ka+1 \leq b \leq k$ \.
a \defn{non-parameter}.
Note that for parameters we have \ts $f_b=v_b$ \ts (and we use these symbols interchangeably),
while for non-parameters there does not exist a \ts $v_b$.
Let
\[
Z \ := \ \big\{\ts \vv f \in \IQ^k \,\. \big| \,\. \eq_{\ka+1}(\vv f)\ts=\ts\ldots\ts=\.\eq_{k}(\vv f)\ts=\ts 0\ts\big\}
\]

Let \. $\tau:\IQ^\ka\to Z$ \. be defined via \. $\tau(v_1,\ldots,v_\ka) \ts :=\ts \big(v_1,\ldots,v_\ka,\zeta_{\ka+1}(\vv v),\ldots,\zeta_{k}(\vv v)\big)$.
Clearly, map~$\ts\tau$ \ts is bijective, where the inverse function
is the projection to the first \ts $\ka$ \ts coordinates.
Let $\tau^*$ be the pullback algebra homomorphism \. $\tau^*:\IQ[\vv f]\to\IQ[\vv v]$ \.
defined as \. $\tau^*(\varphi):=\varphi\circ\tau$,
that replaces each occurrence of each non-parameter \ts $f_b$ \ts
by the polynomial \ts $\zeta_b(\vv v)$ \ts in the parameters.
Let \. $I\subseteq \IQ[\vv f]$ \. be the ideal generated by \. $\eq_{\ka+1},\ldots,\eq_{k}$\..
We start by establishing some well-known basic facts about graph varieties.

\begin{claim}\label{cla:Iisker}
We have: \.
$\tau^*(\eta)\in\eta+I$ \. and \. $I = \ker \.\tau^*$.
\end{claim}

\begin{proof}
First, observe that \. $I\subseteq \ker\tau^*$.  Indeed, the kernel of every ring homomorphism
is always an ideal, so we only have to check this for the generators of~$I$. But clearly \.
$\eq_b \in\ker\.\tau^*$.

Now let \. $\eta \in\ker\tau^*$.
Consider \ts $\eta$ \ts as a polynomial in the non-parameters, i.e., as an element of the ring \.
$\big(\IQ[\vv v]\big)\big[f_{\ka+1},\ldots,f_{k}\big]$.
Since \ts $\tau^*$ \ts is replacing non-parameters by polynomials in parameters only, we can
consider the application of \ts $\tau^*$ \ts as a process
of iteratively substituting (in any order)
a single variable in a monomial of~$\eta$.
This decreases this monomial's degree, so this process terminates.

Let \. $\tau^*_\io$ \. be step \. $\io$ \. in this process of evaluating \ts $\tau^*$.
If at step \ts $\io$ \ts the variable \ts $f_b$, \ts $b > \ka$, occurs in the
monomial~$q$ with a strictly positive power, then we can write
\[
\textstyle
\tau^*_\io(q) \ = \ \tau^*_\io\big(\frac{q}{f_b} \ts f_b\big) \ =  \ \frac{q}{f_b} \. \zeta_b(\vv v)\..
\]
But we also see that
\[
q \. + \. \underbrace{\tfrac{q}{f_b}\big(\zeta_b(\vv v)\ts - \ts f_b\big)}_{\in I} \ = \ \tfrac{q}{f_b} \.\zeta_b(\vv v)\..
\]
This implies that \. $\tau^*_\io(q)\in q+I$.

Iterating this process, in every step only elements of $I$ are added. Hence $\tau^*(\eta)\in\eta+I$.
Since we assumed that $\eta\in\ker \tau^*$ it follows that $0 \in\eta+I$ and hence $\eta\in I$.
\end{proof}

\smallskip

\begin{claim}\label{cla:IequalsIZ}
$I\ts =\ts \mathcal{I}(Z)$.
\end{claim}

\begin{proof}
Let $\eta$ be a multivariate polynomial in $k$ variables.
The direction \. $I\subseteq \mathcal I(Z)$ \. is clear: \ts
if \ts $\eta\in I$, then $\eta$ vanishes on $Z$, because all $\eq_b$ vanish on~$Z$.

Now let $\eta$ vanish on $Z$. We show that $\eta\in I$.
Set \. $\eta' := \tau^*(\eta)$ \. and observe that \ts $\eta' \in \eta + I$ \ts by \textup{Claim}~\ref{cla:Iisker}.
Since $\eta$ vanishes on $Z$ and since $I\subseteq \mathcal I(Z)$, we have that $\eta'$
vanishes on $Z$. But since $\eta'$ only uses parameter variables,
$\eta'$ vanishes identically on $\IQ^\ka$.
By multivariate interpolation on $\IQ^\ka$ we have \ts $\eta'=0$, hence \ts $0 = \eta' \in \eta+I$, and therefore \ts $\eta\in I$.
\end{proof}

In other words,
we get the short exact sequence
\[
0 \ \to \ I \ \to \ \IQ[\vv f] \ \stackrel{\tau^*}{\to} \ \IQ[\vv v] \ \to \ 0.
\]
Since \. $\IQ[Z]=\IQ[\vv f]/I$, it follows \.
$\IQ[Z]\simeq \IQ[\vv v]$ \. with the isomorphism induced by \ts $\tau^*$,
and its inverse given by \ts $\psi\mapsto\psi+I$.

\begin{claim}
Let $S$ lie Zariski-dense in $Z$, and let \ts $S':=\tau^{-1}(S)$.  Then \ts
$S'$ \ts lies Zariski-dense in~$\ts\IQ^\ka$.
\end{claim}

\begin{proof}
Assume that a function $\psi\in\IQ[\vv v]$ vanishes on $S'$. Then $\psi$ also vanishes on $S$. This means that $\psi \in I$. But $\psi$ only uses parameter variables, hence $\psi=0$.
\end{proof}

We now study asymptotic growth behavior in the parameter space.
For $S'=\tau^{-1}(S)$ and $\varphi'=\tau^*(\varphi)$
we define
\[
\IQ[\vv v]_{\leq(\varphi',S')} \
:= \
\big\{\.\psi \in \IQ[\vv v] \ \big| \ \exists \ts \alpha \in \IQ \ \forall \ts \vv g \in S' \, : \, |\psi(\vv g)| \. \leq \. \alpha\ts |\varphi'(\vv g)|\.\big\}.
\]

\smallskip

\begin{claim}\label{cla:growthonparameterspace}
Let \. $\varphi'=\tau^*(\varphi)$ \. and let \. $S'=\tau^{-1}(S)$.
The isomorphism \. $\IQ[Z]\ts \simeq \ts \IQ[\vv v]$ \.
induces an isomorphism of linear subspaces \.
$\IQ[Z]_{\leq(\varphi,S)} \ts \simeq \ts \IQ[\vv v]_{\leq(\varphi',S')}$\..
\end{claim}

\begin{proof}
Let \. $\psi+I \in \IQ[Z]_{\leq(\varphi,S)}$.
Let \ts $\alpha\in \qqq$ \ts be such that for all \. $\vv g \in S$ \. we have:
$$
\big|[\psi+I](\vv g)\big| \ \leq \ \alpha\ts\big|[\varphi+I](\vv g)\big|\ts.
$$
Note that
$$
[\psi+I](\vv g) \, = \, \psi(\vv g)  \, = \,  \psi\big(\tau\big(\tau^{-1}(\vv g)\big)\big)
 \, = \,  \big[\tau^*(\psi)\big]\big(\tau^{-1}(\vv g)\big).
$$
Therefore, for all \. $\vv g \in S$, we have
$$
\big|\big[\tau^*(\psi)\big]\big(\tau^{-1}(\vv g)\big)\big| \ \leq \ \alpha\.\big|\big[\tau^*(\varphi)\big]\big(\tau^{-1}(\vv g)\big)\big|\ts.
$$
In other words,
for all \. $\vv g \in S$, we have
$$\big|[\psi']\big(\tau^{-1}(\vv g)\big)\big| \ \leq \ \alpha\.\big|[\varphi']\big(\tau^{-1}(\vv g)\big)\big|\ts,
$$
Since \. $S' = \tau^{-1}(S)$,
it follows that \.
$|\psi'(\vv g)| \ts \leq \ts \alpha\ts |\varphi'(\vv g)|$ \. for all \. $\vv g \in S'$.
The argument also works in the other direction, because \ts $\tau(S')=S$.
\end{proof}

\smallskip

Let \. $\IOp :=\IQ_+^\ka$. Note that we do not require integrality of the points in~$\IOp$.
We define the cone \ts $C_{S'}$ \ts of rays in \ts $\IOp$ \ts with infinitely many elements of~$S'$, i.e.\
\. $C_{S'} \ts := \ts \big\{\vv v \in \IOp \,\big|\, |S' \cap \IQ\vv v| = \infty\big\}$.
Even though \. $S'\subseteq\IQ^\ka$ \. is Zariski-dense, it can be that
\ts $C_{S'}=\emp$, for example if \ts $\ka=2$ \ts and \.
$S'=\{\vv v \in \IOp \mid v_2\geq v_1^2\}$.
Nevertheless, in many important cases, we have  \ts $C_{S'}$ \ts
is Zariski-dense in \ts $\IQ^\ka$.

\smallskip

\begin{claim}\label{cla:conedense}
If \. $C_{S'}$ \. is Zariski-dense in \ts $\IQ^\ka$,
then \. $\dim\IQ[Z]_{\leq(\varphi,S)}<\infty$.
\end{claim}

\begin{proof}
Let \. $r_{\varphi'} \ts := \ts \max\big\{\ts |\vv e| \. : \. \vv e \in \suppb(\varphi')\ts \big\}$.
Consider \. $\psi' \in \IQ[\vv v]$ \. with coefficients \ts $c_{\vv e}$.
Define \. $E \ts := \ts\big\{\ts\vv f \in \suppb(\psi') \, \big| \, r_{\varphi'} \ts <\ts |\vv f|\ts\big\}$,
and suppose \. $E\neq\emp$.

Let \. $r_E := \max\big\{|\vv e| \,: \,\vv e \in E\big\}$. By assumption, we have \ts $r_E > r_{\varphi'}$\ts.
Denote \. $E_{\max} := \bigl\{\.\vv e~\in~E \,:\. r_E = |\vv e|\.\bigr\}.$
Consider a polynomial
$$
p(v_1,\ldots,v_\ka) \ := \sum_{\vv e\in E_{\max}} \. c_{\vv e}\. v_{1}^{e_1}\.\cdots\. v_{\ka}^{e_\ka}\,,
$$
and define the hypersurface \. $H = \mathcal V(p)\subseteq \IQ^\ka$.
By definition, all \. $c_{\vv e}$ \. are nonzero and there is at least one summand,
so \ts $H$ \ts is indeed a hypersurface. Since \ts $C_{S'}$ \ts is Zariski-dense in \ts $\IQ^\ka$,
we conclude that \ts $C_{S'}$ \ts does not lie in any hypersurface.  Hence,
it follows that there exists \. $\vv v_H \in C_{S'}\sm H$.

Observe that
\begin{eqnarray*}
|\psi'(t \vv v_\varepsilon)|
&=&
\left|\sum_{\vv e\in E_{\max}}c_{\vv e}
v_{H,1}^{e_1}\cdots v_{H,\ka}^{e_\ka}\right|
\, t^{|\vv e|} \ \, + \ \, \text{lower order terms in $t$}
\\
&=&
\underbrace{\left|\sum_{\vv e\in E_{\max}}
c_{\vv e}v_{H,1}^{e_1}\cdots v_{H,\ka}^{e_\ka}
\right|}_{\neq 0\text{, because $\vv v_H \notin H$}} \, t^{r_E}
\ \, + \ \,  o(t^{r_E}),
\end{eqnarray*}
where the lower order terms are bounded by
\[
\big|\varphi'(t \vv v_H)\big| \ = \ O\big(t^{|\vv f_{\varphi'}|}\big) \ = \ o(t^{r_E}).
\]
Since there are infinitely many elements of $S'$ along this ray, we have \.
$\psi' \notin \IQ[\vv v]_{\leq(\varphi',S')}$.
Hence, a necessary criterion for \.
$\psi' \in \IQ[\vv v]_{\leq(\varphi',S')}$ \.
is \,\ts $\suppb(\psi')\subseteq \big\{\ts\vv e \in \IOp \,:\, r_{\varphi'} \geq |\vv e|\ts\big\}$, which is a finite set.
Hence \. $\dim \IQ[\vv v]_{\leq(\varphi',S')}<\infty$. The proof is finished using
Claim~\ref{cla:growthonparameterspace}.
\end{proof}

\smallskip

Recall that \. $Z := \big\{\vv f \in \IQ^k \,:\, \eq_{\ka+1}(\vv f)=\ldots=\eq_{k}(\vv f)=0\big\}$, where \.
$\eq_b := \zeta_b(\vv v) - f_{b}$.

\smallskip

\begin{claim}\label{cla:homparts}
We assume that \ts $Z$ \ts contains at least one integer point.
Let \ts $S := \IO\cap Z$.
Furthermore, we assume that for all \. $\ka+1\leq b \leq k$, the polynomial \ts $\zeta_b$ \ts is not constant.
Let \. $\zeta_b^{\textup{hom}}\in\IQ[\vv v]$ \. be the top nonzero homogeneous part
of \. $\zeta_b$, for all \. $\ka+1\leq b \leq k$.
Suppose there exists a point $\vv v \in \IOp$ that satisfies the strict inequalities
\begin{equation}\label{eq:openeqns}
\tag{$\ast$}
\zeta_b^{\textup{hom}}(\vv v) \ts > \ts 0 \quad \text{for all} \quad \ka+1\ts \leq \ts b \ts\leq k\ts.
\end{equation}
Then \ts $C_{S'}$ \ts lies Zariski-dense in \ts $\IQ^\ka$ $\big($and hence \. $\dim\IQ[Z]_{\leq(\varphi,S)}<\infty$ \. by Claim~\ref{cla:conedense}$\big)$.
\end{claim}

We remark that if all \ts $\zeta_b$ \ts have an integer as their constant coefficient (for example, if all \ts $\zeta_b$ \ts are homogeneous degree \ts $\geq 1$ \ts polynomials), then an integer point in \ts $Z$ \ts is obtained by setting all parameters to zero.

\begin{proof}[Proof of Claim~\ref{cla:homparts}]
Let \ts $\eta$ \ts be the least common multiple of the denominators of the coefficients of all~$\ts\zeta_b$.
Consider the stretched polynomial \. $\eta\ts \zeta_b$, which has integer coefficients.
For all \. $\vv v \in \IZ^\ka$, we have \. $\eta\ts\zeta_b(\vv v) \equiv 0 \pmod{\eta}$ \. if and only if \. $\zeta_b(\vv v)\in\IZ$.
Hence, we can work in the ring \. $\IZ/\eta\ts\IZ$.  Then, for every integer direction vector \.
$\vv z \in \IZ^\ka$ \. we have \. $\zeta_b(\vv v)\in\IZ$ \. if and only if \. $\zeta_b(\vv v+\eta\vv z)\in\IZ$.

By assumption, $Z$ \ts contains an integer point \ts $\vv f$ \ts for all \.
$\vv y = \tau^{-1}(\vv f)$ \. and for all direction vectors \. $\vv z \in \IZ^\ka$.  Therefore,
\. $\zeta_b(\vv y+\eta\vv z)\in\IZ$ \.
for all \. $\ka+1\leq b\leq k$.
Here we think of the set \. $\vv y  + \eta\ts \IZ^\ka \ts = \ts
\big\{\ts\vv y+\eta\vv z \. : \. \vv z \in \IZ^\ka\ts\big\}$, where \. $\vv y \in \tau^{-1}(Z)$,
as an affine shift of the axis parallel orthogonal grid with side length \ts
$\eta$ \ts and offset vector \ts $\vv y$.  Let us prove that
\begin{equation}
\label{eq:billard}
\tag{$\dagger$}
\aligned
& \text{if} \ \
\zeta_b(\vv v)\in\IZ \ \ \text{for all} \ \ \ka+1\leq b\leq k, \ \, \text{then for infinitely many} \ \. s\in\IN\ts, \\
& \text{we have} \ \
\zeta_b(s\vv v)\in\IZ \ \ \text{for all} \ \ \ka+1\leq b\leq k\ts.
\endaligned
\end{equation}
The proof of~$(\dagger)$ is similar to those when studying periodic orbits on a
$\ka$-dimensional rectangular block billiard table torus with side lengths \ts $\eta$ (see e.g.~\cite{Roz19,Tab05}),
and goes as follows.
We have \. $\zeta_b(\vv v)\in\IZ$, if and only if \. $\eta\ts\zeta_b(\vv v) \equiv 0 \pmod{\eta}$.
We work over \. $\IZ/\eta\IZ$.
Since the sequence \. $s\vv v \in (\IZ/\eta\IZ)^\ka$ \. is periodic with period length at most
\ts $\eta^\ka$, it follows that the sequence of vectors \.
$\bigl(\eta\ts\zeta_b(s\vv v) \. \mod \. \eta, \ \ka+1 \leq b \leq k\big) \. \in \. (\IZ/\eta\IZ)^{k-\ka}$ \.
is periodic of length at most $\eta^\ka$. This proves~$(\dagger)$.

\smallskip

Let \. $\overline{\divideontimes}^{\.\IQ}$ \. denote the Euclidean closure,
and \. $\overline{\divideontimes}^{\.\textup{Zar}}$ \. the Zariski closure.
Let \. $B_\varepsilon = \big\{\vv u \.:\. |\vv u|<\varepsilon\big\}$ \. be the open ball with radius~$\varepsilon>0$.
Since the strict inequalities \eqref{eq:openeqns} are open conditions, there exists $\varepsilon\in\IQ_{>0}$ such that each element in $\vv v + B_\varepsilon$ satisfies the strict inequalities~\eqref{eq:openeqns}.
We now prove
\begin{equation}\label{eq:inconeclosure}
\tag{$\ddagger$}
(\vv v + B_\varepsilon)\cap\IOp \subseteq \overline{C_{S'}}^{\.\IQ}
\end{equation}
This implies \. $(\vv v + B_\varepsilon)\cap\IOp\subseteq \overline{C_{S'}}^{\.\textup{Zar}}$.
Since \ts $B_\varepsilon$ \ts is full-dimensional and since each of its affine shifts with center in $\IOp$ has a full-dimensional orthant in $\IOp$, it follows that \.
$\overline{(\vv v + B_\varepsilon)\cap\IOp}^{\.\textup{Zar}} = \IQ^\ka$.
Hence \.
$\IQ^\ka = \overline{(\vv v + B_\varepsilon)\cap\IOp}^{\.\textup{Zar}}
\subseteq \overline{C_{S'}}^{\.\textup{Zar}} \subseteq \IQ^\ka$.
This implies that the subset relationships are actually all equalities, so in particular \. $\overline{C_{S'}}^{\.\textup{Zar}} = \IQ^\ka$.

It remains to show~\eqref{eq:inconeclosure}.
Let \. $\vv w \in (\vv v + B_\varepsilon)\cap\IOp$ \. be arbitrary, in particular,
suppose \ts $\vv w$ \ts satisfies~\eqref{eq:openeqns}.
For all \. $e\in\IQ_{>0}$ \. we now construct \. $\vv w'' \in C_{S'}$ \. with distance to $\vv w$ at most~$e$.
This proves~\eqref{eq:inconeclosure}.

Let \. $B_b := \big\{\ts\vv u \in \IQ^\ka \.:\. \zeta_b^{\textup{hom}}(\vv u) = 0\ts\big\}.$
Each \ts $B_b$ \ts is a closed set, so if a point is not in \ts $B_b$, then it has a strictly positive distance to~$\ts B_b$.
Let \ts $\Delta(\vv w)$ \ts be the minimum of the distances from \ts $\vv w$ \ts to~$\ts B_b$.
By the intercept theorem we have \. $\Delta(t\ts \vv w) = t \ts \Delta(\vv w)$.

Since \ts $\vv w$ \ts satisfies \eqref{eq:openeqns},
we can find $t$ large enough such that
\. $\Delta(t\vv w) > \eta^k$ \. and  \. $\eta^k/t < e$.
Let \ts $\vv w'$ \ts be a point in \. $\vv y+\eta\cdot\IZ^\ka$ \. that is closest to~$\vv w$.
By \eqref{eq:billard}, there are infinitely many nonnegative integer scalars \ts
$s \in \IN$ \ts such that
\[
\zeta_b(s\vv w') \ts \in \ts \IZ \quad \text{for all} \quad \ka+1\ts \leq \ts b\ts \leq \ts k\ts.
\]
Since \ts $\Delta(\vv w')>0$, we have \ts $\vv w'$ \ts also satisfies~\eqref{eq:openeqns}.
Moreover, the distance between \ts $\vv w'/t$ \ts and \ts $\vv w$ \ts is less than~$e$ (again, by the intercept theorem).
Let \. $\vv w'' := \vv w'/t$.  Then:
\begin{itemize}
 \item \. $\vv w''$ \ts satisfies \eqref{eq:openeqns},
 \item \. the distance between $\vv w''$ and $\vv w$ is less than~$e$, and
 \item \. there are infinitely many nonnegative integer scalars \ts $s \in \IN$ \ts such that
\[
\zeta_b(s\vv w'') \ts \in \ts \IZ \quad \text{for all} \quad \ka+1\ts \leq \ts b\ts \leq \ts k\ts.
\]
\end{itemize}
It remains to prove that \. $\vv w'' \in C_{S'}$.  For this, it suffices to show that for
all large enough \ts $s$ \ts we have
\begin{equation}
\label{eq:circledast}
\tag{$\circledast$}
\zeta_b(s\vv w'') \ts > \ts 0 \quad \text{for all} \quad \ka+1\ts \leq \ts b\ts \leq \ts k\ts.
\end{equation}
This can be seen as follows.
Since
\[
\zeta_b^{\textup{hom}}(s\vv w'') \ts > \ts 0 \quad \text{for all} \quad \ka+1\ts \leq \ts b\ts \leq \ts k\ts,
\]
it follows that
$$
\zeta_b(s \vv w'') \ = \ s^{\deg \zeta_d} \underbrace{\zeta_b^{\textup{hom}}(\vv w'')}_{>0} \, + \, o\big(s^{{\deg \zeta_d}}\big)\ts.
$$
Hence, for all large enough $s$ we have~\eqref{eq:circledast}, as desired.
\end{proof}

\subsection{The witness theorem}
Let \. $S_{\leqslant}:= \{\vv f \in \IO \.\mid\. \exists \ts\vv g \leqslant \vv f : \vv g \in S\}$ \. denote the downwards closure of $S$.
Recall from Section~\ref{sec:binomialbasistechnical} that for \ts $D\subseteq \IO$, a function \. $\Psi:D\to\IN$ \.
is called \defn{$D$-good} \ts if all its coefficients in the binomial basis are nonnegative integers.
Recall that a multivariate polynomial is called \defn{binomial-good} \ts if all coefficients in its binomial basis expansion are nonnegative integers, i.e., it is \ts $\IO$-good.
Otherwise we call the polynomial \defn{binomial-bad}.
A coset of polynomials (for example $\varphi+I$) is called \defn{binomial-good} \ts
if at least one of its representatives is binomial-good, otherwise it is called \defn{binomial-bad}.

\smallskip

\begin{theorem}[{\em \defn{The witness theorem}}{}]
\label{thm:findcounterexample}
Let \. $S\subseteq \IO$, let \. $I=I(S)$ \. be the vanishing ideal, and denote \.
$\IQ[Z] := \IQ[\vv f]/I$. Fix \. $\varphi\in\IQ[\vv f]$,
and suppose that \. $\dim\IQ[Z]_{\leq(\varphi,S)}<\infty$.
For an integer \ts $\Delta\in\IN$, denote \. $D=D(\De,S) := \llbracket \Delta\rrbracket^k\cap S_{\leqslant}$\ts.
Then there exists \ts $\Delta\in\IN$, such that for every binomial-bad coset \ts $(\varphi+I)$ \ts and
a $D$-good function \. $\Psi:D \to \IN$\ts,
there exists
\ts $\vv{f}\in S\cap D$ \ts
with \. $\Psi(\vv f) \neq \varphi(\vv f)$.
\end{theorem}

\begin{proof}
A function in \ts $\IQ[Z]$ \ts is completely specified by its evaluations at all points in~$S$,
which can be seen as follows.  Let \. $\gamma,\gamma' \in \IQ[\vv f]$, such that \ts $\gamma|_S\ts =\ts \gamma'|_S$\ts.
Then \. $(\gamma-\gamma')|_S=0$ \. which implies \. $(\gamma-\gamma')|_Z=0$.  Therefore, \.
$\gamma-\gamma'\in I$, and \. $\gamma+I = \gamma'+I$, i.e., we have equality as functions in~$\IQ[Z]$.
We conclude:
\begin{equation}\label{eq:evalinjective}
\text{
the evaluation map \ $\textsu{eval} :  \IQ[Z]\ts \to \ts \IQ^S$ \ is injective.}
\end{equation}

To simplify the notation, denote \. $V := \IQ[Z]_{\leq(\varphi,S)}$.
By~\eqref{eq:evalinjective}, the restriction \. $\textsu{eval}{}\ts|_V : V \to \IQ^S$ \.
is injective.  Let \ts $r$ \ts denote the rank of the linear map \. $\textsu{eval}{}\ts|_V$,
and note that the injectivity implies that \. $|S|\geq r$.
We now construct a set $S' \subseteq S$.
We start with the empty set and iteratively enlarge it point by point from~$S$
in such a way that the rank of the evaluation map
\. $\textsu{eval}{}\ts|_V : V \to \IQ^{S'}$ \.
increases at each step by~1, up to rank~$r$.
This gives a set \ts $S'\subseteq S$ \ts of $r$ points
such that the map \. $\textsu{eval}_{V,S'}:V\to\IQ^{S'}$ \. is injective.
This means:
\begin{equation}\label{eq:interpolationonSprime}
\text{an element in $V$ is completely specified by its evaluations at~$S'$.}
\end{equation}

Let $d$ be the smallest \ts $d\geq \delta$ \ts  such that \.
$S'\subseteq\llbracket d\rrbracket^k$.
For each \ts $\xi_e$ \ts that is large (see above), let \ts $\omega_e \in S$ \ts
be a point with
\begin{equation}\label{eq:farawaypoint}
\xi_e(\omega_e) \. > \. \varphi(\omega_e).
\end{equation}
Let $\Delta$ be such that for all large \ts $e \in \llbracket d\rrbracket^k$ \ts
we have \. $\omega_e \in \llbracket\Delta\rrbracket^k$.

\smallskip

We are now ready to prove the result.
Note that we have not considered~$\Psi$ so far.
Since $\Psi$ is $D$-good, we can enlarge the domain of definition of $\Psi$ to the whole $\llbracket\Delta\rrbracket^k$ in a way such that $\Psi$ is $\llbracket\Delta\rrbracket^k$-good. Let \ts $\Psi_\De$ \ts denote the unique polynomial
with multidegree \. $\leq (\Delta,\ldots,\Delta)$, and which is binomial-good.

Let \. $\widetilde \Psi$ \. be the polynomial that arises from \ts $\Psi_\De$ \ts by setting all coefficients in the binomial basis whose multidegree is not \. $\leqslant (d,\ldots,d)$ to zero.
Clearly \. $\widetilde \Psi$ \. is also binomial-good.
Moreover, for the restrictions to \ts $\llbracket d\rrbracket^k$, we have \.
$\Psi|_{\llbracket d\rrbracket^k} \ts =\ts \widetilde \Psi|_{\llbracket d\rrbracket^k}$\ts.
On the other hand, on \ts $\llbracket \Delta\rrbracket^k$ \ts we have \.
$\Psi|_{\llbracket \Delta\rrbracket^k} \ts \geq \ts \widetilde \Psi|_{\llbracket \Delta\rrbracket^k}$ \. because $\Psi$ is binomial-good.

Suppose \.
$\varphi|_{S'}\neq \widetilde \Psi|_{S'}$.
This gives the desired \ts $\vv f$, because \.
$S' \subseteq \llbracket d\rrbracket^k\subseteq \llbracket \Delta\rrbracket^k$ \. and
$\Psi|_{\llbracket d\rrbracket^k}=\widetilde \Psi|_{\llbracket d\rrbracket^k}\ts.$
Thus we can consider only the case when
\[
\varphi|_{S'} \. = \. \widetilde \Psi|_{S'}\..
\]
Assume for the sake of contradiction that \. $\widetilde \Psi+I \in V$. Then, by~\eqref{eq:interpolationonSprime},
we have \. $\varphi+I = \widetilde \Psi+I$. This is a contradiction to \ts $\varphi+I$ \ts being binomial-bad and \ts
$\widetilde \Psi+I$ \ts being binomial-good.
Hence we are left to consider the case that
\[
\widetilde \Psi+I \. \notin \ts V.
\]
Since \ts $\widetilde \Psi$ \ts is binomial-good, Claim~\ref{cla:existsxi} gives \. $\vv e \in \textup{suppb}(\widetilde \Psi)$ \.
with \ts $\xi_{\vv e}$ \ts large.
Moreover, since \ts $\widetilde \Psi$ \ts is binomial-good, all coefficients in the binomial expansion are nonnegative.
Hence \. $\widetilde \Psi \geq \xi_e$ on $\IO$.
In particular, we have:
$$\bigl[\widetilde \Psi+I\bigr](\vv\omega_e) \ \geq \ \xi_e(\vv\omega_e) \ \stackrel{\eqref{eq:farawaypoint}}{>} \ \varphi(\vv\omega_e)\ts.
$$
Since \. $\Psi\ts|_{\llbracket \Delta\rrbracket^k} \. \geq \. \widetilde \Psi\ts|_{\llbracket \Delta\rrbracket^k}$\ts, we found the desired \ts $\vv f := \vv \omega_e$.
This completes the proof.
\end{proof}

\subsection{Integer points in the polyhedron $\mathcal P(\varphi,\zeta)$}
\label{subsec:affinelinear}
In the Witness Theorem~\ref{thm:findcounterexample} we need to have a binomial-bad \ts $\varphi+I$.
In this section we discuss the important case of having only affine linear constraints. We use it for many of the cases in this paper.

Formally, given an ideal $I$ generated by \ts $\zeta_b-f_b$ \ts with \ts $b>\ka$, such that
\ts $\zeta_b$ \ts are affine linear in $f_1,\ldots,f_\ka$, and given a polynomial $\varphi$,
the task in this section is to determine whether or not \ts $\varphi+I$ \ts is binomial-good.
We will see that this is the case if and only if the polyhedron \ts $\mathcal P(\varphi,\zeta)$ \ts
contains an integer point, see below.

Let \ts $Z$ \ts be the vanishing set of~$I$, and let \ts $Z_\IZ$ \ts denote its integer points.
We work under the assumption that \. $C'_{\tau^{-1}(Z_\IZ)}$ \. lies Zarsiki-dense in \ts $\IQ^\ka$,
which is for example guaranteed if all \ts $\zeta_b$ \ts have only integer coefficients.

First, note that we can express \ts $\varphi$ \ts in the parameters by plugging in \ts
$\zeta_b$ \ts for each $f_b$, $b > \ka$.
We call the resulting polynomial \. $\varphi'=\tau^*(\varphi)$.
Let \ts $\delta$ \ts be the degree of \ts $\varphi'$.
Recall that \ts $\IO := \IN^k$.
Denote \ts $\IO_\delta := \bigl\{\vv e \in \IO \.:\. |\vv e|\leq \delta\bigr\}$.
Take all $k$ variables $f_b$ and consider the binomial basis elements \.
$\binom{\vv f}{\vv e}$ \. with \. $\vv e\in\IO_\delta$\ts.  Each \ts $\vv e$ \ts gets
a variable \ts $x_{\vv e}$, and we define
\[
q \ := \ \sum_{\vv e\ts\in\ts\IO_\delta} \. \tau^*{}\hskip-0.095cm{}\left(\binom{\vv f}{\vv e}\right) \.\cdot \. x_{\vv e}\..
\]
Note here that \ts $\tau^*$ \ts is only affine linear, so the degree does not increase.

By expressing \. $\tau^*\hskip-0.080cm{}\left(\binom{\vv f}{\vv e}\right)$  over the binomial basis
in \ts $\vv v$, we assign to each \. $\vv e \in \IO_\delta$ \. a coefficient vector \.
$w_{\vv e} \in \IQ^{\nn^\ka_\delta}$, where \. $\nn^\ka_\delta := \big\{\vv v \in \IN^\ka \.:\. |\vv v|\leq\delta\big\}$.
Note that if \ts $\vv e$\ts is contained in \ts $\IN^\ka$, then \ts $y_{\vv e}$ \ts has a single 1 at position \ts $\vv e$
\ts and the rest are zeros.
We also express \ts $\varphi'$ \ts over the binomial basis and get coefficients \.
$\varphi'_{\vv v}$, for all \. $\vv v \in \nn^\ka_\delta$\ts.

We search for a nonnegative integer assignment to the \. $x_{\vv e} \in \IN$ \. such that
\[
\varphi' \ = \ \sum_{\vv e \in \IO_\delta} \. x_{\vv e} \, w_{\vv e}\..
\]
Note that the degree restriction \. $\vv e \in \IO_\delta$ \. is not completely obvious
and we will talk about it in Claim~\ref{cla:degreerestriction}.  We conclude:
\begin{equation}\label{eq:polyhedronequations}
\varphi'_{\vv v} \ = \ \sum_{\vv e \in \IO_\delta} \. x_{\vv e} \, (w_{\vv e})_{\vv v}
\quad \text{for all} \quad \vv v \in \nn^\ka_\delta\.,
\end{equation}
where the index \ts $\vv v$ \ts means taking the coefficient of \ts $\vv v$ \ts in the binomial basis.

Now, the question of the existence of such a vector \ts $x$ \ts of nonnegative integers
is the question of an integer point in the polyhedron \ts $\mathcal P(\varphi,\zeta)$ \ts
defined by the equations \eqref{eq:polyhedronequations} and the nonnegativity constraints on all variables.
The following theorem states this in the simplified situation when \ts $\zeta_b$ \ts have only
integer coefficients, while in general we only need that $\tau^{-1}(Z_\IZ)$ is Zarsiki-dense in~$\ts\IQ^\ka$,
where \ts $Z_\IZ$ \ts are the integer points in~$Z$.

\smallskip

\begin{theorem}[{\em \defn{The polyhedron theorem}}{}]
\label{thm:polyhedron}
Let \ts $I$ \ts be an ideal generated by the \ts $\zeta_b-f_b$, \ts $\ka<b\leq k$,
such that all \ts $\zeta_b$ \ts are affine linear in \. $f_1,\ldots,f_\ka$,
with integer coefficients. 
Fix a polynomial~$\varphi$.
Then \ts
$\varphi+I$ \ts is binomial-good if and only if there exists an integer point in the polyhedron
\ts $\mathcal P(\varphi,\zeta)$.
\end{theorem}

\begin{proof}
The existence of a nonnegative integer point \ts $x_{\vv e}$ \ts that satisfies \eqref{eq:polyhedronequations}
implies that \ts $\varphi+I$ \ts is binomial-good by definition.  For the reverse,
if \ts $\varphi+I$ \ts is binomial-good, then there exist finitely many nonnegative integers \ts
$x_{\vv e}$ \ts with \. $\varphi'\ts =\ts \sum_{\vv e}\ts x_{\vv e}\. w_{\vv e}$.
The following Claim~\ref{cla:degreerestriction} implies that here we can indeed assume that $\vv e \in \IO_\delta$ in this sum, which finishes the proof.
\end{proof}

\smallskip

\begin{claim}
\label{cla:degreerestriction}
Let \ts $x_{\vv e} > 0$ \ts for some \ts $\vv e \in \IN^k$ \ts with \ts $|\vv e| > \delta$.
Moreover, assume that all \ts $x_{\vv e}\geq 0$ \ts and that
$x_{\vv e}> 0$ for only finitely many $\vv e$.
Then \. $\varphi' \ts \neq \ts \sum_{\vv e} \. x_{\vv e} \. w_{\vv e}$.
\end{claim}

\begin{proof}
We define the vector space of functions that grow slowly:
\[
\IQ[\vv v]_{\leq\varphi'} \
:= \
\bigl\{\.\psi \in \IQ[\vv v] \ \big| \ \exists \ts\alpha \in \IQ \ \, \forall\ts \vv g \in \IOp \,:\, |\psi(\vv g)| \. \leq \. \alpha \ts |\varphi'(\vv g)|\.\bigr\}.
\]
Clearly \. $\varphi' \in \IQ[\vv v]_{\leq\varphi'}$\ts.
Since all functions \ts $\tau^*(\xi_{\vv f})$ \ts
are eventually nonnegative forever on any ray \. $\IQ_{\geq 0}\vv v$ \. with \ts
$\vv v \in C_{S'}$, it suffices to prove that \. $\tau^*(\xi_{\vv e})\notin\IQ[\vv v]_{\leq\varphi'}$\ts.

Since \ts $|\vv e| > \delta$, we have \. $\deg \xi_{\vv e}> \deg \varphi'$.
Since \ts $\tau^*$ \ts replaces variables by affine linear non-constant polynomials, and since \ts
$\xi_{\vv e}$ \ts has a single monomial of highest degree,
we have \. $\deg\big(\tau^*(\xi_{\vv e})\big) \ts = \ts \deg(\xi_{\vv e})$.

Since \ts $C_{S'}$ \ts is dense in \ts $\IQ^\ka$\ts, we can choose
\ts $\vv v \in \IOp$ \ts that is strictly positive in each component.
We have
\[
\varphi'(t \vv v) \ = \ \Theta\big(t^{\deg(\varphi')}\big),
\]
whereas
\[
\xi_{\vv e}'(t \vv v) \ = \ \Theta\big(t^{\deg(\xi_{\vv e})}\big) \ = \ \omega\big(t^{\deg(\varphi')}\big).
\]
This implies the result. \end{proof}

\medskip

\section{The diagonalization theorem}
\label{sec:diagonalizationtheorem}
In this section we use set-instantiators, the Witness Theorem~\ref{thm:findcounterexample}
and a multivariate version of Ramsey's hypergraph theorem to prove the Diagonalization Theorem~\ref{thm:diagonalization}.
The Diagonalization Theorem is our only method of proving oracle separations from \ts $\sharpP$ \ts
when the polynomial \ts $\varphi$ \ts is monotone, such as for example \ts $f-1$, $f/2$ \ts or \ts $(f-1)/2$.

\smallskip

For a subset \. $B\subseteq \{0,1\}^j$,  we write \. $\tilde B \subseteq\{0,1\}^{j-1}$ \. do denote
the set of suffixes of all strings that start with~1.  For a subset \. $B\subseteq \{0,1\}^{j-1}$,
we write \. $\{1\}\boxplus B \subseteq\{0,1\}^{j}$ \. to denote the union of \ts
$\{0^j\}$ \ts with the set of strings that start with~1 and continue with a string from~$B$.

\smallskip

\subsection{Set-instantiators}

We want to consider computation paths of nondeterministic Turing machines, but the actual computational device we are arguing about is a nondeterministic Turing machine with oracle access to an oracle that is defined up to strings of length \ts $<j$, and where the oracle answers with~0 for all oracle queries of length \ts $>j$.
We capture this in the following definition.

\begin{definition}
A \defn{computation path} $\tau$ of a nondeterministic Turing machine
on some input is defined as the sequence of its nondeterministic choice bits and the answers to its length \ts $j$ \ts
oracle queries $($both types of bits appear in the same list, ordered chronologically$)$. 
Formally, it is an element of \ts $\{0,1\}^*$. 
\end{definition}

The same Turing machine can yield the same computation path on different inputs (for example, when not the whole input is read) or when having access to different oracles, because the oracles can differ in positions that are not queried.
We are especially interested in the case where the input is \ts $0^j$ \ts
and the oracles differ in exactly the set \. $A_j\subseteq\{0,1\}^j$ \.
of length \ts $j$ \ts strings.

Given a nondeterministic Turing machine \ts $M$ \ts  and an oracle \. $A_{<j} := \bigcup_{j'<j}A_{j'}$ \.
where \. $A_{j'}\subseteq\{0,1\}^{j'}$, and give a subset \. $B \subseteq \{0,1\}^{j-1}$,
we are interested in the number of accepting paths of \ts $M$ \ts when given oracle access to \.
$A_{<j}\cup (\{1\}\boxplus B)$, where \ts $A_{<j}$ \ts is fixed.
We define
\begin{equation}
\label{eq:hMB}
h_M^B(w) \ := \ \#\acc_{M^{A_{<j}\cup (\{1\}\boxplus B)}} \ts (w).
\end{equation}
It is instructive to think of \ts $A_{<j}$ \ts and \ts $M$ \ts
as together forming a computational device that has oracle access to some subset \.
$B\subseteq\{0,1\}^{j-1}$.

\smallskip

For \. $\vv b \in \IN^k$, we write \. $\mP(\vv b) := \mP([b_1]) \times \cdots \times \mP([b_k])$,
where \ts $\mP([a])$ \ts is the set of all subsets of \. $[a]=\{1,\ldots,a\}$.
For \. $\vv s,\ts \vv t \in \mP(\vv b)$, we write \. $\vv s \subseteq \vv t$ \. if \.
$s_a \subseteq t_a$ \. for all \ts $1\leq a \leq k$.
For an element \. $\vv s \in \mP(\vv b)$, we write \. $|\vv s| := \big(|s_1|,\ldots,|s_k|\big)$.

\begin{definition}[Set-instantiator against $(M,j,A_{<j},S,\vv b)$]
\label{def:setinstantiator}
Let \ts $M$ \ts be a nondeterministic Turing machine,
let \ts $j \in \IN$, and let \. $A_{<j} \subseteq \{0,1\}^*$ \.
be a language that contains only strings of length \ts $<j$.
Let \ts $S \subseteq \IN^k$ \ts be a set and let \ts $\vv b \in \IN^k$. 
We set \. $\mP(\vv b)_S := \big\{\vv s \in \mP(\vv b) \,: \, |\vv s|\in S\big\}$.

Let \ts $\top$ \ts be a symbolic top element above \. $\mP(\vv b)$, i.e. $\vv s \subsetneq \top$ \. for all \. $\vv s \in \mP(\vv b)$.
A \ts \defn{set-instantiator} \ts $\SI$ \ts is a pair of

$\circ$ \ an instantiation function \. $\textup{\textsu{inst}}_\SI : \mP(\vv b)_S \to \{0,1\}^{[2^j]}$\ts,  \ts and

$\circ$ \  a perception function \. $\textup{\textsu{perc}}_\SI : \{0,1\}^* \to \mP(\vv b) \cup \{\top\}$\ts,

\smallskip
\nin
such that the following property holds for all \. $\vv s \in \mP(\vv b)_S\.:$

\smallskip
$\bu$ \ $\tau\in\{0,1\}^*$ is an accepting path for the computation \. $h_M^{\textup{\textsu{inst}}_\SI(\vv s)}(0^j)$
if and only if $\textup{\textsu{perc}}_\SI(\tau) \subseteq \vv s$.
\end{definition}

\smallskip

The intuition is that a computation path queries the oracle and sees the existence of several objects ($k$ different types of objects), and then decides to accept or not based solely on the set of objects perceived, independent of whether or not there are actually other unqueried objects in the oracle. The Turing machine might even \emph{know} that there must be other objects for some syntactic reason and can take that information into account.

For example, in \problem{Sperner} we have \ts $k=2$, and we consider rainbow triangles of positive/negative orientation.  We know that \. $t_+ - t_- -1=0$, so if we see \ts $t_+\geq 3$ \ts and \ts $t_-\geq 3$, then we know that there must be at least one rainbow triangle of positive orientation that we have not seen.
Note that if an accepting path $\tau$ sees an object in the oracle and then we change the oracle, then running the same computation we will at some point get a different oracle answer, and hence $\tau$ will not be a computation path of this \ts (input, oracle) \ts combination.

Formally, in the above definition we think of accepting paths as having a perception from $\mP(\vv b)$, while computation paths that never accept on any of the instantiations are given perception $\top$.
Note also that from the definition it is immediately clear that from a set-instantiator with a set~$S$, we get a set-instantiator with the same parameters for every subset of~$S$.

We usually do not mention $A_{<j}$ in the context of set-instantiators, as it has no effect on the construction of set-instantiators, and is also understood from the context.
When discussing polynomial closure properties of $\sharpP$, set-instantiators almost trivially exist (see~$\S$\ref{sec:CountallOccSetInstantiator}), but for counting classes coming from $\TFNP$ this is not obvious. We create the necessary set-instantiators in $\S$\ref{sec:IterSetInstantiator},
$\S$\ref{sec:SourceOrExcessSetInstantiator},
$\S$\ref{sec:CLSSetInstantiator},
$\S$\ref{sec:CountallLeafSetInstantiator},
and
$\S$\ref{sec:BipartiteUnbalanceSetInstantiator}.

\subsection{Diagonalization theorem statement}
\label{sec:diagonalization}

Our main tool for constructing oracles that separate from $\sharpP$ is the following Theorem~\ref{thm:diagonalization}, which depends on the parameters \. $\varphi$, \ts $\zeta$, \ts $S$, \ts \problem{Multiplicities}, and \ts $\vv t$.
In most situations \ts $\vv t\in S$ \ts can be arbitrary, so $\vv t$ is often not specified.
To use the theorem well, \problem{Multiplicities} should map into~$S$.

For \. $A_{<j}\in\{0,1\}^{<j}$,
we say that a nondeterministic oracle Turing machine \ts $M$ \. \defn{answers consistently} \ts for \. $(j,A_{<j},\problem{Multiplicities})$,
if for every $B\subseteq\{0,1\}^{j-1}$ we must have the number \. $\#\acc_M^{A_{<j}\cup {(\{1\}\boxplus B)}}(w)$ \.
is the same for all \ts $B$ \ts that have the same \. $\problem{Multiplicities}(B)$, for all \ts $w\in\{0,1\}^*$.
In other words, we must have:
$$\#\acc_M^{A_{<j}\cup (\{1\}\boxplus B)}(w) \ = \ \#\acc_M^{A_{<j}\cup (\{1\}\boxplus C)}(w)
$$
for all \. $C\in\{0,1\}^{j-1}$ \. with \. $\problem{Multiplicities}(B)=\problem{Multiplicities}(C)$.

\smallskip

\begin{theorem}[{\em \defn{Diagonalization Theorem}}{}]
\label{thm:diagonalization}
Fix \. $0 \leq \ka \leq k$ \. and let \ts $\IO := \IN^k$.
We write \. $\vv f=(f_1,\ldots,f_k)$ \. and \. $\vv v = (v_1,\ldots,v_\ka)$.
Fix \. $\varphi \in \IQ[\vv f]$.  Let  \. $\zeta_b \in \IQ[\vv v]$ \.
be non-constant functions and set $I$ to be the ideal generated by the \.
$\zeta_b(\vv v)-f_b$, where \. $\ka+1 \leq b \leq k$.

Define \. $Z := \{\vv f \in \IQ^k \mid \eq_{\ka+1}(\vv f)=\ldots=\eq_{k}(\vv f)=0\}$,
where \. $\eq_b := \zeta_b(\vv v) - f_{b}$. Denote \. $T := \IO\cap Z$ \. and let \. $S \subseteq T$.
Consider a map \. $\tau:\IQ^\ka\to Z$ \. defined as
$$\tau(v_1,\ldots,v_\ka) \ := \big(v_1,\ldots,v_\ka,\zeta_{\ka+1}(\vv v),\ldots,\zeta_{k}(\vv v)\big).
$$
Set \, $C'_{S} := \big\{\vv v \in \IQ_{\geq 0}^\ka \,:\, |\tau^{-1}(S) \cap \IQ\vv v| = \infty\big\}$.
Assume that

\smallskip

$(1)$ \ $Z$ contains at least one integer point,

$(2)$ \  $\overline{C'_{S}}^{\.\textup{Zar}} = \overline{C'_{T}}^{\.\textup{Zar}}$\ts, \ts and

$(3)$ \  there exists a point \. $\vv v \in \IQ_{\geq 0}^\ka$ \. that satisfies  strict inequalities
\begin{equation*}
\zeta_b^{\textup{hom}}(\vv v) > 0 \quad \text{for all} \quad \ka+1\ts \leq \ts b \ts \leq \ts k,
\end{equation*}
\hskip1.35cm where \. $\zeta_b^{\textup{hom}}\in\IQ[\vv v]$ \. is the top nonzero homogeneous part of \ts $\zeta_b$.

\smallskip

\nin
Fix a set of multivariate functions $\problem{Multiplicities}: \.\mP\big(\{0,1\}^{j-1}\big)\to \IO$.
Assume that for every nondeterministic polynomial-time Turing machine \ts $M$ \ts and
for every \ts $\vv f$, there exist infinitely many \ts $j\in \nn$, such that for every \. $A_{<j}\in\{0,1\}^{<j}$,
either \ts $M$ \ts does not answer consistently for \. $\big(j,A_j,\problem{Multiplicities}\big)$ \. or there
is a set-instantiator \ts $\SI$ \ts against \. $\big(M,j,A_{<j},S,\vv f\big)$
with \. $\problem{Multiplicities}\big(\inst_\SI(\vv s)\big)\ts =\ts |\vv s|$ \. for all \. $\vv s \in \mP(\vv f)_S$.

\smallskip

Fix any \. $\vv t \in S$.  For \. $A \subseteq\{0,1\}^*$, we write:
$$A  \, = \, \bigcup_{j\geq 0} \. A_j\,, \quad \text{where} \quad A_j \subseteq \{0,1\}^j\ts.
$$
Define
\[
p_{A}(w) \ := \
\begin{cases}
\varphi\big(\problem{Multiplicities}\big(\tilde A_{|w|}\big)\big)  & \text{if} \ \  A_{|w|}(0^{|w|}) \. = \. 1 \\
\varphi(\vv t) & \text{otherwise},
\end{cases}
\]
where \. $\tilde A_{j}$ \. is the set of length \ts $j-1$ \ts suffixes of the strings in \ts $A_{j}$ \ts
that start with a~1.

Finally, suppose \. $\varphi+I$ \. is binomial-bad.
Then there exists an oracle \. $A\subseteq\{0,1\}^*$ \. such that
for every nondeterministic polynomial-time Turing machine \ts $M$ \ts
there exists \ts $j$ \ts such that \. $p_A(0^j)\neq \#\acc_{M^A}(0^j)$ \.
and whenever $A(0^j)=1$, then \. $A_j = \{1\}\boxplus\, \inst_\SI(\vv s)$ \.
for some $\vv s$ and one of the $\SI$ above.
\end{theorem}

\smallskip

Note that the technical conditions $(1)$, $(2)$, and~$(3)$ \. are very easy
to check in most situations. They exist to prevent degenerate cases.
The rest of this section is devoted to the proof of Theorem~\ref{thm:diagonalization}.

\smallskip

\subsection{Proof setup}
\label{sec:Diag-thm-proof-setup}
In the beginning, we follow the diagonalization framework from
Theorem~3.1.1(b) in~\cite{C+89},  to construct an oracle \. $A \subseteq \{0,1\}^*$ \.
such that the function \ts $p_A$ \ts has the desired properties.
The actual implementation of this strategy gets quite technical
in some instances: we invoke the Witness Theorem~\ref{thm:findcounterexample}
and use a generalization of Ramsey's theorem.

The computational complexity of the function \ts $\problem{Multiplicities}$ \ts
will not play any role here, but may play a role when invoking the theorem,
where it is usually solved by a $k$-tuple of \ts $\sharpP$ \ts machines.
In fact, $\problem{Multiplicities}$ \ts is usually evaluated on length \ts $2^{j-1}$ \ts
instances.  With the standard identification \. $\mP(\{0,1\}^{j-1}) \simeq \{0,1\}^{2^{j-1}}$ \.
a subset \. $B \subseteq \{0,1\}^{j-1}$ \. can be interpreted as a length \ts $2^{j-1}$ \ts
instance of the \ts $\problem{Multiplicities}$ \ts problem.

We construct \. $A:=\bigcup_{j\in \IN} A_j$ \. as the union of sets \.
$A_j \subseteq \{0,1\}^{j}$, where we will define \ts $A_j$ \ts iteratively for larger
and larger~$j$.
Once we define \ts $A_j$, all \ts $A_{j'}$ \ts with \ts $j'<j$ \ts
will not be changed again. All those \ts $A_{j'}$ \ts with \ts $j'<j$ \ts and
that have not been defined up to that point, are set to~$\emp$.

We enumerate the set of nondeterministic polynomial-time Turing machines that have oracle access.
Note that this enumeration is independent of the specific oracle.
Let $M_i$ denote the $i$-th machine.
For a specific oracle~$A$, we denote the corresponding counting function by \. $H_i^A \in \sharpP^A$. 
We proceed by diagonalizing over $i$, so assume that even though $A$
is not fully specified, we have enough information to know that \.
$H_{i'}^A \neq p_A$ \. for all \ts $i'<i$.
We show that \. $H_{i}^A \neq p_A$ \. by defining further details of~$A$,
i.e., defining \ts $A_j$ \ts for larger \ts $j$.

Let \ts $j$ \ts be the smallest integer such that \ts $A_{j}$ \ts
has not been already defined, and such that \ts $j$ \ts is large enough
for all upcoming claims
(the upcoming constructions do not depend on the exact value of~$j$,
but just require~$j$ to be large).  We now define \ts $A_j$ \ts in this case as follows.

Let \. $A_{<j} \. := \. \bigcup_{j'<j} \. A_{j'}$.
If there exist \. $B \subseteq \{0,1\}^j$ \. with
$$
H_i^{A_{<j} \cup B}(0^j) \, \neq \, p_{A_{<j} \cup B}(0^j)\ts,
$$
then define \ts $A_j := B$.  Also define \ts
$A_{j'}:=\emp$ \ts for all \ts $j < j' \leq j''$, where \ts $j''$ \ts
is the longest length of an oracle query made by the computation \. $H_i^{A_{<j} \cup B}(0^j)$.
In this case we have \. $H_{i'}^A \neq \varphi(f)$ \. for all \ts $i'\leq i$, which was the goal.

For treating the other case, it is sufficient to
analyze the case where for \emph{all} \.
$B \subseteq \{0,1\}^j$ \ts such that \ts $0^j\in B$.
Here we have:
\begin{equation}\label{eq:fphinew}
H_i^{A_{<j} \cup B}(0^j) \ = \ p_{A_{<j} \cup B}(0^j) \ = \ \varphi\big(\problem{Multiplicities}(\tilde B)\big)
\end{equation}
and find a contradiction as follows.

Equation~\eqref{eq:fphinew} implies that the number of accepting paths in the computation \.
$H_i^{A_{<j} \cup B}(0^j)$ \. only depends on the value \ts $\problem{Multiplicities}(\tilde B)$ \ts
and not on  $\tilde B$ itself.
Hence, \ts $M_i$ \ts answers consistently for \. $(j,A_{<j},\problem{Multiplicities})$.
We will need this property mainly for the existence of a set-instantiator for \.
$\#\PLS\problemp{Iter}$, see~$\S$\ref{sec:IterSetInstantiator}.

The case \eqref{eq:fphinew} is then brought to a contradiction in~$\S$\ref{sec:setinstantiators} via instances generated by set-instantiators, see~\eqref{eq:counterexamplepoint}.
We write
$$h_i^B(0^j) \ := \ H_i^{A_{<j}\cup {(\{1\}\boxplus B)}}
$$
to simplify the notation, which is also a simplification of the notation \. $h_{M_i}^{\{1\}\boxplus B}(0^j)$ \. from~\eqref{eq:hMB}.

\smallskip

\subsection{A Ramsey-type theorem}
\label{sec:Ramsey-proof}
The following Ramsey-type result is used in
the next section to construct set-instantiators.   We need the following
definitions, some of which we recall from the previous sections.

\smallskip

Fix the \defn{dimension} \ts $k$ \ts and the \defn{number of colors} \ts
$c$.
Let \. $\vv m = (m_1,\ldots,m_k)\in \nn^k$ \. and \.
$\vv n = (n_1,\ldots,n_k)\in \nn^k$.  We write \. $\vv m \leqslant \vv n$ \.
if \. $m_i\le n_i$ \. for all \. $1\le i \le k$. Recall that \ts
$[n]=\{1,\ldots,n\}$.

Let \ts $X_1,\ldots,X_k$ \ts be an ordered list of finite or countably
infinite sets, and
write \. $\vv X = (X_1,\ldots,X_k)$.  Denote \. $|\vv X| := \bigl(\ts
|X_1|\ts,\ts\ldots\ts,\ts|X_k|\ts\bigr)$.
For $\vv a \in \IN^k$ we say that \. $\vv X$ \. is an \ts \defn{$\vv
a$-list} \ts if \. $|\vv X| = \vv a$.

Denote \. $\vv Y\subseteq \vv X$ \.
when we have \. $Y_i\subseteq X_i$ \. for all \. $1\le i \le k$.  We say
that
$\vv Y$ \. is a \ts \defn{subset-list} \ts of \. $\vv X$ \. in this
case, and
denote by \. $\BB(\vv X):=\bigl\{\vv Z \.:\. \vv Z \subseteq \vv
X\bigr\}$ \.
the set of subset-lists of~$\ts \vv X$.
A subset-list that is also an $\vv a$-list is called an \ts \defn{$\vv
a$-subset-list}.

A \defn{$c$-coloring} \ts of \ts $\vv X$ \ts is a map \.
$\BB(\vv X) \to [c]$.  We say that \. $\vv X$ \.  is \ts \defn{$\vv
m$--monochromatic} \ts
if for all \. $\vv a \leqslant \vv m$, all \. $\vv a$-subset-lists \.
$\vv Y\subseteq \vv X$
have the same color.

\smallskip

\begin{theorem}[{\em\defn{multipartite hypergraph Ramsey theorem}}{}]
\label{thm:ramsey}
Fix \ts $k\ge 1$.
For every \. $\vv n$, $\vv m \in \nn^k$ \. and every integer \ts $c\ge 1$,
there exists a natural number \. $\RR=\RR(\vv n,\vv m,c)$,
such that for every $c$-coloring of a list \. $\vv X=(X_1,\ldots,X_k)$ \.
of finite sets of size \. $|X_i|\ge \RR$, there is a \ts $\vv m$--monochromatic
subset-list \. $\vv Y\subseteq \vv X$ \. with \. $|\vv Y|\geqslant \vv n$.
\end{theorem}

\smallskip

Since we do not need the explicit quantitative bounds,
the proof below is based on the approach in~\cite{GPS12}, rather than a more
standard argument in Ramsey theory, see e.g.~\cite{GRS90,Nes95}.  We give it here
for the sake of completeness and because it is not covered by the (usual)
\defng{hypergraph Ramsey theorem}.

\smallskip

\begin{proof}
Given $k$ countably infinite sets \ts $X_1,\ldots,X_k$\ts, and
let all subset-lists of cardinality \ts $\leqslant \vv m$ \ts be colored with~$[c]$.
We first prove:

\smallskip

\nin
($\lozenge$) \  there is an infinite $\vv m$--monochromatic subset list \. $\vv Y\subseteq \vv X$.
Here \. $\vv Y = (Y_1,\ldots,Y_k)$, such that \. $Y_i\subseteq X_i$ \. are infinite, for all \.
$1\le i \le k$.

\smallskip

Fix \. $\vv a \leqslant \vv m$.
If we can prove that there are infinite subsets \. $Y_1\subseteq X_1,\ldots,Y_k\subseteq X_k$ \.
for which all $\vv a$-subset-lists are monochromatic, then we can iterate this argument on \.
$\vv Y$ for a another vector \. $\vv a' \leqslant \vv m$, and so on, until we treated all \. $\vv a \leqslant \vv m$.

We proceed by induction. We assume by induction that for each \. $\vv b \leqslant \vv a$, \.
$\vv b \neq \vv a$, we have a subset-list of infinite subsets such that each $\vv b$-subset is monochromatic.
We choose \ts $q \in [k]$ \ts with \ts $a_q \neq 0$.
Let \. $\vv b := (a_1,\ldots,a_{q-1},0,a_{q+1},\ldots,a_k)$.
Take any way of enumerating the set of all $\vv b$-subset-lists.
Note that every $\vv b$-subset-list can be extended to an $\vv a$-subset-list in infinitely many ways.
Each of these extensions has a color.

For the first $\vv b$-subset-list we choose an infinite monochromatic subset in the set of all extensions
(which exists by the usual hypergraph Ramsey theorem).
We attach the color of this subset to the $\vv b$-subset-list and call it its \emph{infcolor}.
Now delete all vertices in coordinate $q$ that are not in this monochromatic set and proceed with the next $\vv b$-subset-list.
We iterate this, and in this way attach an infcolor to each $\vv b$-subset-list.
By induction we can choose a subset-list of infinite subsets in which all
$\vv b$-subset-lists are
monochromatic w.r.t.\ infcolor.
In the resulting subset-list of infinite subsets, all $\vv a$-subset-lists have the same color.  This completes the proof of~$(\lozenge)$.

\smallskip

We now prove the existence of the lower bound \ts $\RR(\vv n,\vv m,c)$ \ts as in the theorem.
Assume that there exists no such bound.
Then for each \ts $z\in\IN$ \ts we find a coloring \ts $\textup{col}_z$ \ts
of the subset-lists of \ts $[z]^{k}$ \ts such that \ts $[z]^{k}$ \ts
contains no cardinality \ts $\vv m$ \ts subset-list in which for all \ts $\vv a \leqslant \vv m$ \ts
all $\vv a$-subset-lists are monochromatic.

We enumerate the set of all subset-lists of \ts $\IN^{k}$ \ts
that are an $\vv a$-subset-list for some \ts $\vv a \leqslant \vv m$.
We assign a color to the first of this list so that the color agrees
with the color choice of infinitely many $\textup{col}_z$ (note that not all
\ts $\textup{col}_z$ \ts might assign a color to this element, but infinitely many do).
Now assign a color to the second of the list so that it agrees with the color choice of infinitely
many \ts $\textup{col}_z$ \ts that already agree with the color of the first element.
We proceed in this manner and obtain a coloring of all $\vv a$-subset-lists, $\vv a \leqslant \vv m$,
of \ts $\IN^{k}$.  By construction, no cardinality \ts $\vv m$ \ts subset-list
has for all \ts $\vv a \leqslant \vv m$ \ts that all $\vv a$-subset-lists are monochromatic.
This is a contradiction to ($\lozenge$), as desired.
\end{proof}

\smallskip

\subsection{Set-instantiators via Ramsey's theorem}
\label{sec:setinstantiators}
We can now finish the proof of the Diagonalization Theorem~\ref{thm:diagonalization} that
we started in~$\S$\ref{sec:Diag-thm-proof-setup}.
Recall that we invoke the Witness Theorem~\ref{thm:findcounterexample} with our \ts $S$ \ts and \ts $\varphi$ \ts
to obtain a \ts $\Delta\in\IN$ \ts as in the theorem.

Denote \. $\Delta^{\times k} := (\Delta,\ldots,\Delta) \in \IN^k$.
Let \ts $\Lam \geq \Delta$ \ts be large enough for Theorem~\ref{thm:ramsey} to hold
(note that this also implies that $j$ is large).
Let \. $\SI:=(\SI)_{\Lam^{\times k}}$ \. be the set-instantiator with \. $\vv b = \Lam^{\times k}$,
with \. $M=M_i$, \ts  $j$ \ts and \. $A_{<j}$ \. being as before, and with \ts $S$ \ts
is intersected with \. $\llbracket\Lam\rrbracket^{\times k}$.
From now on we have $S \leftarrow S \cap \llbracket\Lam\rrbracket^{k}$.
Recall that \. $\mP\bigl(\Lam^{\times k}\bigr) = \mP([\Lam]) \times \ldots \times \mP([\Lam])$, \. $k$ \ts times.

\smallskip

\begin{definition}[Filling the definitional holes with ``perception or below'']
For $\vv t \in \mP(\Lam^{\times k})$ we define $\Phi(\vv t)$ as the number of $\tau\in\{0,1\}^*$
that satisfy $\textup{\textsu{perc}}_\SI(\tau) \subseteq \vv t$.
\end{definition}

We observe that for \. $\vv t \in \mP(\Lam^{\times k})_S$ \. we have \. $\Phi(\vv t) = h_i^{\textsu{inst}_\SI(\vv t)}(0^j)$.
So this definition fills the holes in the domain of definition of the function $\vv t \mapsto h_i^{\textsu{inst}_\SI(\vv t)}(0^j)$ that is only defined on $\mP(\Lam^{\times k})_S \subseteq \mP(\Lam^{\times k})$.
Note that
\begin{equation}\label{eq:Phibinomialgood}
\Phi(\vv t) \ = \ \bigl|\bigl\{\.\tau\in\{0,1\}^* \,: \, \textsu{perc}_\SI(\tau)\subseteq\vv t\.\bigr\}\bigr|,
\end{equation}
so in particular the function \ts $\Phi$ \ts is monotonic.

By Definition~\ref{def:setinstantiator},
for every \. $\vv t \in \mP(\Lam^{\times k})$ \.
with \. $|\vv t| \in S_{\leqslant}$\ts,
we find \. $\vv s \in \mP(\Lam^{\times k})_S$ \. such that \. $\vv t \subseteq \vv s$
Take an \ts $\vv s$ \ts for which \ts $|\vv s|$ \ts is lexicographically smallest.
Define \. $\Omega(\vv t):=\varphi(|\vv s|)$.
Note that by the definition of a set-instantiator,
\ts $\Omega(\vv t)$ \ts is the same number
for all \ts $\vv t$ \ts that have the same \ts $|\vv t|$,
so \. $\overline\Omega\bigl(|\vv t|\bigr)\ts :=\ts \Omega(\vv t)$ \. is well-defined.

If \. $\Phi(\vv t)>\Omega(\vv t)$, then
also \. $\Phi(\vv s)>\Omega(\vv t)$ \. by monotonicity.  Hence \.
$h_i^{\textsu{inst}_\SI(\vv s)}(0^j) > \varphi(|\vv s|)$.
By Definition~\ref{def:setinstantiator},
we have \.
$\problem{Multiplicities}(\inst_\SI(\vv s))=|\vv s|$ \.
It now follows that \.
$h_i^{\textsu{inst}_\SI(\vv s)}(0^j) > \varphi(\problem{Multiplicities}(\inst_\SI(\vv s)))$.
This allows us to define \. $A_j := (\inst_\SI(\vv s))'$ \.
and get a contradiction to~\eqref{eq:fphinew}.

It remains to treat the case when this is impossible.
Therefore, from now on we assume that for all \ts $\vv t$ \ts with \. $|\vv t| \in S_\leqslant$ \.
we have \. $\Phi(\vv t) \leq \Omega(\vv t) = \varphi\big(|\vv s|\big)$.
We are not interested in \. $\vv t\notin S_\leqslant$.
Define \. $c \ts := \ts 1\ts +\ts \max\big\{\.\overline\Omega(\vv f) \,:\, \vv f \in \llbracket\Delta\rrbracket^k\cap S_\leqslant\ts\big\}$.

\smallskip

\begin{proposition}
\label{pro:ramseyLambdaDelta}
Let \ts $\Phi$ \ts be a function defined on the set \ts $\mP([\Lam])^k$ \ts
such that \. $\Phi(\vv t) \in\{0,\ldots,c-1\}$ \.
for all \ts $\vv t$ \ts with $|\vv t| \in \llbracket \Delta\rrbracket^k$.
Then, there exists \. $\La=\La(\De,c)\in \nn$, such that for all \ts $\Lam\ge \La$ \ts
there exist subsets \. $Q_1, \ldots,Q_k \subseteq[\Lam]$ \. with \.
$|Q_1|=\ldots = |Q_k|=\Delta$, with the property:
$$
|\vv s| = |\vv t| \in S_\leqslant \ \  \text{for some} \ \ \vv s, \vv t \in \vv Q
\quad \Longrightarrow \quad \Phi(\vv s)\. = \. \Phi(\vv t)\ts.
$$
\end{proposition}

\begin{proof}
We assign to each element \. $\vv s \in \mP([\Lam])^k$ \. with \. $|\vv s|\in\llbracket\Delta\rrbracket^k$
the color \. $\Phi(\vv s)$ \. for all \. $|\vv s|\in S_{\leqslant}$.  Otherwise, assign the color~$0$.
Let \. $\La := \RR\big(\Delta^{\times k},\Delta^{\times k},c\big)$ \. and use Theorem~\ref{thm:ramsey}
to find the desired subsets \. $Q_1,\ldots,Q_k$.
\end{proof}

\smallskip

From Proposition~\ref{pro:ramseyLambdaDelta}, we can readily construct a new set-instantiator $\overline{\SI}$ with $\vv b=\Delta^{\times k}$ as follows (this is similar to restricting to a subset of $S$ and renaming the elements).

Take a bijection \. $\beta_a: Q_a \to [\Delta]$ \. for all \ts $a \in [k]$.
This induces a bijection
$$\vv \beta \. : \. \mP(Q_1)\ts \times\ts \ldots \ts \times \ts \mP(Q_k) \ \longrightarrow
\ \mP\big([\Delta]\big)\ts \times \ts \dots \ts \times \ts \mP\big([\Delta]\big),
$$
which we use for identification of the sets.
We set \. $\inst_{\overline\SI}(\vv s):=\inst_{\SI}(\vv\beta^{-1}(\vv s))$.
Define \.
$$
\perc_{\overline\SI}(\tau)\ts :=\ts \vv\beta\big(\perc_{\SI}(\tau)\big) \quad \text{if} \quad
\vv\beta\big(\perc_{\SI}(\tau)\big)\ts \in\ts \mP\big([\Delta]\big)^k\cup\{\top\}\ts,
$$
and let \. $\perc_{\overline\SI}(\tau):=\top$ \. otherwise.
With \. $S \leftarrow S \cap\llbracket\Delta\rrbracket^k$,
it is easy to check that \ts $\overline{\SI}$ \ts is a set-instantiator.

By analogy to \eqref{eq:Phibinomialgood}, define
\begin{equation}\label{eq:Phibinomialgoodnew}
\overline\Phi(\vv t) \,
:= \,
\bigl|\big\{\ts\tau\in\{0,1\}^* \,:\, \textsu{perc}_{\overline\SI}(\tau)\subseteq\vv t\ts\big\}\bigr| \
 = \
\sum_{\vv s \ts\subseteq \ts\vv t} \, \bigl|\big\{\ts\tau\in\{0,1\}^* \,:\, \textsu{perc}_{\overline{\SI}}(\tau)=\vv s\ts\big\}\bigr|.
\end{equation}
Since \. $\overline{\Phi}(\vv t)$ \. depends only on \. $|\vv t|$ \. for all \. $|\vv t| \in S_{\leqslant}$,
by induction we see that the number
$$
\big|\big\{\tau\in\{0,1\}^* \, :\, \textsu{perc}_{\overline{\SI}}(\tau)=\vv s\big\}\big|
$$
also depends only on $|\vv s|$.
Denote this number by \ts $c_{\vv g}$, where \. $\vv g := |\vv s| \in S_\leqslant$\ts, and
set \. $\Psi(\vv f) := \overline{\Phi}(\vv t)$ \. for \. $\vv f = |\vv t|$.
This means that \eqref{eq:Phibinomialgoodnew} simplifies:
\begin{equation}\label{eq:PhibinomialgoodII}
\Psi(\vv f) \ =  \ \sum_{\vv g \ts\leqslant\ts \vv f} \. \binom{\vv f}{\vv g}\. c_{\vv g}\..
\end{equation}

The function \ts $\Psi(\vv f)$ \ts is defined on \.
$S_{\leqslant}\cap \llbracket\Delta\rrbracket^k$ \.  and is \.
$S_{\leqslant}\cap \llbracket\Delta\rrbracket^k$-good,
which is a requirement for the second part of the Witness Theorem~\ref{thm:findcounterexample} (we already used it to obtain $\Delta$).
We invoke the second part of the Witness Theorem~\ref{thm:findcounterexample} to find a point \.
$\vv f \in S\cap\llbracket\Delta\rrbracket^k$ with $\Psi(\vv f)\neq \varphi(\vv f)$.
We take \. $\vv s \in (\mP([\Delta])^k)_S$ \. with \ts $|\vv s|=\vv f$, and observe
\begin{equation}
\label{eq:counterexamplepoint}
h_i^{\textup{\textsu{inst}}_{\SI(\beta^{-1}(\vv s))}}(0^j) \, = \,
h_i^{\textup{\textsu{inst}}_{\overline{\SI}(\vv s)}}(0^j) \, = \, \Psi(\vv f) \, \neq \, \varphi(\vv f)\ts.
\end{equation}
On the other hand, recall that \.
$\problem{Multiplicities}(\inst_\SI(\vv t))=|\vv t|$,
and that by Definition~\ref{def:setinstantiator},
we have \.
$|\vv t|=|\beta(\vv t)|$.
Hence, for $\vv s = \beta(\vv t)$ \. we have:
\[
h_i^{\textup{\textsu{inst}}_{\SI(\beta^{-1}(\vv s))}}(0^j) \, \neq \,
\varphi\big(\problem{Multiplicities}\big(\inst_\SI\big(\beta^{-1}(\vv s)\big)\big)\big).
\]
Finally, define \. $A_j = \big(\textup{\textsu{inst}}_{\SI}\big(\beta^{-1}(\vv s)\big)\big)'$ \. to obtain
the desired contradiction to \eqref{eq:fphinew}.  This completes the proof of
Theorem~\ref{thm:diagonalization}.  \qed

\smallskip

\subsection{A set-instantiator for
{\normalfont\problem{OccurrenceMulti}}
}
\label{sec:CountallOccSetInstantiator}
Let \. $\problem{OccurrenceMulti}_{k}:\{0,1\}^*\to\IN^k$ \. be the function defined as follows.
On input \ts $w\in\{0,1\}^*$ \ts we split $w$ into $k$ parts of roughly the same size,
and the output is the vector that specifies how many $1$'s are in the first part,
how many $1$'s are in the second part, and so on.
We omit the index~$k$ when it is clear from the context.

\smallskip

\begin{theorem}\label{thm:occurrencesetinstantiator}
Let \ts $M$ \ts be a  given a polynomial time nondeterministic Turing machine.
Fix~$k$.  For every \ts $\vv b \in \IN^k$ \ts
there exists a threshold \ts $j_0\in \IN$, such that for every \ts $j \geq j_0$ \ts
and every \. $A_{<j}\in\{0,1\}^{<j}$ \.
there exists a set-instantiator \ts $\SI$ \ts against \.
$(M,j,A_{<j},S \subseteq \IN^k,\vv b)$,  such that
\. $\problem{OccurrenceMulti}_k(\inst_\SI(\vv s)) = |\vv s|$.
\end{theorem}

\smallskip

The rest of this section is devoted to proving this theorem for $S = \IN^k$, which immediately proves it for all subsets.
We will first define a \defn{creator} whose \defn{creations} will be the instantiations in the end, but the creator is not limited to a single creation for each set. We will then define the set-instantiator from the creator by picking for every subset \ts $\vv s$ \ts with \ts
$|\vv s|\in \IN^k$ \ts just any one of its creations for~$\ts \vv s$.

\subsubsection{Creations}
Having $2^{j-1}$ many bits available, we can encode in a standard way a list of $k$ subsets \.
$\psi=(\psi_1,\ldots,\psi_k)$, where \ts $\psi_a \subseteq [2^{j-2}]$.
We ignore any extra bits that are not needed for this encoding.
Let \ts $n := j-2$.
For a \. $(c,x) \in [k]\times [2^n]$ \. we write \ts
$\psi(c,x):=1$ \ts if \ts $x \in \psi_c$, and \ts $\psi(c,x):=0$ \ts otherwise.
We say that \ts $(c,x)$ \ts is a \defn{hit} \ts if \ts $\psi(c,x)=1$, otherwise it is a \defn{miss}.

Let \. $\textsu{hitvec}(\psi) := \{(c,x) \.:\. \psi(c,x)=1\}.$
Let $\textsu{hit}(\psi_c) \subseteq [2^n]$ denote the set of $x \in [2^n]$ with $\psi(c,x)=1$.
For a vector \. $\vv a \in \IN^k$,
a list $\psi$ is called an \defn{$\vv a$-creation} \ts if \.
$|\textsu{hit}(\psi_c)| = a_c$ \. for all \. $1 \leq c \leq k$.
Let \ts $\textsu{$\vv a$-creations}$ \ts denote the set of all \ts $\vv a$-creations.
Finally, let  \ts
$$\textsu{$a^-$-creations} \, := \, \bigcup_{1 \leqslant \vv v \leqslant \vv a}\textsu{$\vv v$\,-creations}
$$
These will be used extensively in Section~\ref{sec:TFNP}.

\subsubsection{Lucky creators}
In this section we introduce the concept of a \ts \defn{creator}.
A creator is a slightly less restrictive version of a set-instantiator.
For each blueprint a creator outputs a creation, where a blueprint is slightly
more general than the set $\vv s$ for a set-instantiator, but serves the same purpose.

For a set $X$, let \ts $\binom{X}{c}_\textup{ordered}$ \ts denote the set of cardinality~$c$
ordered subsets of~$X$ (i.e., length $c$ lists of pairwise distinct objects of~$X$).
For \ts $\vv b\in \IN^k$, a \defn{$\vv b$-creator} \ts $\xi$ \ts is a
length $k$ list \ts $\textsu{hitmatrix}(\xi)$.
The $c$-th entry of this list is a length $b_c$ list of distinct entries from $[2^n]$,
i.e., an element from \. $\binom{[2^n]}{b_c}_\textup{ordered}$\..
We can (and will) think of \ts $\textsu{hitmatrix}(\xi)$ \ts as a set of points in \ts
$[k]\times[2^n]$.

Given a $\vv b$-creator~$\xi$ and an \ts
$\vv L\in\mP(\vv b)$, we obtain an $|\vv L|$-creation \ts $\xi_{\vv L}$ \ts
by setting \. $\xi_{\vv L}(c,\textsu{hitmatrix}(\xi)_{c,d})=1$ \. for all \ts $d \in L_c$, and \ts
$\xi_{\vv L}(c,x)=0$ \ts otherwise.

\begin{definition}
We call $\xi \in \textsu{$b$-creators}$ \ts \defn{lucky} \ts if
for all \ts
$\vv L\in\mP(\vv b)$,
all accepting paths of the computation \ts $h_i^{\xi_{\vv L}}(0^j)$ \ts
do not access the oracle at any point in \. $\textsu{hitmatrix}(\xi)\sm \textsu{hitvec}(\xi_{\vv L})$.
\end{definition}

In a sense, this says that the accepting paths do not access the oracle at positions where lonely nodes are not, but could potentially be.

Suppose for all $j\ge 1$, we have a probability distribution \ts $D_{j}$ \ts and an event \ts $e_{j}$, \ts which satisfy \. $\lim_{j\to\infty}\Pr_{D_{j}}[e_{j}] = 1$.  Then we say that \ts \emph{$e_{j}$ happens with high probability} (w.h.p.) \emph{in $D_j$}.
Given a finite number of events~$e_j$ that each happen with high probability in~$D_j$, by the union bound all these events happen simultaneously with high probability in~$D_j$.  Let \ts $U_X$ \ts denote the uniform distribution on a finite set~$X$.

\begin{claim}
If $\xi$ is sampled from \. $U_\textup{\textsu{$b$-creators}}$\ts, then \ts $\xi$ \ts is lucky w.h.p.
\end{claim}
\begin{proof}
We observe that for a fixed \ts $L\in\mP(\vv b)$, we have that the $\vv L$-creation \ts
$\xi_{L}$ \ts is uniformly distributed from \. \mbox{$\vv L$\textsu{-creations}}.
We show that for a fixed \ts $\vv L\in\mP(\vv b)$ \ts we have that $\xi$ is lucky w.h.p.

Since there are only constantly many $\vv L$, the claim immediately follows from the union bound.
Hence, for the rest of the proof fix \ts $\vv L$.
We have that \ts $\xi_{\vv L}$ \ts is uniformly distributed.
For \ts $l := \sum_{a=1}^k |L_a|$,
the probability of an oracle access picking one of these positions is \. $\leq \frac{kl}{2^n-l}$.
By Bernoulli's inequality, we have
$$\left(1-\frac{kl}{2^n-l}\right)^{t_i(j)} \, \geq \, 1 \. - \. \frac{k\ts l \ts t_i(j)}{2^n-l}\..
$$ Since $k$ and $l$ are fixed and
since $n$ and $j$ are polynomially related, this proves the claim.
\end{proof}

\subsubsection{Defining the set-instantiator}
For a $\vv a$-creation $\psi$ and a computation path \ts $\tau \in \{0,1\}^*$ \ts of a computation \ts $h_i^\psi(0^j)$, let  \. $\textsu{perception}(\tau) \subseteq \textsu{hit}(\psi)$ \. denote the set of accessed oracle positions that are hits.

Let $\xi$ be a lucky $\vv b$-creator.
We interpret the set \ts $\textup{\textsu{hitmatrix}}(\xi)$
as a bijection \. $\bij : \mP(\vv b)\to\textup{\textsu{hitmatrix}}(\xi)$.  Let \ts $S = \IN^k$.
We set
\[
\inst_\SI(\vv s) \. := \. \xi_{\vv s}\ts,
\]
and for \ts $\tau\in\{0,1\}^*$ \ts we set
\[
\textsu{perc}_\SI(\tau) = \begin{cases}
\textsu{bij}^{-1}(\textsu{perception}(\tau)) & \text{ if \ts $\tau$ \ts is an accepting path of the computation \. $h_i^{\inst_\SI(\vv b)}(0^j)$},
\\
\top & \text{ \ts otherwise.}
\end{cases}
\]

The rest of this section is devoted to proving that $\SI$ satisfies the requirements of Definition~\ref{def:setinstantiator}, which then proves Theorem~\ref{thm:occurrencesetinstantiator}, because clearly
$\problem{OccurrenceMulti}(\inst_\SI(\vv s))=|\vv s|$.
Formally:

\smallskip

\begin{proposition}
For all $\vv s \in \mP(\vv b)_S$:
$\tau\in\{0,1\}^*$ is an accepting path for the computation \. $h^{\inst_\SI(\vv s)}(0^j)$ \. if and only if \. $\textup{\textsu{perc}}_\SI(\tau) \subseteq \vv s$.
\end{proposition}

\begin{proof}
Since \ts $\tau$ \ts is an accepting path of \. $h^{\inst_\SI(\vv s)}(0^j)$,
and since \ts $\xi$ \ts is lucky, we conclude that \ts $\tau$ \ts is also an accepting
path of the computation \. $h^{\inst_\SI(\vv b)}(0^j)$. This implies that \.
$\textsu{perc}_\SI(\tau) = \textsu{bij}^{-1}(\textsu{perception}(\tau))$.
Clearly \. $\textsu{perception}(\tau) \subseteq \textsu{bij}(\vv s)$,
since otherwise \ts $\tau$ \ts would not even be a computational path of \. $h^{\inst_\SI(\vv s)}(0^j)$ \. because of its oracle answers when querying lonely nodes in \. $\textsu{perception}(\tau) \sm \textsu{bij}(\vv s)$. We conclude: \.
$\textup{\textsu{perc}}_\SI(\tau) \subseteq \vv s$.

The argument above is reversible.  Indeed,
let \. $\textup{\textsu{perc}}_\SI(\tau) \subseteq \vv s$, so in particular \. $\textup{\textsu{perc}}_\SI(\tau) \neq \top$.
Then \ts $\tau$ \ts is an accepting path of the computation \. $h^{\inst_\SI(\vv b)}(0^j)$.
Since \ts $\tau$ \ts is an accepting path of \. $h^{\inst_\SI(\vv b)}(0^j)$,
and since \ts $\xi$ \ts is lucky, we conclude that \ts $\tau$ \ts is also
an accepting path of the computation \. $h^{\inst_\SI(\vv s)}(0^j)$.
\end{proof}

\subsection{Binomial-good polynomials and relativizing closure properties of {\normalfont$\sharpP$}}
\label{sec:affinevarietyseparation}

We now draw an important corollary from the Diagonalization Theorem~\ref{thm:diagonalization} in a simple subcase that allows us to completely characterize the relativizing multivariate polynomial closure properties of $\sharpP$.

\smallskip

\begin{theorem}\label{thm:relatclosure}
The relativizing multivariate polynomial closure properties of \ts $\sharpP$ \ts
are exactly the binomial-good polynomials.
\end{theorem}

\begin{proof}
Let $\varphi$ be binomial-bad. We prove that there exists \ts $A\subseteq\{0,1\}^*$ \ts such that \. $\varphi(\vv\sharpP^A)\not\subseteq\sharpP^A$.
We use a very simple instantiation of the Diagonalization Theorem~\ref{thm:diagonalization} as follows.

Let $k$ be the arity of~$\varphi$, and let \ts $\ka=k$.
We have no functions \ts $\zeta_b$ \ts in this case.  The ideal \ts $I=\langle 0\rangle$. Let \ts $S = T= \IO$. Then \ts $C'_S = \IQ_{\geq 0}^k$.
The assumptions of the Diagonalization Theorem are readily verified: (1) $\vv 0 \in Z$, (2) $S=T$, (3) there are no inequalities.
We set \. $\vv t = \vv 0$ \. and set \.
$\problem{Multiplicities}=\problem{OccurrenceMulti}_k$  \ts
(for which we have set-instantiators from Theorem~\ref{thm:occurrencesetinstantiator}).

Since \. $\varphi+I = \varphi+\langle 0 \rangle$ \. is binomial-bad, there exists $A$ such that for every~$M$, there is $j$ with \ts
$p_A(0^j) \neq \#\acc_{M^A}(0^j)$. The proof is finished by observing that \ts $p_A \in \varphi(\vv \sharpP^A)$, which
can be seen as follows.
There exists a list of oracle Turing machines that on input $w$ first query $A$ at $0^{|w|}$.  Then, if \ts $0^{|w|}\notin A$, accept with multiplicities $\vv t$; if \ts $0^{|w|}\in A$, then the machines run \. $\problem{OccurrenceMulti}$ \. on \ts $\tilde A_j$.\end{proof}

\medskip

\section{Applications to classical problems}
\label{sec:combinineq}
In this section we first continue the approach in $\S$\ref{ss:main-monotone}
to apply Proposition~\ref{p:GapP-squared-PH} to complete squares.  We then
apply the Diagonalization Theorem~\ref{thm:diagonalization}
and its corollary, Theorem~\ref{thm:relatclosure},
in several settings.

\subsection{Complete squares}
\label{sec:completesquaresandnonmonotone}

The \defn{Cauchy inequality} \ts is the basic inequality which goes into
the definition of the \defng{scalar product} in~$\rr^n$. It is the starting point in
\cite[Eq.~(1.1.1)]{HLP52}:
\begin{equation}\label{eq:Cauchy}
\bigl(x_1y_1 \. + \. \ldots \. + \.x_ny_n\bigr)^2 \, \le \, \bigl(x_1^2 \. + \. \ldots \. + \.x_n^2\bigr)\bigl(y_1^2 \. + \. \ldots \. + \.y_n^2\bigr)
\end{equation}
Geometrically, the inequality says that cosine of every angle in \ts $\rr^n$ \ts is at most~$1$.
It is a special case of the \defn{H\"older inequality} for the $L^p$-norm when $p=2$, see
e.g.~\cite[$\S$17]{BB61}.

\smallskip

\begin{proposition}[Cauchy inequality] \label{p:Cauchy}
Denote by \ts $C_n=C_n(x_1,\ldots,x_n,y_1,\ldots,y_n)$ \ts the
counting function given by~\eqref{eq:Cauchy}, i.e., the difference of the right-hand side and the left-hand side of the inequality.  For \. $n\ge 2$, we have \.
$C_n \not\geqslant_\# 0$ \. unless \. $\PH = \Sigma_2^\textup{\textsu{p}}$\ts.
\end{proposition}

\begin{proof} Take \ts $n=2$.
We have
$$
C_2 \, = \, \bigl(x_1^2\ts +\ts x_2^2\bigr)\bigl(y_1^2\ts +\ts y_2^2\bigr) \, - \, \bigl(x_1y_1\ts +\ts x_2y_2\bigr)^2 \ = \
x_1^2y_2^2 \. + \. x_2^2y_1^2 \. - \. 2\ts x_1y_1x_2y_2 \, \ge \, 0.
$$
When \ts $y_1=y_2=1$, this is equivalent to \ts $(x_1-x_2)^2\geqslant 0$, and the result follows
from Corollary~\ref{c:AMGM}. For $n>2$, let \. $x_3=\ldots = x_n = y_3=\ldots=y_n=0$.
\end{proof}

\smallskip

The \defn{Minkowski inequality} \ts is another basic inequality, see e.g.~\cite[$\S$21]{BB61}:
\begin{equation}\label{eq:Mink}
\prod_{i=1}^{n}
\bigl(x_i^n \. + \ y_i^n\bigr) \ \ge \
\left[\,\prod_{i=1}^{n} \. x_i \, + \, \prod_{i=1}^{n} y_i\,\right]^n
\end{equation}
This inequality is a special case of the \defng{Brunn--Minkowski inequality} (for bricks in~$\rr^n$),
which is foundational in the theory of geometric inequalities, see e.g.~\cite[$\S$8]{BZ88}

\smallskip

\begin{corollary}[Minkowski inequality] \label{c:Mink}
Denote by \ts $M_n=M_n(x_1,\ldots,x_n,y_1,\ldots,y_n)$ \ts the
counting function given by~\eqref{eq:Mink}, i.e., the difference of the left-hand side and the right-hand side of the inequality.  For \. $n= 2$, we have \.
$M_2 \not\geqslant_\# 0$ \. unless \. $\PH = \Sigma_2^\textup{\textsu{p}}$\ts.
\end{corollary}

\begin{proof}  Note that \. $M_2(x_1,x_2,y_1,y_2) = C_2(x_1,y_1,x_2,y_2)$.
\end{proof}

\smallskip

The \defn{Alexandrov--Fenchel inequality} \ts (for \defng{mixed volumes}) is a deep inequality
in convex geometry independently proved by Alexandrov (1938)
and Fenchel (1936), see e.g.~\cite[$\S$20]{BZ88} and \cite[$\S$7.3]{Sch14}.
We refer to~\cite{SvH19} for a notable recent proof and its popular exposition
in~\cite{CP22}.  A special case of the inequality for bricks is especially notable
as it led to a complete resolution of the long open
\defng{van der Waerden conjecture}~\cite{vL81,vL82},
which in turn led to further combinatorial inequalities, see e.g.~\cite{Alon03,Gur08,Sta81}.

\begin{theorem}[Alexandrov--Fenchel inequality for bricks]\label{t:AF}
Let \ts $n\ge 2$ \ts and let \ts $x_i, y_i, z_{ij}\ge 0$, for all \ts $1\le i, j \le n$.  Then:
{\small
$$
\per
\begin{pmatrix} x_{1} & y_{1} & z_{13} & \ldots & z_{1n} \\
x_{2} & y_{2} & z_{23} & \ldots & z_{2n} \\
\vdots &\vdots & \vdots & \ddots & \vdots \\
x_{n} & y_{n} & z_{n3} & \ldots & z_{nn}
\end{pmatrix}^2
\ge \
\per
\begin{pmatrix} x_{1} & x_{1} & z_{13} & \ldots & z_{1n} \\
x_{2} & x_{2} & z_{23} & \ldots & z_{2n} \\
\vdots &\vdots & \vdots & \ddots & \vdots \\
x_{n} & x_{n} & z_{n3} & \ldots & z_{nn}
\end{pmatrix} \,
\per
\begin{pmatrix} y_{1} & y_{1} & z_{13} & \ldots & z_{1n} \\
y_{2} & y_{2} & z_{23} & \ldots & z_{2n} \\
\vdots &\vdots & \vdots & \ddots & \vdots \\
y_{n} & y_{n} & z_{n3} & \ldots & z_{nn}
\end{pmatrix}.
$$
}
\end{theorem}

\smallskip

\begin{proposition}\label{p:AF-not}
Denote by \ts $AF_n$ \ts the polynomial in $k=n^2$ variables given in
Theorem~\ref{t:AF} by subtracting the right-hand side from the left-hand side of the inequality.  For $n\ge 2$, we have \. $AF_n \not\geqslant_\# 0$ \.
unless \. $\PH = \Sigma_2^\textup{\textsu{p}}$\ts.
\end{proposition}

\begin{proof}  Take \ts $n=2$, $y_1=y_2=1$.  Then the \ts $AF_2$ \ts inequality
becomes:
$$
\per\begin{pmatrix}
      x_1 & 1 \\
      x_2 & 1
    \end{pmatrix}^2 \,  = \,  (x_1+x_2)^2 \ \geqslant \ \per\begin{pmatrix}
      x_1 & x_1 \\
      x_2 & x_2
    \end{pmatrix} \per\begin{pmatrix}
      1 & 1 \\
      1 & 1
    \end{pmatrix} \, = \, 4\ts x_1x_2\ts.
$$
This is equivalent to \ts $(x_1-x_2)^2\geqslant 0$, and the result follows from Corollary~\ref{c:AMGM}.
\end{proof}

\smallskip

\subsection{Hadamard inequality}
Recall the \defn{Hadamard inequality}~\eqref{eq:Had-3x3} discussed in the introduction.
This inequality was proved by Hadamard~(1893) and is crucial is the study of
\defng{positive definite matrices}, see e.g.~\cite[$\S$7.8]{HJ13}.
Denote by \ts $H_d=H_d\bigl(x_{11},x_{12},\ldots,x_{dd}\bigr)$ \ts the
nonnegative polynomial in \ts $k=d^2$ \ts variables given by~\eqref{eq:Had-3x3}.
The following result is deduced from our Theorem~\ref{thm:relatclosure}.

\begin{proposition}[Hadamard inequality] \label{p:Had}
For \. $d\ge 3$, if  \.
$H_d \geqslant_\# 0$, then there is an oracle \ts $A$ \ts such that \. $H_3(\vv\sharpP^A)\not\subseteq\sharpP^A$.
\end{proposition}

\begin{proof}
Observe that
\begin{eqnarray*}
H_3\begin{pmatrix}
x & \binom{x}{3} & 0 \\
0 & 1 & 1 \\
1 & 0 & 1
\end{pmatrix} \ &=& \ \tfrac{1}{12} \ts x^6 \. - \. \tfrac{1}{2} \ts x^5 \. + \. \tfrac{3}{4} \ts x^4 \. + \. \tfrac{8}{3} \ts x^2 \\
\ &=& \ 
3 \binom{x}{1}
\. + \. 6 \binom{x}{2}
\. - \. 3 \binom{x}{3}
\. + \. 28 \binom{x}{4}
\. + \. 90 \binom{x}{5}
\. + \. 60 \binom{x}{6}\ts.
\end{eqnarray*}
Since the coefficient of \ts $\binom{x}{3}$ \ts is negative, by Theorem~\ref{thm:relatclosure}
this is not a univariate relativizing closure property of \. $\sharpP$.
Therefore, since $0$, $1$, $x$, and $\binom{x}{3}$ are univariate relativizing closure properties of $\sharpP$, it follows that $H_3$ is not a relativizing closure property of $\sharpP$.
\end{proof}

Note that if the (univariate) binomial basis conjecture (Conjecture~\ref{conj:binomialbasisuniv}) is true,
then this is not a closure property of \. $\sharpP$, and we could conclude that \. $H_3(\vv\sharpP)\not\subseteq\sharpP$.

\begin{remark}{\rm
One might think that plugging in arbitrary binomial coefficients in the Hadamard-matrix is a valid strategy, but this does not work as the following random choice illustrates:
\[
H_3\begin{pmatrix}
\binom{x}{4} & \binom{x}{7} & \binom{x}{12} \\[1ex]
\binom{x}{6} & \binom{x}{3} & \binom{x}{2} \\[1ex]
\binom{x}{8} & \binom{x}{11} & \binom{x}{5}
\end{pmatrix}
\]
is indeed a relativizing closure property of $\sharpP$.
}\end{remark}

\subsection{Fermat's little theorem}
As we mentioned in the introduction, \defn{Fermat's little theorem} \ts
states that \. $p\.| \. a^{p-1}-1$ \.  for all integers~$a$, $p \nmid a$,  and prime~$p$.
Fermat stated this result in~1640 without proof, and the
first published proof was given by Euler in~1736. According to Dickson,
``this is one of the fundamental theorems of the theory of numbers''
\cite[p.~V]{Dic52}.  Note, see e.g.\ in~\cite{CC16}, that Fermat's little
theorem can be rephrased to say that the polynomial \. $\tfrac 1 p (x^p-x)$ \.
is integer valued.

\begin{proposition}[Fermat's little theorem is in $\sharpP$]
\label{pro:fermat}
If \. $f \in \sharpP$ \. and $p$ is prime, then
\[
\tfrac 1 p (f^p - f) \ \in \sharpP.
\]
\end{proposition}

We include a variation on Peterson's original proof~\cite{Pet72} for completeness.

\begin{proof}
Consider sequences \ts $(a_1,\ldots,a_p)$ \ts of integers \. $1\le a_i \le f$ \.
and partition them into orbits under the natural cyclic action of \ts $\zz/p\zz$.
Since $p$ is prime, these orbits have either $1$ or $p$ elements.  There
are exactly $p$ orbits with one elements, where $a_1=\ldots=a_p$.  The remaining
orbits of size~$p$ have a total of \ts $f^p-f$ \ts elements.  Since $p$ is fixed,
the lex-smallest orbit representative can be found in polytime.
\end{proof}

\smallskip

This proof can be rephrased to show that the polynomial \. $\tfrac 1 p (x^p - x)$ \. is binomial-good.
For example, for $p=5$, we have:
$$
\frac{1}{5}\ts \bigl(x^5-x\bigr) \ = \ 24\binom{x}{5} \. + \. 48\binom{x}{4}
\. + \. 30\binom{x}{3} \. + \. 6\binom{x}{2}.
$$

\begin{remark}{\rm
Peterson also discovered a similar proof of \defng{Wilson's theorem}~\cite{Pet72},
see also~\cite[p.~50]{Car14}.   Stanley's elegant proof of the \defng{Lucas congruences}
is also combinatorial and in the same spirit~\cite[Exc.~1.15(c)]{Sta12}.
Let us also mention \defng{Kummer's congruences} for the
Catalan numbers, which can be proved via group action on binary trees,
see e.g.~\cite{DS06,KPP94}.  We refer to~\cite{Ges84} for a
survey of combinatorial congruences,
and to~\cite{RS80,Sag85} for the general group action approach.
See also~\cite{AZ17} for conjectures on whether there are combinatorial proofs
of various binomial congruences.

Finally, note that some congruences have combinatorial proofs for highly
nontrivial reasons.  Recall the \defng{Ramanujan's congruence} \. $5 \ts |\ts p(5n-1)$ \.
for the number of integer partitions (see e.g.~\cite[$\S$6.4]{Har40}).  This
congruence was famously interpreted by Dyson~\cite{Dys44}, by dividing partitions of \ts
$(5n-1)$ \ts according to its \defn{rank} (first row minus first column) modulo~$5$.
It remains open to find a direct bijective
proof of this equal division, see~\cite[$\S$2.5.6]{Pak06}. We refer to
\cite{AG88,GKS90} for some remarkable generalizations of this approach.
}\end{remark}

\smallskip

\subsection{Ahlswede--Daykin inequality}
\label{sec:ahlswededaykinkleitman}
Let \ts $\cA\subseteq 2^{[n]}$ \ts and \ts $\zeta: 2^{[n]}\to \rr$.  Denote
$$
\zeta(\cA) \, := \, \sum_{A\in \cA} \. \zeta(A).
$$
For \ts $\cA,\cB\subseteq 2^{[n]}$,  denote
$$
\cA \. \bcup \. \cB \, := \, \bigl\{A\cup B \, : \, A\in \cA, \. B \in \cB\bigr\}, \qquad
\cA \. \bcap \. \cB \, := \, \bigl\{A\cap B \, : \, A\in \cA, \. B \in \cB\bigr\}.
$$

\smallskip

\begin{theorem}[{\defn{Ahlswede--Daykin inequality}~\cite{AS16}}]\label{t:AD}
Let \ts $\al,\be,\ga,\de: 2^{[n]} \to \rr_+$ \ts
be such that
\begin{equation}\label{eq:AD-elts}
\al(A) \. \be(B) \, \le \, \ga(A\cap B) \. \de(A \cup B) \., \qquad \forall \. A, \ts B\in \ts 2^{[n]}.
\end{equation}
Then
\begin{equation}\label{eq:AD-sets}
\al(\cA) \. \be(\cB) \, \le \, \ga\bigl(\cA\.\bcap \.\cB\bigr) \, \de\bigl(\cA \.\bcup\. \cB\bigr) \.,
\qquad \forall \. \cA, \. \cB \.\subseteq \. 2^{[n]}.
\end{equation}
\end{theorem}

\smallskip

This inequality is classical, and is an advanced generalization of the
Kleitman inequality, see below.  In its most general form it is usually
stated for general lattices, not just the Boolean lattice.  Among its many
applications, let us single out the \defng{FKG inequality} \cite[$\S$6.2]{AS16}
and the \defng{XYZ inequality} \cite[$\S$6.4]{AS16}, see also~$\S$\ref{ss:open-LE}.

\smallskip

Let \ts $n=1$, and let \ts $\cA=\cB=\{0,1\}$.  Operations \ts $\cap$ \ts and \ts
$\cup$ \ts are replaced with \ts $\min$ \ts and \ts $\max$ \ts
in this case.
Let \ts $\al_i, \be_i, \ga_i,\de_i$, \ts $i\in \{0,1\}$,
\ts be $\SP$ functions satisfying~\eqref{eq:AD-elts}.
This is guaranteed by taking functions \ts $h_1,h_2,h_3,h_4 \in \SP$, such that
$$
\al_0 \be_0 \ts +\ts h_1 \. = \. \ga_0\de_0, \quad \al_0\be_1\ts +\ts h_2 \. = \. \ga_0\de_1,
\quad \al_1\be_0\ts +\ts h_3 \. = \. \ga_0\de_1, \quad \al_1\be_1\ts +\ts h_4 \. = \. \ga_1\de_1
$$
Let \ts $AD_n$ \ts be the function defined by~\eqref{eq:AD-sets}.
Then we have:
$$
AD_1 \ = \ (\ga_0+\ga_1)(\de_0+\de_1) \. - \. (\al_0+\al_1)(\be_0+\be_1).
$$
Now Proposition~\ref{pro:ahlswededaykin} gives the oracle separation for this
inequality.

\smallskip

\subsection{Karamata inequality}
\label{sec:karamata}
Let \ts $\bbx=(x_1,\ldots,x_n)$, $\bby=(y_1,\ldots,y_n)\in \rr^n$ \ts be nonincreasing 
sequences of real numbers.  We say that \ts $\bbx$ \ts \defn{majorizes}~$\ts\bby$,
write \ts $\bbx \tre \bby$, if
$$\aligned
& x_1 \. + \. \ldots \. + \. x_i \,\ge \, y_1 \. + \. \ldots \. + \. y_i \quad \text{for all} \ \. 1\le i < n\ts, \ \ \text{and}\\
& x_1 \. + \. \ldots \. + \. x_n \,= \, y_1 \. + \. \ldots \. + \. y_n\..
\endaligned
$$
In combinatorial context, this is also called the \defng{dominance order}, and
appears throughout the area, see e.g.\ \cite{Bru06,Mac95,Sta12}.
See also \cite{Bar07} for a recent connection to the problem of
counting \defng{contingency tables}.

\smallskip

\begin{theorem}[{\em\defn{Karamata inequality}}{}]
Let \ts $\bbx,\bby\in \rr^n$, such that $\bbx \tre \bby$.  Then,
for every convex function \ts $F: \rr^n \to \rr$, we have \ts $f(x) \ge f(y)$.
\end{theorem}

\smallskip

This result is classical, see e.g.~\cite[$\S$3.17]{HLP52}
and~\cite[$\S$28, $\S$30]{BB61}.  Analytically, the inequality can be used
to derive \defng{Jensen's inequality}, which in turn implies the
AM-GM inequality.  See \cite{BP21,PPS20} for some recent
applications of the Karamata inequality to combinatorial
problems on linear extensions and Young tableaux, respectively.
See also \cite{MOA11} for a modern proof, numerous generalizations
and further references.

\smallskip

We now convert the Karamata inequality into a counting function problem.
Suppose we are given \ts $f_i, g_i \in \sharpP$, $1 \leq i \leq n$,
such that the following functions $h_i$, $1 \leq i < n$, are also all in $\sharpP$:
\begin{equation}\label{eq:karamata-h}
h_i \, := \, f_1 \. + \. \ldots \. + \. f_i \, - \, g_1 \. - \. \ldots \. - \.  g_i
\end{equation}
and we are also guaranteed that
\begin{equation}\label{eq:karamatalast}
f_1\. + \. \ldots \. + \. f_n \, - \, g_1 - \ldots - g_n \, = \, 0.
\end{equation}
Moreover, the functions
\begin{equation}
\label{eq:karamatade}
\textup{
$d_i \. := \. f_i \ts - \ts f_{i+1}$ \, and \, $e_i \. := \. g_i \ts - \ts g_{i+1}$
}
\end{equation}
are also in \ts $\sharpP$ \ts for all \ts $1 \leq i < n$.
Let \. $Z \subseteq \IQ^{5n-3}$ \. denote the variety of points that satisfy the constraints \eqref{eq:karamatalast} and \eqref{eq:karamatade}.
Let \. $\gamma\in\GapPP$ \. be any convex function.

Define the \defn{Karamata function} as
\[
K_{n,\gamma}(\vv f,\vv g) \ := \ \sum_{i=1}^{n} \. \gamma(f_i) \, - \, \sum_{i=1}^{n} \. \gamma(g_i)\ts.
\]
Clearly \. $K_{n,\gamma}(\vv \sharpP) \subseteq \GapP$.
Karamata's inequality implies that the answer is always nonnegative for inputs from~$Z$.
This is a semantic guarantee, i.e., we have no information as to why the guarantee holds.
This is analogous to \ts $\sharpP_{\geq 1}$, but in contrast to problems in Section~\ref{sec:TFNP}.

We write \ts $\vv \sharpP_{\in Z}$ \ts for a tuple of \ts $\sharpP$ \ts
functions whose function values always lie on~$Z$.
Hence, \ts $K_{n,\gamma}(\vv{\sharpP}_{\in Z})\subseteq\GapP_{\geq 0}$\ts.
Consider the following examples of \ts $n,\gamma$ \ts
for which \. $K_{n,\gamma}(\vv{\sharpP}_{\in Z})\subseteq\sharpP$~:
\begin{itemize}
\item For affine linear \ts $\gamma$, we clearly have \. $K_{n,\gamma}=0\in\sharpP$.
\item For \ts $\gamma(t)=t^2$, we have \. $K_{2,\gamma}(f_1,f_2,g_1,g_2) = (d_1 + e_1 )h_1$ \.
as a function on~$Z$. This can be seen by plugging in \ts $d_1 = f_1-f_2$, \ts $e_1 = g_1-g_2$, and \ts
$g_2=f_1+f_2-g_1$. Clearly \. $(d_1 + e_1 )h_1 \in \sharpP$.
This has several proofs, for example instead of \ts $(d_1 + e_1 )h_1$ \ts
we could have taken \. $2 h_1 + 2 e_1 h_1 + 4\binom{h_1}{2}$ \. with the same argument.
\item For \ts $\gamma(t)=t^2$, we have \. $K_{3,\gamma}(f_1,f_2,f_3,g_1,g_2,g_3) = (d_1 +e_1) h_1 + (d_2+e_2) h_2 \in \sharpP$
\. on~$Z$.
\item For \ts $\gamma(t)=\binom{t}{2}$, we have \.
$K_{2,\gamma}(f_1,f_2,g_1,g_2) = (e_1+1) h_1 + 2\binom{h_1}{2} \in \sharpP$ \.
on~$Z$.
\item For \ts $\gamma(t)=\binom{t}{2}$, we observe that for the \emph{double} we have \ts
$2 K_{3,\gamma}\in\sharpP$ \ts via the observation that \ts $2 K_{3,\binom{t}{2}} = K_{3,t^2}$ \ts
on~$Z$ (the affine linear parts cancel out).
\end{itemize}
All inclusions $K_{n,\gamma}\subseteq\sharpP$ in this section so far relativize.
The next proposition shows that the doubling we just used was in fact necessary, because otherwise obtain an oracle separation.

\smallskip

\begin{proposition}
\label{pro:karamatadiagonalization}
There exists a language $A\subseteq\{0,1\}^*$ such that
$K_{3,\binom{t}{2}} (\vv{\sharpP}^A_{\in Z}) \not\subseteq \sharpP^A$.
\end{proposition}

\begin{proof}
We use the Diagonalization Theorem~\ref{thm:diagonalization}.
We have $5n-3 = 12$, so in the notation of the theorem, we set $S = Z\cap \IN^{12}$.
We fix an arbitrary order of the 12 variables: $(f_1,f_2,f_3,g_1,g_2,g_3,d_1,d_2,e_1,e_2,h_1,h_2)$.
The variety $Z$ is then given as the kernel of the linear map given by the following matrix:
\[
\begin{pmatrix}
  1& -1&  0&  0&  0&  0& -1&  0&  0&  0&  0&  0\\
  0&  1& -1&  0&  0&  0&  0& -1&  0&  0&  0&  0\\
  0&  0&  0&  1& -1&  0&  0&  0& -1&  0&  0&  0\\
  0&  0&  0&  0&  1& -1&  0&  0&  0& -1&  0&  0\\
  1&  0&  0& -1&  0&  0&  0&  0&  0&  0& -1&  0\\
  1&  1&  0& -1& -1&  0&  0&  0&  0&  0&  0& -1\\
  1&  1&  1& -1& -1& -1&  0&  0&  0&  0&  0&  0
\end{pmatrix}
\]
We convert it to row echelon form:
\[
\begin{pmatrix}
  1 &  0 &  0 &  0 &  0 & -1 &  0 &  0 & -1 & -1 & -1 &  0\\
  0 &  1 &  0 &  0 &  0 & -1 &  0 &  0 &  0 & -1 &  1 & -1\\
  0 &  0 &  1 &  0 &  0 & -1 &  0 &  0 &  0 &  0 &  0 &  1\\
  0 &  0 &  0 &  1 &  0 & -1 &  0 &  0 & -1 & -1 &  0 &  0\\
  0 &  0 &  0 &  0 &  1 & -1 &  0 &  0 &  0 & -1 &  0 &  0\\
  0 &  0 &  0 &  0 &  0 &  0 &  1 &  0 & -1 &  0 & -2 &  1\\
  0 &  0 &  0 &  0 &  0 &  0 &  0 &  1 &  0 & -1 &  1 & -2
\end{pmatrix}
\]
We set \ts $\ka=5$, \ts $k=12$. 
Permuting the order of the columns to \. $(6,9,10,11,12,1,2,3,4,5,7,8)$,
we obtain affine linear functions \ts $\zeta_8,\ldots,\zeta_{12}$, where each \ts $\zeta_b$ \ts
depends linearly on the first $5$ variables, which we call \. $(v_1,\ldots,v_\ka) := (g_3,e_1,e_2,h_1,h_2)$.
We verify the assumptions of the Diagonalization Theorem~\ref{thm:diagonalization}:
\begin{enumerate}
\item Since all constraints are homogeneous, we have zero is an integer point: \ts $0\in Z$.
\item All $\zeta$ have integer coefficients and all constraints are homogeneous, so the fact that
\ts $C'_S$ \ts lies Zariski-dense in $\IQ^{5}$ follows from the next point
and having a small open ball that contains~$\vv v$.
\item The point \. $\vv v = (g_3=2,e_1=2,e_2=1,h_1=1,h_2=1)$ \. satisfies all \ts $\zeta_b(\vv v)>0$.
Here \ts $\vv v$ \ts is the point corresponding to the case \ts $\vv f = (6,3,1)$ \ts and \ts $\vv g =(5,3,2)$.
\end{enumerate}
We set \ts $\problem{Multiplicities} = {\problem{OccurrenceMulti}_{12}}$ \ts
and use the set-instantiators from Theorem~\ref{thm:occurrencesetinstantiator}.
We fix \ts $\vv t \in S$, and set
\[
\varphi(f_1,f_2,f_3,g_1,g_2,g_3,d_1,d_2,e_1,e_2,h_1,h_2) \ := \ f_1^2 \. + \. f_2^2 \. + \.
f_3^2 \. - \. g_1^2 \. - \. g_2^2 \. - \. g_3^2\ts.
\]
We now show that \ts $\varphi+I$ \ts is binomial-bad.
We use the Polyhedron Theorem~\ref{thm:polyhedron}, since all our constraints are affine linear.
Since \. $\dim[\IQ^{12}]_{\leq 2} = \binom{14}{2} = 91$ \.
and \. $\dim[\IQ^5]_{\leq 2} = \binom{7}{2} = 21$,
this is a polyhedron in \ts $\IQ^{91}$ \ts with~$21$ linear equations
intersected with the nonnegative orthant.  We use a computer to set up the polyhedron.
Indeed, it contains the half-integer point that shows that \. $2K_{3,\binom{t}{2}}\in\sharpP$,
but it does \emph{not} contain an integer point, which
gives a function \ts $p_A$ \ts with \. $\varphi(p_A)\notin\sharpP^A$.

It remains to show that \. $\varphi(p_A) \in K_{3,\binom{t}{2}} \big(\vv{\sharpP}^A_{\in Z}\big)$.
This can be seen by the existence of the list of oracle Turing machines
that on input $w$ first query $A$ at \ts $0^{|A|}$.
Then, if \ts $0^{|w|}\notin A$ \ts accept with multiplicities $\vv t$; if \ts $0^{|w|}\in A$,
then the machines run \ts $\problem{OccurrenceMulti}$ \ts on \ts $\tilde A_j$.
The fact that the output vector is in $Z$ is guaranteed by the fact that
if \ts $0^{|w|}\in A$, then \ts $\tilde A_j$ \ts is generated by the set-instantiator,
which only generates instances from~$S$.
\end{proof}

\smallskip

\begin{remark}\label{r:karamata-h}{\rm
One can ask if the assumptions underlying Proposition~\ref{pro:karamatadiagonalization} are reasonable,
e.g.\ whether there is a natural combinatorial problem where we have \ts $h_i\in \SP$ \ts for all functions
defined in~\eqref{eq:karamata-h}.  For an example of this, we refer to above mentioned result
in~\cite{PPS20}, where the majorization is proved by a direct injection.   In the case of Young diagrams,
the ``shuffling in the plane'' proof in~\cite[$\S$4]{PPS20} is based on the
\defng{wall-equivalence}, see e.g.~\cite[$\S$23]{Mat10}.

We should warn the reader that some applications of the Karamata inequality can in fact be in~$\SP$.
For example, the \defng{hook inequality} \ts for the number of increasing trees is
proved in~\cite[$\S$2]{PPS20} via the Karamata inequality, but is a special
case of the \defng{Bj\"orner--Wachs inequality} (Theorem~\ref{t:BW-ineq}) known to be in~$\SP$.
}
\end{remark}

\medskip

\section{{\normalfont$\TFNP$} and {\normalfont$\sharpP$}}
\label{sec:TFNP}

\subsection{Definitions and background} \label{ss:TFNP-Background}
%
Recall the classical inclusion diagram of search complexity classes:
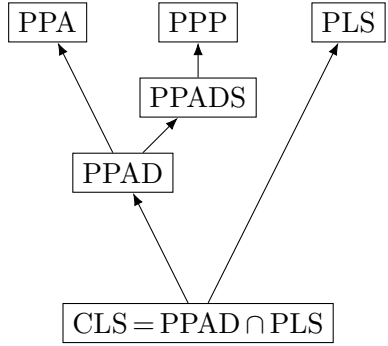
\begin{center}
\begin{tikzpicture}
\node[draw] (PPA) at (0,0) {PPA};
\node[draw] (PPP) at (2,0) {PPP};
\node[draw] (PLS) at (4,0) {PLS};
\node[draw] (PPADS) at (2,-1) {PPADS};
\node[draw] (PPAD) at (1,-2) {PPAD};
\node[draw] (CLS) at (2,-4) {CLS\,=\,PPAD\,$\cap$\,PLS};
\draw[-Latex] (CLS) -- (PPAD);
\draw[-Latex] (CLS) -- (PLS);
\draw[-Latex] (PPAD) -- (PPADS);
\draw[-Latex] (PPADS) -- (PPP);
\draw[-Latex] (PPAD) -- (PPA);
\end{tikzpicture}
\end{center}

When going to the counting analogs minus~1, as explained in~$\S$\ref{sec:countingclassesandTFNPINTRO}, one gets the same picture (the right-hand side of the tree in Figure~\ref{fig:inclusions}), but care must be taken with the complete problems, since for several of the classical problems in $\CLS$, $\PPAD$, and $\PPADS$ the counting version minus~$1$ actually lies in $\sharpP$.

We study the following classical problems and slightly adjusted problems that are not parsimoniously equivalent.
We assume that the reader is familiar with how to encode an exponential graph or digraph via successor or predecessor/successor circuits, see \cite{GW83, MP91}.

\begin{itemize}
\item The $\PPAD$-complete problem \problem{SourceOrSink}, which is parsimoniously equivalent to \problem{Sperner}, see \cite{CD09}.
We are given two circuits \ts $C_{\textup{succ}}$ \ts and \ts $C_{\textup{pred}}$ \ts that describe a directed graph in which there is an edge from $x$ to~$y$ \ts if and only if \ts
$C_{\textup{succ}}(x)=y$ \ts and \ts $C_{\textup{pred}}(y)=x$. We syntactically ensure that the indegree of the $0$ vertex is~$0$ and its outdegree is~$1$. We search for sources or sinks, i.e., nonzero vertices of \ts (indegree,outdegree$)=(0,1)$ \ts or \ts (indegree,outdegree$)=(1,0)$.
\item The $\PPAD$-complete problem \ts \problem{SourceOrPresink} (see Claim~\ref{cla:SourceOrPresink}).
The setup is the same as for \ts \problem{SourceOrSink}, but we count sources and presinks,
where a presink is a vertex that is adjacent to a sink.
Note that the two-vertex graph with a single source and sink only counts once,
as it has only one node that is a source or a presink (it is actually both a source and a presink).
\item The $\PPAD$-complete problem \ts \problem{SourceOrExcess(2,1)} (see Claim~\ref{cla:SourceOrExcess}).
We are given three circuits \ts $C_{\textup{succ}}$ \ts and \ts $C_{\textup{pred}_1}$ \ts and \ts $C_{\textup{pred}_2}$ \ts
that describe a directed graph in which there is an edge from $x$ to $y$ \ts if and only if \ts
$C_{\textup{succ}}(x)=y$ \ts and \ts $\big(C_{\textup{pred}_1}(y)=x \. \vee \. C_{\textup{pred}_2}(y)=x\big)$. This results in a digraph with indegrees from $\{0,1,2\}$ and outdegrees from $\{0,1\}$.
As for \problem{SourceOrSink} we syntactically ensure that the indegree of the 0 vertex is 0 and its outdegree is 1.
We search for sources or excess vertices, i.e., nonzero vertices where the indegree differs from the outdegree. It is crucial here that these digraphs can have double sinks, i.e., vertices with indegree 0 and outdegree 2, that we only count once.
\item Instead of being given a source, we can also count \problem{AllSourcesOrSinks}.
This is mainly interesting, because the number is always even, so we can ask if \.
$\problem{AllSourcesOrSinks}/2 \in \sharpP$ \ts (which it is, but for the undirected analog we have an oracle separation from $\sharpP$, see Theorem~\ref{thm:halvingseparation}).
We call the corresponding counting class \. $\sharpCOUNTALL[-PPAD]\problemp{SourceOrSink}$.
\item The $\PPADS$-complete problem \ts \problem{Sink}.
The setup is the same as for \ts \problem{SourceOrSink}, but we only count sinks.
\item The $\PPADS$-complete problem \ts \problem{Presink} (see Claim~\ref{cla:PPADSpresink}).
The setup is the same as for \ts \problem{Sink}, but we only count presinks,
which are vertices that are adjacent to sinks.
\item The $\PPADS$-complete problem \ts \problem{Excess(2,1)} (see Claim~\ref{cla:PPADSexcess}).
The setup is the same as for \ts \problem{SourceOrExcess(2,1)},
but we only count nodes with indegree greater than outdegree.
\item The $\PPA$-complete problem \problem{Leaf}, which is parsimoniously equivalent to \problem{Lonely} and \problem{Even} and \problem{Odd}, see~\cite{BCEIP}. We are given two circuits $C_{1}$ and $C_{2}$ that describe a graph in which an edge between $x$ and $y$ is present \ts if and only if \ts
    $$
    \big(C_1(x)=y \.\vee \. C_2(x)=y\big) \. \wedge \. \big(C_1(y)=x \. \vee \. C_2(y)=x\big).
    $$
    We syntactically ensure that the degree of the 0 vertex is~1. We search for leaves, i.e., nonzero vertices of degree~1.
\item
The hardness of \. $\big(\problem{Leaf}-1\big)/2$ \. actually comes from the hardness of the easier problem \. $(\problem{AllLeaves})/2$. In the problem \ts \problem{AllLeaves}
the setup is the same as for \problem{Leaf}, but we have no syntactic guarantee about the zero vertex. The number of solutions is always even, and it can be zero.
We call the corresponding counting class \.
$\sharpCOUNTALL[-PPA]\problemp{Leaf}$.
\item The $\PPA$-complete problem \ts \problem{Preleaf} \ts (see Claim~\ref{cla:preleaf}), which has the same setup as \ts \problem{Leaf}, but we search for vertices that are adjacent to leaves. Note that the line graph with~3 vertices has~2 leaves, but only~1 preleaf.
\item A slightly adjusted $\PPP$-complete version of the problem \ts \problem{Pigeon} \ts (see Claim~\ref{cla:PPPpigeno}), where we are given a circuit $C$ and search for vertices $x$ with \ts $C(x)=0$, or for pairs of vertices \ts $(x,y)$ \ts
    with \. $C(x)=C(y)$ \ts and \ts $C(C(x))\neq 0$. This version gives a cleaner counting problem.
\item The classical \PLS-complete problem \ts \problem{Iter} \ts which is parsimoniously equivalent to \ts \problem{LocalOpt} (see e.g.~\cite{FGHS21}), where we are given a circuit $C$ that has the syntactically ensured guarantee that $C(x)\geq x$ and $C(0)>0$ and $C(C(0))> C(0)$;
and we search for presinks, i.e., for vertices \ts $x$ \ts such that \ts $C(x)\neq x$ \ts and \ts $C(C(x))=x$.\footnote{Note that if given only a successor circuit, it is difficult to check if a vertex is a sink, so we count presinks instead. Also note that allowing 0 to be a presink would also have been an option.}
It is important to note that \problem{Iter} differs significantly from the other problems: its instances cannot be freely permuted. This makes the construction of a set-instantiator in~$\S$\ref{sec:IterSetInstantiator} quite different from the other set-instantiator constructions.
\item A slightly adjusted version of the \CLS-complete problem
$$
\text{\problem{EitherSolution(SourceOrSink,\ts{}Iter)}}
$$
(see Claim~\ref{cla:CLSsourceorsink}), where
in the counting version
we are given a pair of a \problem{SourceOrSink} and an \problem{Iter} instance and count each solution to either of them, with a slight adjustment:
If the \problem{SourceOrSink} instance has a source or sink at the last possible position, then we say that the instance contains the \emph{last option}.
If the \problem{Iter} instance has a presink at the last possible position (i.e., the second to last vertex), then we say that the instance contains the last option.
Note that we can efficiently check if either instance contains the last option.
If both instances contain the last option, then we count these two solutions only once.
This adjustment is to ensure that it is possible to have a single solution.
Otherwise the counting version would be in \ts $\sharpP_{\geq 2}$ \footnote{\cite{Goo22} give a combinatorial version of \CLS, which was pointed out to us by M.~G\"o\"os after this paper was written.}.
\item The analogous \CLS-complete problem \. \problem{EitherSolution(SourceOrPresink,\ts{}Iter)} \ts (see Claim~\ref{cla:CLSsourceorsink}).
\item The analogous \CLS-complete problem \. \problem{EitherSolution(SourceOrExcess(2,1),\ts{}Iter)} \ts (ibid.)
\item The problem \. \problem{BipartiteUnbalance} \.
has no associated search problem. It is the following problem in \ts $\GapP_{\geq 0}$.
Let \ts $G=(V,E)$ \ts be a bipartite graph with
two parts given by \ts $V=V_-\sqcup V_+$.
We say that $G$ is \defn{unbalanced}
if for every \ts $(uv)\in E$, $u\in V_-$, $v\in V_+$, we have \ts $\deg(u)\ge \deg(v)$.
An instance of the problem is given by a list of polynomially many circuits \ts $C_1,\ldots,C_n$ \ts
such that an edge between \ts $u\in V_-$ \ts and \ts $v\in V_+$ \ts exists \ts if and only if \ts $C_i(u)=v$ \ts for some~$i$,
and \ts $C_j(v)=u$ \ts for some~$j$.
We syntactically ensure that \ts $\deg(u)\ge \deg(v)$ \ts locally by adding vertices from \ts $V_-$ \ts
whenever \. $\deg(u) < \max_{v \in N(u)}\deg(v)$, where \ts $N(u)$ \ts is the neighborhood of~$u$:
In this case we add \. $\max_{v \in N(u)}-\deg(u)$ \. many vertices in \ts $V_-$ \ts
to the graph and connect them to~$u$, but do not connect them to any other vertex.
Proposition~\ref{p:flow-graph} below shows that \ts $|V_+|-|V_-| \geq 0$.
The problem \ts \problem{BipartiteUnbalance} \ts is the counting problem with value \ts
$|V_+|-|V_-|$.
We study this problem in~$\S$\ref{sec:unbalancedflowseparation}, cf.\ Open Problem~\ref{ss:open-LE}(3).
\end{itemize}

\subsection{Simple completeness results, equalities to {\normalfont$\sharpP$}, and simple inclusions}\label{sec:syntacticcountingsubclasses}
All results in this section are fairly straightforward.
They appear here to avoid any oversights, because Figure~\ref{fig:inclusions} suggests that $\sharpP$ is
strictly contained in classes that differ only slightly in terms of definition.

\subsubsection{{\normalfont\#\PPAD}}
\begin{claim}\label{cla:sourceorsinkequalssharpP}
$(\#\PPAD\problemp{SourceOrSink}-1)/2 = \sharpP$ \. via relativizing parsimonious reductions.
\end{claim}

\begin{proof}
We first note that \.
$(\#\PPAD\problemp{SourceOrSink}-1)/2 \subseteq \sharpP$, because we can just count the nonzero sources, which gives us the correct number.
For the other direction,
we use that \ts $\problem{CircuitSat}^A$ \ts is $\sharpP^A$-complete.
If we are given a Boolean single-output circuit $C$ (with oracle gates) that has $c$ many $x$ for which $C(x)$ is true,
then we can simply construct a \ts \problem{SourceOrSink} instance with the zero vertex immediately going into a sink,
the rest of the vertices all being self-loops with one exception:
for every position \ts $2 x$ \ts for which \ts $C(x)$ \ts is true we add a source-sink edge from \ts
$2x$ \ts to \ts $2x+1$.
We end up with an instance of value \ts $2c+1$, as desired.
\end{proof}

\smallskip

\begin{claim}\label{cla:sourceorpresinkequalssharpP}
$\#\PPAD\problemp{SourceOrPresink}-1 = \sharpP$ \. via relativizing parsimonious reductions.
\end{claim}

\begin{proof}
We have $s$ nonzero sources, $t$ presinks, and $a$ nodes that are both, which we call amalgamations.
We count \ts $s+t+a-1$ = $2 s + a$, i.e., we count the nonzero sources twice and the amalgamations once.
The second direction is the same construction as for Claim~\ref{cla:sourceorsinkequalssharpP},
but here the instance ends up with value \ts $c+1$ \ts instead of \ts $2c+1$, because the source is also a presink.
\end{proof}

\smallskip

\begin{claim}\label{cla:countallsourceorpresinkhalfequalssharpP}
$\sharpCOUNTALL[-PPAD]\problemp{SourceOrSink}/2 = \sharpP$ \. via relativizing parsimonious reductions.
\end{claim}

\begin{proof}
To see the containment in $\sharpP$, we just count only the sources (or only the sinks).
The second direction is the same construction as in Claim~\ref{cla:sourceorpresinkequalssharpP}, just without the source-sink pair for the zero vertex.
\end{proof}

\smallskip

\begin{claim}\label{cla:SourceOrPresink}
\problem{SourceOrPresink} \ts is \PPAD-complete.
\end{claim}

\begin{proof}
This follows directly from the fact that if we find a sink, then it is easy to find a presink,
and vice versa, and the fact that  \ts \problem{SourceOrSink} \ts is \PPAD-complete.
\end{proof}

\smallskip

\begin{claim}
\label{cla:SourceOrExcess}
\problem{SourceOrExcess(2,1)} \ts is \ts \PPAD-complete.
\end{claim}

\begin{proof}
Clearly, every \ts \problem{SourceOrSink} \ts instance is a \ts \problem{SourceOrExcess(2,1)} \ts instance.
Given a \ts \problem{SourceOrExcess(2,1)} \ts instance, we replace every double sink with two sinks, and every indegree~2,
outdegree~1 vertex with a sink and an indegree~1, outdegree~1 vertex. Note that we replaced only excess vertices and added~1
or~2 sinks each time.
Since the construction was local, a solution to this \ts \problem{SourceOrSink} \ts instance can be converted back into a solution of the \ts \problem{SourceOrExcess(2,1)} \ts instance.
\end{proof}

\smallskip

\begin{claim}
\label{cla:PPADinPPADS}
\begin{eqnarray*}
\#\PPAD\problemp{SourceOrExcess(2,1)}-1 \,
\subseteq \, \#\PPADS\problemp{Excess(2,1)}-1 \end{eqnarray*}
via relativizing parsimonious reductions.
\end{claim}
\begin{proof}
Given a \ts $\problem{SourceOrExcess(2,1)}$ \ts instance we create an \ts $\problem{Excess(2,1)}$ \ts instance of the same value by adding a (source,\ts{}sink) pair that is not connected to the rest of the instance, for every source vertex in the input.
\end{proof}

\smallskip

\begin{claim}
\label{cla:PPADinPPA}
\begin{eqnarray*}
\#\PPAD\problemp{SourceOrExcess(2,1)}-1
\, \subseteq \, \#\PPA\problemp{Preleaf}-1 \end{eqnarray*}
via relativizing parsimonious reductions.
\end{claim}
\begin{proof}
Given a \ts \problem{SourceOrExcess(2,1)} \ts instance, we make local graph replacements with gadgets, making sure that no vertex that connects to another gadget is a leaf, so that we know exactly which vertices of a gadget are preleaves. Since the indegree is an element from \ts $\{0,1,2\}$ \ts and the outdegree is an element from \ts $\{0,1\}$, there are \. $2\cdot 3 = 6$ \. possible vertex types in the input \ts \problem{SourceOrExcess(2,1)} \ts  instance.

The transformation table below shows the gadgets that are used to replace the parts in the \ts \problem{SourceOrExcess(2,1)} \ts instance.
Here the sources and excess nodes are marked gray on the left-hand side, preleaves are marked gray on the right-hand side.
The proof follows from the fact that in each row the number of gray vertices on the left equals the number of gray vertices on the right.
\end{proof}

\begin{center}
\bgroup
\def\arraystretch{2.5}
\begin{tabular}{c|c}
\begin{tikzpicture}[every node/.style={draw,circle,inner sep=0pt,minimum size=0.3cm}]
\node (v) at (0,0) {};
\end{tikzpicture}
&
\begin{tikzpicture}[every node/.style={draw,circle,inner sep=0pt,minimum size=0.3cm}]
\node (v) at (0,0) {};
\end{tikzpicture}
\\
\hline
\begin{tikzpicture}[every node/.style={draw,circle,inner sep=0pt,minimum size=0.3cm}]
\node[fill=gray] (v) at (0,0) {};
\draw[-Latex] ($(v)+(-1,0)$) -- (v);
\end{tikzpicture}
&
\begin{tikzpicture}[every node/.style={draw,circle,inner sep=0pt,minimum size=0.3cm}]
\node[fill=gray] (v) at (0,0) {};
\node (w) at (1,0) {};
\draw ($(v)+(-1,0)$) -- (v);
\draw (w) -- (v);
\end{tikzpicture}
\\
\hline
\begin{tikzpicture}[every node/.style={draw,circle,inner sep=0pt,minimum size=0.3cm}]
\node[fill=gray] (v) at (0,0) {};
\draw[-Latex] ($(v)+(-1,0.3)$) -- (v);
\draw[-Latex] ($(v)+(-1,-0.3)$) -- (v);
\end{tikzpicture}
&
\begin{tikzpicture}[every node/.style={draw,circle,inner sep=0pt,minimum size=0.3cm}]
\node (v) at (0,0) {};
\node (w1) at (1,0) {};
\node[fill=gray] (w2) at (2,0) {};
\node (w3) at (3,0) {};
\draw ($(v)+(-1,0.3)$) -- (v);
\draw ($(v)+(-1,-0.3)$) -- (v);
\draw (w1) -- (w2);
\draw (w2) -- (w3);
\end{tikzpicture}
\\
\hline
\begin{tikzpicture}[every node/.style={draw,circle,inner sep=0pt,minimum size=0.3cm}]
\node[fill=gray] (v) at (0,0) {};
\draw[-Latex] (v) -- ($(v)+(+1,0)$);
\end{tikzpicture}
&
\begin{tikzpicture}[every node/.style={draw,circle,inner sep=0pt,minimum size=0.3cm}]
\node[fill=gray] (v) at (0,0) {};
\node (w) at (-1,0) {};
\draw ($(v)+(+1,0)$) -- (v);
\draw (w) -- (v);
\end{tikzpicture}
\\
\hline
\begin{tikzpicture}[every node/.style={draw,circle,inner sep=0pt,minimum size=0.3cm}]
\node (v) at (0,0) {};
\draw[-Latex] (v) -- ($(v)+(+1,0)$);
\draw[-Latex] ($(v)+(-1,0)$) -- (v);
\end{tikzpicture}
&
\begin{tikzpicture}[every node/.style={draw,circle,inner sep=0pt,minimum size=0.3cm}]
\node (v) at (0,0) {};
\draw ($(v)+(+1,0)$) -- (v);
\draw ($(v)+(-1,0)$) -- (v);
\end{tikzpicture}
\\
\hline
\begin{tikzpicture}[every node/.style={draw,circle,inner sep=0pt,minimum size=0.3cm}]
\node[fill=gray] (v) at (0,0) {};
\draw[-Latex] (v) -- ($(v)+(+1,0)$);
\draw[-Latex] ($(v)+(-1,0.3)$) -- (v);
\draw[-Latex] ($(v)+(-1,-0.3)$) -- (v);
\end{tikzpicture}
&
\begin{tikzpicture}[every node/.style={draw,circle,inner sep=0pt,minimum size=0.3cm}]
\node (v) at (0,0) {};
\node (w1) at (1,0) {};
\node[fill=gray] (w2) at (2,0) {};
\draw ($(v)+(-1,0.3)$) -- (v);
\draw ($(v)+(-1,-0.3)$) -- (v);
\draw ($(w2)+(1,0)$) -- (w2);
\draw (w1) -- (w2);
\end{tikzpicture}
\end{tabular}
\egroup
\end{center}

\medskip

\subsubsection{{\normalfont\#\PPADS}}

\begin{claim}\label{cla:PPADSpresink}
\problem{Presink} \ts is \ts \PPADS-complete.
\end{claim}

\begin{proof}
This follows directly from the fact that if we find a sink, then it is easy to find a presink, and vice versa.
\end{proof}

\smallskip

\begin{claim}\label{cla:PPADSexcess}
\problem{Excess(2,1)} \ts is \ts \PPADS-complete.
\end{claim}
\begin{proof}
Clearly, every \problem{Sink} instance is a \problem{Excess(2,1)} instance.
Given an \problem{Excess(2,1)} instance, we replace every double sink with two sinks, and every indegree~2, outdegree~1 vertex with a sink and an indegree 1, outdegree 1 vertex. Note that we replaced only excess vertices and for each replacement we locally added 1 or 2 sinks.
Hence, a solution to this \problem{Sink} instance can be converted back into a solution for the \problem{Excess(2,1)} instance.
\end{proof}

\smallskip

\begin{claim}\label{cla:sinkequalssharpP}
$\#\PPADS\problemp{Sink}-1 = \sharpP$ \ts via relativizing parsimonious reductions.
\end{claim}
\begin{proof}
To see the containment in $\sharpP$, we just count only the nonzero sources.
The second direction is the same construction as in Claim~\ref{cla:sourceorpresinkequalssharpP}.
\end{proof}

\smallskip

\begin{claim}\label{cla:presinkequalssharpP}
$\#\PPADS\problemp{Presink}-1 = \sharpP$ \ts via relativizing parsimonious reductions.
\end{claim}
\begin{proof}
This is the same proof as for Claim~\ref{cla:sinkequalssharpP}.
\end{proof}

\smallskip

\subsubsection{{\normalfont\#\PPA}}

\begin{claim}\label{cla:preleaf}
\problem{Preleaf} \ts is \ts \PPA-complete.
\end{claim}

\begin{proof}
This follows directly from the fact that if we find a leaf, then we can readily find a preleaf, and vice versa.
\end{proof}

\smallskip

\begin{claim}
\label{cla:PPAhalfinPPA}
$\sharpCOUNTALL[-PPA]\problemp{Leaf}/2 \.
\subseteq \. (\#\PPA\problemp{Leaf}-1)/2$ \.
via relativizing parsimonious reductions.
\end{claim}

\begin{proof}
Given a \ts $\problem{AllLeaves}$ \ts instance $G$ we create a \ts $\problem{Leaf}$ \ts
instance~$G'$ as follows.
We add a zero leaf vertex and pair it with another newly created leaf, so \.
$\problem{AllLeaves}(G)+1=\problem{Leaf}(G')$.  Therefore, if \ts $\problem{AllLeaves}(G)/2=k$,
then \. $(\problem{Leaf}(G')-1)/2 = (\problem{AllLeaves}(G)+1-1)/2 = k$, which finishes the proof.
\end{proof}

\smallskip

\subsubsection{{\normalfont\#\PPP}}

We recall the slightly different definition of the classical $\problem{Pigeon}$ problem:
In our case, if $C(x)=C(y)$ and $C(C(x))=0$, then we do not count this as a solution.
However, in this case $z=C(x)$ obviously is a solution, because $C(z)=0$.

\smallskip

\begin{claim}\label{cla:PPPpigeno}
\problem{Pigeon} \ts
is \ts $\PPP$-complete.
\end{claim}

\begin{proof}
For the sake of this proof, let \ts \problem{ClPigeon} \ts denote the classical pigeon search problem
and \ts \problem{Pigeon} \ts our problem with the slight modification.
Given a \ts \problem{Pigeon} \ts instance.
If we find a solution to the \ts \problem{Pigeon} \ts instance,
then this is also a solution for the same instance interpreted as a \ts
\problem{ClPigeon} \ts instance.
If we find a solution to a \problem{ClPigeon} problem, then this is also a solution for
the same instance interpreted as a \ts \problem{Pigeon} \ts instance, with one exception:
If the solution pair $(v,w)$ satisfies \ts $\varphi(v)=\varphi(w)=x$ \ts and \ts $\varphi(x)=0$.
But in this case we immediately find that \ts $x$ \ts is a solution for \ts \problem{Pigeon}.
\end{proof}

\smallskip

\begin{claim}
\label{cla:PPADSinPPP}
$\#\PPADS\problemp{Excess(2,1)}-1 \.
\subseteq \. \#\PPP\problemp{Pigeon}-1
$ \.
via relativizing parsimonious reductions.
\end{claim}

\begin{proof}
Given an \ts \problem{Excess(2,1)} \ts instance with edges \.
$(v,w) \in \{0,1\}^n\times \{0,1\}^n$,
we design a \ts \problem{Pigeon} \ts instance, which is a map \.
$\varphi:\{0,1\}^n\to\{0,1\}^n$.
If the outdegree of $w$ is~1, say $(w,v)$ is an edge, then we set \ts $\varphi(w)=v$.
If the outdegree of $w$ is~0, we set \ts $\varphi(w)=w$, with one exception:
for every double sink $w$ we set \ts $\varphi(w)=0$.
This table explains in which cases the output \problem{Pigeon} instance has witnesses to count:

{\small
\begin{center}
\begin{tabular}{l|l}
(indegree, outdegree) of $w$ & witnesses in $\varphi$
\\
\hline
(0,0) & no witnesses, $\varphi(w)=w$
\\
\hline
(1,0) & yes, $\varphi(v)=\varphi(w)=w$
\\
\hline
(2,0) & \begin{minipage}{10cm}
yes, $\varphi(w)=0$. Note that we do not count the pair $(v,v')$ with $\varphi(v)=\varphi(v')=w$ here, because $\varphi(w)=0$.
\end{minipage}
\\
\hline
(0,1) & no
\\
\hline
(1,1) & no
\\
\hline
(2,1) & \begin{minipage}{10cm}
yes, one witness: Either we have $\varphi(v)=\varphi(v')=w$ and $\varphi(w)\neq 0$, or we have $\varphi(w) = 0$.
\end{minipage}
\end{tabular}
\end{center}
}

\nin
We observe that we have one witness in exactly the excess cases, and no witnesses otherwise.
\end{proof}

\subsubsection{{\normalfont\#\CLS}}

The counting problems \ts
\problem{EitherSolution($*$,\ts{}Iter)} \ts
have the technicality that if both component instances contain the last option, then this witness is only counted once.
In particular, this is used to show
Claim~\ref{cla:CLSinPLSandPPAD} below.
This has no effect on the $\CLS$-completeness of the search problem.
For the oracle separation in Theorem~\ref{thm:decrementationseparation} this also plays no role.

\smallskip

\begin{claim}\label{cla:CLSsourceorsink}
The search problems to the counting problems
$$\aligned
& \text{\problem{EitherSolution(SourceOrSink,\ts{}Iter)},} \\
& \text{\problem{EitherSolution(SourceOrPresink,\ts{}Iter)}, \. and} \\
& \text{\problem{EitherSolution(SourceOrExcess(2,1),\ts{}Iter)}}
\endaligned
$$
are all \ts $\CLS$-complete.
\end{claim}

\begin{proof}
The search problems are in fact the same as the classical $\CLS$-complete problems. Only the counting versions are different, as in some situations 2 solutions are counted together as 1 solution.
\end{proof}

\begin{claim}
\label{cla:CLSinPLSandPPAD}
\begin{eqnarray*}
\#\CLS\problemp{EitherSolution(SourceOrExcess(2,1),\ts{}Iter)}-1 \quad\quad\quad\quad\\
\subseteq \ \ \#\PLS\problemp{Iter}-1 \ \cap \  \#\PPAD\problemp{SourceOrExcess(2,1)}-1
\end{eqnarray*}
via relativizing parsimonious reductions.
\end{claim}
\begin{proof}
We first prove that
$$
\#\CLS\problemp{EitherSolution(SourceOrExcess(2,1),\ts{}Iter)}-1 \
\subseteq \ \#\PLS\problemp{Iter}-1\ts.
$$
Given an \. $\problem{EitherSolution(SourceOrExcess(2,1),\ts{}Iter)}$ \. instance,
we construct an \ts $\problem{Iter}$ \ts instance of the same value.
This new instance consists of two parts.

We copy the input \ts $\problem{Iter}$ \ts instance into the first part of the output \ts $\problem{Iter}$ \ts instance.
Then for every source or excess vertex in the input \ts $\problem{SourceOrExcess}$ \ts instance we construct a (source,\ts{}sink) pair that consists of a single edge in the second part of the output \ts $\problem{Iter}$ \ts instance, with one exception: if both the \ts
$\problem{SourceOrExcess}$ \ts input and the \ts $\problem{Iter}$ \ts input
contain the last option, then we do not create the last (source,\ts{}sink) pair. The value of the resulting instance is the sum of the two values of the input instances in the case where not both use the last option, and is 1 less otherwise. That is exactly the correct amount, by definition of \. \problem{EitherSolution(SourceOrExcess(2,1),\ts{}Iter)}.

\smallskip

The proof for
$$\aligned
& \#\CLS\problemp{EitherSolution(SourceOrExcess(2,1),\ts{}Iter)}-1 \\
& \hskip1.5cm \subseteq \ \, \#\PPAD\problemp{SourceOrExcess(2,1)}-1
\endaligned
$$
is only slightly more tricky.
Given an \. $\problem{EitherSolution(SourceOrExcess(2,1),\ts{}Iter)}$ \. instance,
we construct a \. $\problem{SourceOrExcess(2,1)}$ \. instance of the same value.
This new instance consists of two parts, where we have the zero vertex in the second part.

We copy the input \. $\problem{SourceOrExcess(2,1)}$ \. instance into the first part
of the output \. $\problem{SourceOrExcess(2,1)}$ \. instance.
Then we create the second part as a long line starting from the zero vertex and pointing from \ts $x$ \ts to \ts $x+2$ \.
leaving 1 loop in between,
until the last vertex maps to the source of the original zero source of the instance that is embedded into the first part.
Then for every presink in the input $\problem{Iter}$ instance at position~$\ts x$,
we construct a source in the second part at \ts $x+1$ \ts pointing at~$\ts x$. This makes \ts $x$ \ts an excess vertex.

We have one exception to this rule: If both the \ts $\problem{SourceOrExcess}$ \ts input and the \ts $\problem{Iter}$ \ts input
contain the last option, then then we do not create the last excess vertex. The value of the resulting instance is the sum of the two values of the input instances in the case where not both use the last option, and is 1 less otherwise. That is exactly the correct amount, by definition of \. \problem{EitherSolution(SourceOrExcess(2,1),\ts{}Iter)}.
\end{proof}

\smallskip

\begin{claim}\label{cla:CLSssiequalssharpP}
$\#\CLS\problemp{EitherSolution(SourceOrSink,\ts{}Iter)}-1 \. = \. \sharpP$ \.
via relativizing parsimonious reductions.
\end{claim}
\begin{proof}
We first prove the inclusion \. $\sharpP \subseteq \#\CLS(..)-1$.
Given a \ts $\problem{CircuitSat}$ \ts instance, we create a \.
\problem{EitherSolution(SourceOrSink,\ts{}Iter)} \. instance with one more solution:
the \ts \problem{SourceOrSink} \ts part of our output instance points
from the zero vertex to the last vertex, hence using the last option.
The \ts \problem{Iter} \ts part of our output instance points from the zero vertex
to the vertex before the last vertex and from there to the last vertex,
hence also using the last option.
Now, we add \ts (source,\ts{}sink) \ts pairs to the \ts \problem{Iter} \ts part
at all position where the \ts \problem{CircuitSat} \ts instance yields True.
This gives a \. \problem{EitherSolution(SourceOrSink,\ts{}Iter)} \.
instance of value 1 more than the \ts \problem{CircuitSat} \ts instance.

We now show the opposite inclusion \.
$\#\CLS(..)-1 \subseteq \sharpP$.
The proof is similar as for \ts \problem{Sperner}.
The $\sharpP$ machine first checks if both parts use the last option.
If yes, then not two, but only one of the nondeterministic computation paths accepts.
The rest of the computation counts all \ts \problem{SourceOrSink} \ts
nonzero sources twice and all \ts \problem{Iter} \ts presinks once.
\end{proof}

\smallskip

\begin{claim}
\label{cla:CLSspiequalssharpP}
$\#\CLS\problemp{EitherSolution(SourceOrPresink,\ts{}Iter)}-1 \ts = \ts \sharpP$ \.
via relativizing parsimonious reductions.
\end{claim}
\begin{proof}
The inclusion \. $\sharpP \subseteq \#\CLS(..)-1$ \.
is the same as in Claim~\ref{cla:CLSssiequalssharpP}, but in the
\ts \problem{SourceOrPresink} \ts part instead of directly pointing
to the last vertex we create the edge from the zero vertex to the
second to last vertex and then an edge to the last vertex, with the same effect.

The proof of the opposite inclusion \. $\#\CLS(..)-1 \subseteq \sharpP$ \.
is only slightly more subtle than for Claim~\ref{cla:CLSssiequalssharpP}, because
a \ts \problem{SourceOrPresink} \ts instance can have vertices that are both
source and presink, which we call amalgamations.
The $\sharpP$ machine first checks if both parts use the last option.
If yes, then not two, but only one of the nondeterministic computation paths accepts.
The rest of the computation counts all \ts \problem{SourceOrPresink} \ts
nonzero sources twice, all \ts \problem{SourceOrPresink} \ts amalgamations,
and all \ts \problem{Iter} \ts presinks once.
\end{proof}

\smallskip

\subsubsection{{\normalfont\problem{BipartiteUnbalance}}}
\label{subsubsec:BipartiteUnbalance}
Recall the setting of this problem.
Let \ts $G=(V,E)$ \ts be a bipartite graph with
two parts given by \ts $V=V_-\sqcup V_+$.
We say that $G$ is \ts \emph{unbalanced}
if for every \ts $(uv)\in E$, \ts $u\in V_-$ \ts and \ts $v\in V_+$,
we have \. $\deg(u)\ge \deg(v)$.

\smallskip

\begin{proposition} \label{p:flow-graph}
Let \ts $G=(V,E)$, \ts $V=V_-\sqcup V_+$\., be an unbalanced
bipartite graph as above.  Then \ts $|V_+|\ge |V_-|$.
\end{proposition}

\begin{proof}  The proof is one line:
$$\aligned
|V_+| \ & = \ \sum_{v\in V_+} \. 1  \ = \ \sum_{v\in V_+} \. \sum_{(u,v)\in E} \. \frac{1}{\deg(v)}
 \  = \  \sum_{(u,v)\in E} \. \frac{1}{\deg(v)} \  \ge
\\
&  \ge \ \sum_{(u,v)\in E} \. \frac{1}{\deg(u)} \  = \
\sum_{u\in V_-}  \. \sum_{(u,w)\in E} \. \frac{1}{\deg(u)}  \ = \ \sum_{u\in V_-} \. 1 \ = \
|V_-|\.,
\endaligned
$$
where the only inequality is by the unbalanced condition.
\end{proof}

This first insight shows that \. $\problem{BipartiteUnbalance}\in\GapP_{\geq 0}$.
We study this problem further in \ts $\S$\ref{sec:unbalancedflowseparation}.

\subsection{The decrementation separation}\label{sec:decrementationseparation}
In this section we prove the following technical result:

\smallskip

\begin{theorem}
\label{thm:decrementationseparation}
We have \. $\sharpP \subseteq \#\CLS\problemp{EitherSol(SourceOrSink(2,1),\ts{}Iter)}-1$ \.
via a parsimonious relativizing reduction.
Moreover, there exists a language \ts $A\subseteq\{0,1\}^*$ \ts with respect to which \.
$\#\CLS\problemp{EitherSol(SourceOrSink(2,1),\ts{}Iter)}^A-1 \not\subseteq \sharpP^A$.
\end{theorem}

\begin{proof}
The inclusion \. $\sharpP^A \subseteq \#\CLS(..)^A-1$ \. is of the same level of difficulty
as similar inclusions in~$\S$\ref{sec:syntacticcountingsubclasses}.
Given a \problem{CircuitSat} instance, we construct an instance of \.
$\#\CLS(..)$ \. as follows.
The \ts \problem{SourceOrSink(2,1)} \ts instance jumps directly from the zero vertex to the last vertex.
The \problem{Iter} instance also does that. Now wherever the \ts \problem{CircuitSat} \ts circuit yields true,
we add a \ts (source,\ts{}sink) \ts pair on the \ts \problem{Iter} \ts side.
We obtain an instance of value \ts $f+1$, where the \ts $+1$ \ts comes
from the two last options that are counted together as~1.
We are done, because \ts $(f+1)-1 = f$.

For the other direction we use the Diagonalization Theorem~\ref{thm:diagonalization}
and the set-instantiator for \.
\problem{Multi(NZSource,\ts{}Excess,\ts{}Iter)}, see Theorem~\ref{thm:CLSsetinstantiator},
with
\[
k=3\., \quad S \. = \. \big\{\.\vv f \in \IN^k \, : \, 2f_2=f_1+1, \. f_2\geq 1, \. f_3\geq 1\.\big\}.
\]
We switch the positions~1 and~3 to get the set-instantiator for
\problem{Multi(Iter,\ts{}Excess,\ts{}NZSource)}
with
\[
k=3\., \quad S \.= \. \big\{\.\vv f \in \IN^k \, : \, 2f_2-1=f_3, \. f_2\geq 1, \. f_1\geq 1\.\big\}.
\]

We set \ts $\ka=2$ \ts and \ts $\zeta_3(v_1,v_2) = 2v_2-1$, so that \. $I = \langle 2v_2-1-f_3\rangle$ \. and \.
$$
Z \, = \, \big\{\.\vv f \in \IQ^3 \.: \. 2f_2-1-f_3=0\.\big\}.
$$
We choose any \ts $\vv t \in S$ \ts
such that \ts $C_\bot$ \ts exists with \.
$\problem{Multi(Iter,\ts{}Excess,\ts{}NZSource)}(C_\bot)=\vv t$, and \ts $C_\bot$ \ts
does not use both last options.
Now we verify the preconditions of the Diagonalization Theorem~\ref{thm:diagonalization}:
\begin{enumerate}
\item $Z$ \ts contains the integer point \ts $(0,0,-1)$.
\item We prove that \ts $C'_S$ \ts lies Zariski-dense in \ts $\IQ^2$ \ts
by giving infinitely many distinct rays from the origin that lie in~$\ts C'_S$.
This proves the claim, because we are in a 2-dimensional situation.
In fact, for every \ts $(1,f_2)$, \ts $f_2 \in \IN_{\geq 1}$, the corresponding ray \.
$\{(\alpha,\alpha f_2) \.: \. \alpha \in \IQ_{\geq 0}\}$ \.
is in~$C'_S$, because \. $(\alpha,\alpha f_2, 2 \alpha f_2 -1) \in S$ \.
for all \ts $\alpha \in \IN_{\geq 1}$.
\item $\vv v = (0,1)$ \ts satisfies \ts $2v_2>0$.
\end{enumerate}
We fix \. $\problem{Multiplicities} =
\problem{Multi(Iter,\ts{}Excess,\ts{}NZSource)}$ \.
for which we know for every $M$ the existence of a set-instantiator against $M$.
We set \. $\varphi(f_1,f_2,f_3) = f_1+f_2+f_3-1$.
We verify that \. $\varphi+I = f_1+f_2+f_3-1+\langle2f_2-1-f_3\rangle$ \. is binomial-bad.
This is proved using the Polyhedron Theorem~\ref{thm:polyhedron}.

In the notation of
the Polyhedron Theorem~\ref{thm:polyhedron}, we have \ts $\varphi' = f_1+3 f_2 -2$ \ts
and \ts $\delta=1$, so \ts $\IO_\delta$ \ts has only 3 elements.  Hence the polyhedron \.
$\mathcal P(\varphi,\zeta)$ \. will be defined using 3 equations.
Let \ts $(1,v_1,v_2)$ \ts be the order of basis vectors of basis vectors of \ts
$\IQ[\vv v]_{\leq 1}$. We express the 4 polynomials \ts $y_{\vv e}$ \ts over this basis,
and obtain coefficient vectors \ts $(1,0,0)$, \ts $(0,1,0)$, \ts $(0,0,1)$ \ts and \ts $(-1,0,2)$.

Now, the polyhedron given by
\[
\begin{array}{rrrrrr}
 x_0 &     &       & - x_3   &=& -2 \\
     & x_1 &       &         &=& 1 \\
     &     & x_2   & + 2 x_3 &=&  3 \\
 \multicolumn{4}{r}{x_0, x_1, x_2, x_3}  &\geq& 0 \\
\end{array}
\]
has no integer points.  Hence \ts $\varphi+I$ \ts is binomial-bad by the Polyhedron Theorem~\ref{thm:polyhedron}.\footnote{In fact, even the LP relaxation is empty in this case, so every positive integer multiple of the problem \problem{EitherSolution(SourceOrExcess(2,1),\ts{}Iter)} can also be separated from $\sharpP$ via an oracle. This is not always the case, and we use this property to separate two classes that are outside of $\sharpP$ (w.r.t.\ an oracle), see Proposition~\ref{pro:unbalancedflowseparation}.}

From the Diagonalization Theorem~\ref{thm:diagonalization} it follows that for every Turing machine~$M$,
we have \. $p_A(0^j)\neq \#\acc_{M^A}(0^j)$, and hence \. $p_A\notin\sharpP^A$.
It remains to show that
$$
p_A \, \in \, \#\CLS\problemp{EitherSolution(SourceOrExcess(2,1),\ts{}Iter)}^A-1\ts.
$$
Let $C_\bot$ be an input to \problem{Multi(Iter,\ts{}Excess,\ts{}NZSource)}
(i.e., a circuit, see \cref{footnote:singlecircuit}) that yields output~$\ts\vv t$.
Let \ts $\ell$ \ts be the number of inputs of the circuit~$\ts C_\bot$.
Let \ts $\nu_j(w) = \textup{True}$ \ts if and only if \ts the string $w$ has only zeros at positions $j+1,j+2,\ldots$
Finally, let \ts $C_j$ \ts be an input to \. $\problem{Multi(Iter,\ts{}Excess,\ts{}NZSource)}^A$ \.
that describes the $(j-1)$-input circuit that consists of a single arity $j$ oracle gate that takes a constant 1 into its first input, takes all inputs into its remaining $j-1$ inputs, and forwards its output directly to the circuit output.

We construct a circuit $\alpha(x) := D_{|x|}$
with $\max\{\ell,|x|\}$ many inputs
and
with oracle gates as follows $(w\in\{0,1\}^{\max\{\ell,|x|\}})$:
\begin{eqnarray*}
D_j(w) & := & \Big((A(0^{|w|+1})=0) \textsf{ and }
\nu_\ell(w)
\textsf{ and }
C_\bot(w_1,\ldots,w_\ell)\Big)
\\
&\textsf{ and }&
\Big( (A(0^{|w|+1})=1)
\textsf{ and }
\nu_{|x|}(w)
\textsf{ and }
C_{|x|+1}(w_1,\ldots,w_{|x|})
\Big)
\end{eqnarray*}
We define the polynomially balanced relation \ts $R$ \ts via \ts $R(w,y)$ \. if and only if \.
$$
\problem{rEitherSolution(SourceOrExcess(2,1),\ts{}Iter)}(\alpha(w),y).
$$
We set \ts $\beta$ \ts to be the identity function.
By definition, \ts $R \in \ComCla{rCLS}$.
The following function is thus in $\#\CLS$:
\[
\sum_{y \in\{0,1\}^*} \. R(x,y) \ = \
\sum_{y \in\{0,1\}^*} \. \problem{rEitherSolution(SourceOrExcess(2,1),\ts{}Iter)}(\alpha(x),y)
\]
which equals (because the set-instantiator does not create instances that use the last option)
\[
\begin{cases}
\problem{Iter}(D_{|x|}) +
\problem{Excess}(D_{|x|}) +
\problem{NZSource}(D_{|x|})
& \textup{ if } A(0^{|x|+1})=1
\\
t_1 + t_2 + t_3 & \textup{ otherwise}
\end{cases}
\]
which equals $p_A(w)+1$. This finishes the proof.
\end{proof}

\smallskip

\subsection{A set-instantiator for {\normalfont\#\PLS\problemp{Iter}}}
\label{sec:IterSetInstantiator}

The set-instantiator in this section is the most complicated we create.
The problem is that instances cannot be permuted arbitrarily,
so one has to view \ts $[2^j]$ \ts as a  \ts $[2^{j/2}]\times[2^{j/2}]$  \ts square,
and permute within the rows.  But that does not immediately give a way to mask
from the  \ts $\sharpP$ \ts machine which presink belongs to the zero source,
a property which crucially must be achieved.

\smallskip

\begin{theorem}
\label{thm:PLSsetinstantiator}
For every \ts $b \in \IN$,
there exists a threshold \ts $j_0\in \IN$, such that for every \ts $j \geq j_0$,
every \ts $A_{<j}\in\{0,1\}^{<j}$ \ts and every polynomial time Turing machine~$M$
that answers consistently for \ts $(j,A_{<j},\problem{Iter})$, there exists a
set-instantiator \ts  $\SI$  \ts against \ts $(M,j,A_{<j},S=\IN_{\geq 1},b)$,
such that \. $\problem{Iter}(\inst_\SI(s)) = |s|$ \.  for all \. $s$ \ts with \ts $|s|\in S$.
\end{theorem}

\smallskip

Intuitively, this says that for all our purposes \ts $\#\PLS$ \ts behaves exactly as \ts $\sharpP_{\geq 1}$.
The rest of this section is devoted to proving this theorem.
We will first define a \emph{creator} whose \emph{creations} will be the instantiations in the end, but the creator is not limited to a single creation for each set. We will then define the set-instantiator from the creator by picking for every subset $\vv s$ just any one of its creations for~$\vv s$.

\smallskip

\subsubsection{Creations}
Having \ts $2^{j-1}$ \ts many bits available, we can encode in a standard way
a function \ts $[2^{j'}]\to[2^{j'}]$, where $j'$ and $j$ are polynomially related.
We ignore any extra bits that are not needed for this encoding.
We think of \ts $[2^{j'}]$ \ts arranged as a square, so set \.
$n := \lfloor\log_2(\sqrt{2^{j'}})\rfloor$ \.
and interpret \ts $2^n$ \ts as the number of rows, i.e., assume we have a polynomial
time computable injective map \. $[2^n]\times[2^n]\to[2^{j'}]$.
We ignore numbers in \ts $[2^{j'}]$ \ts that are outside this square (i.e., outside the image of this injective map).
We sometimes call elements of \ts $[2^n]\times[2^n]$ \ts nodes, as is usual for \ts \problem{Iter} \ts instances.

The function \ts $\psi$ \ts is given by the $j$-th layer $A_j$ of the oracle $j$,
so can say that a computation ``queries $\psi$'' when we mean it queries~$A_j$.
We will use these interchangeably.
For every \. $x=(x_1,x_2) \in [2^n]\times[2^n]$, let \ts $\textup{\textsu{row}}(x)=x_1$.
Let \ts $\textsu{toprows} := [1,2^{n-1}]$ (the upper half of the rows), \ts
$\textsu{presinkrow} := 2^n-1$ (one from the very bottom), and let
\ts $\textsu{sinkrow} := 2^n$ (at the very bottom).
All our instances will have all sinks in row \textsu{sinkrow} and all presinks in row \textsu{presinkrow}.

A map \. $\psi:[2^n]\times[2^n]\to[2^n]\times[2^n]$ \. is called
\defn{row-layered} if \ts $\psi(x)>x$ \ts implies that \.
$\textup{\textsu{row}}(\psi(x))=\textup{\textsu{row}}(x)+1$.
A point \. $x \in [2^n]\times[2^n]$ \. with \ts $\psi(x)=x$ \ts is called a \defn{fixed point}.
A point \. $x \in [2^n]\times[2^n]$ \. for which \. $\exists y\in[2^n]\sm\{x\} \. : \. \psi(y)= x$ \. is called a \defn{sink}.
A fixed point that is not a sink is called a \defn{singleton}.
A point \. $y \in [2^n]\times[2^n]$ \. for which \ts $\psi(y)$ \ts is a sink is called a \defn{presink}.
Note that when given oracle access to $\psi$ it is easy to check if a specific $y$ is a presink,
but it is not easy to check if a specific $x$ is a sink (because it is hard to distinguish it from a singleton).

A \defn{path part} \ts in \ts $\psi$ \ts is a set of vertices \. $v_1,\ldots,v_q$ \.
such that \. $\psi(v_i)=v_{i+1}$, $\psi(v_q)\neq v_q$ \ts and \ts $\psi(\psi(v_q))=\psi(v_q)$.
The node $v_1$ is called the \defn{source}.
By definition, the node $v_q$ is a presink and each path part has exactly one presink.
$\psi(v_q)$ is called the \defn{sink}.
We consider the presink to be part of the path part, while we consider the sink to \emph{not} be part of the path part.
It is clear that a single vertex with a loop is not a path part.
A \defn{path} \ts is a path part that is maximal with respect to inclusion.

A map \. $\psi:[2^n]\times[2^n]\to[2^n]\times[2^n]$ \. is called
a \defn{set-of-paths} \ts
if every node $x$ has no predecessor or exactly one predecessor in $\psi$, i.e.,
we partition \ts $[2^n]\times[2^n]$ \ts into a disjoint union of paths and sinks and singletons.
A set-of-paths is called \defn{rooted} if one of its sources is at position~$\ts (1,1)$.
This path is called the \defn{main path}.
All other paths are called non-main paths.

For \ts $a\geq 1$,
a rooted set-of-paths is called an \defn{$a$-creation} \ts if it consists of exactly $a$ many paths,
and the rows of the sources of all non-main paths are in row $\textsu{presinkrow}$,
and all paths have their sink in row $\textsu{sinkrow}$ in the same column as their presink
(in particular, those paths consist only of a single edge).
Let \ts $\textsu{$a$-creations}$ \ts denote the set of all $a$-creations.
Finally, let \, $\textsu{$a^-$-creations} \. := \. \bigcup_{1 \leq c \leq a} \. \textsu{$c$\,-creations}$.

\subsubsection{The unaccessed row $\eta$}
Fix $b \in \mathbb N$.  From this point on, $\eta$ will depend on $b$, but $b$ will remain constant.
Let \ts $t_i(j)$ \ts be a polynomial in~$j$, defined as the upper bound on the number of computation
steps that $M$ makes on inputs of length~$j$ (this is independent of to which oracle $M$ has access).
Since $M$ answers consistently,
for every \ts $\psi\in \textsu{$b^-$-creations}$ \ts with \ts $\problem{Iter}(\psi)=a$,
there exist exactly \ts $\varphi(a)$ \ts many accepting computation paths for the
computation \ts $h^\psi_i(0^j)$. Each of these paths queries the oracle $\psi$
at most \ts $\varphi(a)t_i(j)$ \ts many positions.
Thus, together they access at most an \. $\frac{\varphi(a)t_i(j)}{2^{j/2}}$ \. fraction of $\textsu{toprows}$.

From above, for a uniformly random \ts $\eta\in\textsu{toprows}$ \ts and any \ts $\psi\in\textsu{$b^-$-creations}$,
we have w.h.p.\ that no accepting computation path accesses row~$\eta$.
This holds for every \ts $\psi\in\textsu{$b^-$-creations}$,
so it also holds when sampling $\psi$ from any distribution $E$ on \ts \textsu{$b^-$-creations}:
When sampling \ts $(\eta,\psi)$ \ts from \ts $U_{\textsu{toprows}}\times E$, then w.h.p.\
none of the accepting paths of \ts $h^\psi_i(0^j)$ \ts accesses row~$\eta$.
This can be seen, for example, by first sampling \ts $\psi\in E$,
and then sampling \. $\eta\in U_{\textsu{toprows}}$ \. independently.
Therefore, using the union bound, for every finite set of distributions \. $E_1,\ldots,E_b$ \.
on \ts $\textsu{$b^-$-creations}$, there exists \ts $\eta_E\in\textsu{toprows}$, such that
\. for all \. $1\leq a \leq b$ \. we have:
\begin{equation}\label{eq:noaccesstoeta}
\aligned &
\text{when sampling \ts $\psi$ \ts from \ts $E_a$, we have w.h.p.\ that} \\
& \text{no accepting path of \ts $h_i^\psi(0^j)$ \ts accesses row \ts $\eta_E$\ts.}
\endaligned
\end{equation}
For the rest of this proof, we fix
$$\eta \. = \. \eta_{(U_{1\textsu{-creations}}\ts,\ts\ldots\ts,\ts U_{b\textsu{-creations}})}\..
$$
We remark that it is an important technical difficulty when constructing set-instantiators that the sampling process of the $b^-$-creations must be independent of~$\eta$.

\subsubsection{Lucky creators}
In this section we introduce the concept of a \defn{creator}. 
A creator is a slightly less restrictive version of a set-instantiator.
For each blueprint a creator outputs a creation, where a blueprint is slightly more general than the set $\vv s$ for a set-instantiator, but serves the same purpose.

A layered path part with source $(1,1)$ to a presink in row $\eta-1$ is called a \defn{head}.
A layered path part with source row $=$ presink row $= \eta$ is called a \defn{fork} (just one edge).
A layered path part with source row $\eta+1$ to a presink in row $\textsu{presinkrow}-1$ is called a \defn{tail}.
A layered path part $x$ with source row $=$ presink row \ts $=\textsu{presinkrow}$, where
\ts $\text{column}(x)=\text{column}(\psi(x))$ \ts is called a \defn{sink part} (just one edge).
Recall that by definition an $a$-creation consists of 1~head, 1~fork, 1~tail, $a$~many sink parts,
and is uniquely defined by these.
If the sink of a path part equals the source of another path part, then their union is another path part.
For every tail there is a unique way to take a union with a sink part, which is by adding an edge that stays in the same column.

For \ts $b\geq 1$ a \defn{$b$-creator} \ts $\xi$ is a 3-tuple \. $(\textup{\textsu{head}}(\xi),\ts{}\textup{\textsu{tails}}(\xi),\ts{}\textup{\textsu{presinks}}(\xi))$,
where \ts $\textup{\textsu{head}}(\xi)$ \ts is a head, \ts
$\textup{\textsu{tails}}(\xi)$ \ts is a cardinality~$b$ ordered set of tails whose nodes are pairwise disjoint,
and \ts
$\textup{\textsu{presinks}}(\xi)$ \ts is a cardinality~$b$ ordered set of sink parts,
so that the tails line up properly with the corresponding sink parts: adding an edge
that stays in the same column to a tail is the same as taking the union with the sink part corresponding to the tail.
Let $\textsu{$b$-creators}$ denote the set of all $b$-creators.
Instead of writing ``presink of $\xi$'' we sometimes write $\xi$-presink. Analogously for tails.

Given a $b$-creator \ts $\xi$ \ts and \ts $o' \in s'\subseteq[b]$,
we define the $|s'|$-creation \ts $\xi_{s',o'}$ \ts via removing all tails
but keeping the tail that connects to the presink \ts $\textsu{presinks}(\xi)_{o'}$,
and removing all presinks outside of \. $\{\textsu{presinks}(\xi)_{a}\mid a \in s'\}$,
defining the function $\phi$ that has exactly those path parts, and then setting \ts $\phi(x)=y$
for the fork $x$ and the source of the tail~$y$.
Clearly $\xi_{s',o'}$ is an $|s'|$-creation.

Given a $b$-creator \ts $\xi$, define \ts $\xi_{o'}:= \xi_{[b],o'}$.
Clearly $\xi_{o'}$ is a $b$-creation.

\smallskip

\begin{definition}\label{def:PLSsuccess}
We call \ts $\xi \in \textsu{$b$-creators}$ \ts \defn{lucky} \ts  if
for all subsets \ts $L\subseteq [b]$, \ts $|L|\geq 1$, and all \ts $o \in L$,
all accepting paths of the computation \ts $h_i^{\xi_{L,o}}(0^j)$
\begin{enumerate}[label=(\roman*)]
\item do not access the oracle in row $\eta$,
\item do not access the oracle in any $\xi$-presink besides the presinks that correspond to indices in $L$ (they do not have to access \emph{all} $L$-presinks), and
\item do not access the oracle in any tails besides the tail \. $\textsu{tails}_o$.
\end{enumerate}
\end{definition}

\smallskip

Informally, this says that the accepting paths do not access the oracle
at positions where presinks and tails are not, but could potentially be.
\smallskip

\begin{claim}\label{cla:PLSwhpsuccess}
Let \ts $\xi$ \ts be sampled from \. $U_\textup{\textsu{$b$-creators}}$.  Then \ts $\xi$ \ts is lucky w.h.p.
\end{claim}
\begin{proof}
We describe a way of sampling from $U_\textup{\textsu{$b$-creators}}$.
Let \ts $\textsu{heads}$ \ts be the set of all heads,
let \ts $\textsu{tails}$ \ts be the set of all tails, and let
$$\textsu{tail-$b$-tuples}\, =\, \big\{\vv q \in \textsu{tails}^b \, : \, \text{ all \ts $q_1,\ldots,q_b$ \ts have pairwise disjoint nodes and sinks} \big\}.
$$
We first sample \ts $(h,\vec q)$ \ts from \. $U_{\textsu{heads}}
\times U_{\textsu{tail-$b$-tuples}}$.
Choosing the matching presinks, this process samples a random variable \.
$\textsu{crea}(h,\vv q) \in \textsu{$b$-creators}$ \. according to \. $U_\textup{\textsu{$b$-creators}}$.
We observe that for fixed \ts
$L\subseteq[b]$ \ts and \ts $o \in L$, and $\xi:=\textsu{crea}(h,\vv q)$,
the \ts $|L|$-creation \ts $\xi_{L,o}$ \ts is uniformly distributed from \ts $|L|$\textsu{-creations}.

We show that for a fixed $L\subseteq[b]$, $|L|\geq 1$, and $o \in L$ we have that $(h,\vec q)_{L,o}$ satisfies all three
properties of Definition~\ref{def:PLSsuccess} with high probability.
Since there are only \ts $2^b-1$ \ts many possible sets $L$, since there are only \ts $b$ \ts
many possible values for~$o$, and since $b$ is fixed, the claim then immediately follows
from the union bound. 
Hence, for the rest of this proof we fix \ts $L\subseteq[b]$ \ts and \ts $o \in L$.
In this fixed case we will also use the union bound.

\smallskip

\nin
\underline{Property~$(i)$}.
Since \ts $(h,\vec q)_{L,o}$ \ts is distributed as \. $U_{\textsu{$|L|$-creations}}$, \eqref{eq:noaccesstoeta} proves $(i)$ in Definition~\ref{def:PLSsuccess} with high probability.

\smallskip

\nin
\underline{Properties $(ii)$ and $(iii)$}.
For a fixed \. $\psi \in a\textsu{-creations}$, if we draw \ts $(h,\vec q)$ \ts from
\. $U_{\textsu{heads}} \times U_{\textsu{tail-$b$-tuples}}$ \.
under the assumption that \. $\textsu{crea}(h,\vv q)_{L,o} = \psi$, then we see that
the tails~$q_l$, \ts $l \notin L$, are \emph{uniformly} distributed
(and hence this also holds for the corresponding presinks).
Moreover, by definition, each \ts (tail,\ts{}presink) \ts pair only contains at most~1 node in each row.
Hence each oracle query accesses such a tail or presink with probability \.
$<\frac{b}{2^n-t_i(j)} \leq \frac{b}{2^{n-1}}$\.. Here \ts $t_i(j)$ \ts
is subtracted in the denominator, because that many positions could have
been queried already and there is no reasone to query the oracle twice at the same position).
By Bernoulli's inequality we have
$$\left(1\ts - \ts\frac{b}{2^{n-1}}\right)^{t_i(j)} \ \geq \ 1 \. - \. \frac{b \ts \cdot \ts  t_i(j)}{2^{n-1}}\..
$$
This proves the property $(ii)$ and $(iii)$.

\smallskip

Since each property $(i)$--$(iii)$ is satisfied with high probability, the union bound gives that they are simultaneously satisfied with high probability.
This finishes the argument for fixed $L$ and $o$. According to the previous discussion
about the union bound
the overall claim is proved.
\end{proof}

\smallskip

The following lemma handles the technical difficulty of the tails.

\smallskip

\begin{lemma}
\label{lem:notailquery}
Let \ts $\xi$ \ts be lucky, let \ts $b \geq 1$, and suppose \ts $\tau$ \ts is an accepting path of \.
$h_i^{\xi_{s,o}}(0^j)$ \. for every \ts $o \in s \subseteq [b]$.  Then \ts $\tau$ \ts
does not query any of the $b$ many tails of~$\ts\xi$.
\end{lemma}

\begin{proof}
Let there be \ts $c$ \ts many accepting paths of the computation \.
$h_i^{\xi_{s,o}}(0^j)$, let \ts $\tau$ \ts be one of those accepting paths,
and w.l.o.g.\ assume that \ts $\tau$ \ts queries the tail~$q_1$.
Let \ts $\psi$ \ts be obtained from the instance \ts
$\xi_{s,o}$ \ts
by removing tail~$\ts q_1$ \ts and adding tail~$\ts q_2$.
Clearly \ts $\psi$ \ts is also a $c$-creation.
Hence, since \ts $\xi$ \ts is lucky, and
since $\tau$ queries $q_1$, we conclude that \ts $\tau$ \ts does not query~$\ts q_2$ (Definition~\ref{def:PLSsuccess}).

Since \ts $M$ \ts answers consistently, there are $c$ many accepting paths for the computation \. $h_i^{\psi}(0^j)$.
Since \ts $\psi$ \ts is a $c$-creation and \ts $\xi$ \ts is lucky, these accepting paths do not query~$\ts q_1$.
Therefore, if we add \ts $q_1$ \ts into the instance (so that now both $q_1$ and $q_2$ are present),
all these \ts $a$ \ts many paths are still accepting paths, and \ts $\tau$ \ts is also accepting (and is obviously different from all these $a$ many paths). Thus there are now at least \ts $a+1$ \ts many accepting paths for this instance,
which must have (because $M$ answers consistently) exactly \ts $a$ \ts many accepting paths.
This is a contradiction.
\end{proof}

\subsubsection{Defining the set-instantiator}
For an $a$-creation \ts $\psi$ \ts and a computation path \ts $\tau \in \{0,1\}^*$ \ts of a computation \ts $h_i^\psi(0^j)$, let \. $\textsu{perception}(\tau)$ \. $\subseteq \ts \textsu{presinks}(\psi)$ \. denote the set of oracle positions with
presinks that get accessed by~$\tau$.

Let \ts $\xi$ \ts be a lucky $b$-creator.
We interpret the list \ts $\textup{\textsu{presinks}}(\xi)$ \ts
as a bijection \. $\bij : [b]\to\textup{\textsu{presinks}}(\xi)$.
Let \ts $S = \IN_{\geq 1}$.
For every \ts $\vv s \in \mP(\vv b)_S$, we choose a \.
$o_{\vv s} \in \textup{\textsu{presinks}}(\xi_{\bij(\vv s)})$.
Set
\[
\inst_\SI(\vv s) \ := \ \xi_{\bij(\vv s),\ts{}o_{\vv s}}
\]
and for \ts $\tau\in\{0,1\}^*$ \ts set
\[
\textsu{perc}_\SI(\tau) \ := \ \begin{cases}
\textsu{bij}^{-1}(\textsu{perception}(\tau)) & \text{ if $\tau$ is an accepting path of the computation $h_i^{\inst_\SI(\vv b)}(0^j)$}
\\
\top & \text{ otherwise. }
\end{cases}
\]

The rest of this section is devoted to proving that \ts $\SI$ \ts satisfies the requirements of
Definition~\ref{def:setinstantiator}, which then proves Theorem~\ref{thm:PLSsetinstantiator},
because clearly \. $\problem{Iter}(\inst_\SI(\vv s))=|\vv s|$.
Formally, we prove the following result.

\smallskip

\begin{proposition}
\label{pro:PLSsetinstantiatorcorrect}
For all \. $\vv s \in \mP(\vv b)_S$ \. we have: \. $\tau\in\{0,1\}^*$ \ts is an accepting path
for the computation \. $h^{\inst_\SI(\vv s)}(0^j)$ \. if and only if \. $\textup{\textsu{perc}}_\SI(\tau) \subseteq \vv s$.
\end{proposition}

\begin{proof}
Since \ts $\tau$ \ts is an accepting path of \. $h^{\inst_\SI(\vv s)}(0^j)$,
 since \ts $\xi$ \ts is lucky, and since $\tau$ does not query any tail (Lemma~\ref{lem:notailquery}),
 we conclude that \ts $\tau$ \ts is also an accepting path of the computation \. $h^{\inst_\SI(\vv b)}(0^j)$. This implies
that \. $\textsu{perc}_\SI(\tau) = \textsu{bij}^{-1}(\textsu{perception}(\tau))$.
Clearly \. $\textsu{perception}(\tau) \subseteq \textsu{bij}(\vv s)$,
because otherwise \ts $\tau$ \ts would not even be a computational path of \ts $h^{\inst_\SI(\vv s)}(0^j)$ \ts
because of its oracle answers when querying presinks in \.
$\textsu{perception}(\tau) \sm \textsu{bij}(\vv s)$. We conclude:  $\textup{\textsu{perc}}_\SI(\tau) \subseteq \vv s$.

The argument is reversible: let \. $\textup{\textsu{perc}}_\SI(\tau) \subseteq \vv s$, so in particular \.
$\textup{\textsu{perc}}_\SI(\tau) \neq \top$.
Then \ts $\tau$ \ts is an accepting path of the computation \ts $h^{\inst_\SI(\vv b)}(0^j)$.
Since \ts $\tau$ \ts is an accepting path of \ts $h^{\inst_\SI(\vv b)}(0^j)$,
since \ts $\xi$ \ts is lucky, and since \ts $\tau$ \ts does not query any tail
(Lemma~\ref{lem:notailquery}), we conclude that \ts $\tau$ \ts
is also an accepting path of the computation \ts $h^{\inst_\SI(\vv s)}(0^j)$.
\end{proof}

\subsection{A set-instantiator for {\normalfont\#\PPAD\problemp{SourceOrExcess(2,1)}}}
\label{sec:SourceOrExcessSetInstantiator}

Let \. $S_{\textup{line}} := \{(f_1,f_2) \. : \. 2f_2-f_1-1=0\}$.  Think of $f_1$ as the number of sources,
and $f_2$ as the number of double sinks.
Let \. \problem{Multi(NZSource,\ts{}Excess)} \. be
the bivariate counting problem of counting sources and counting excess nodes in a \ts
\problem{SourceOrExcess(2,1)} \ts instance.
This set $S = S_{\textup{line}}$ will be sufficient for our set-instantiator.\footnote{Larger $S$ are possible, but not necessary for our purposes,
because already we have that \ts $f_1 + f_2 -1 + \langle2f_2-f_1-1\rangle$ \ts is binomial-bad, as can be seen via the Polyhedron Theorem~\ref{thm:polyhedron}.}

\smallskip

\begin{theorem}
\label{thm:PPADsetinstantiator}
Fix a polynomial time nondeterministic Turing machine~$M$.
Then, for every \ts $\vv b \in \IN^2$,
there exists a threshold \ts $j_0\in \IN$,
such that for every \ts $j \geq j_0$ \ts and every \ts
$A_{<j}\in\{0,1\}^{<j}$, there exists a set-instantiator \ts $\SI$ \ts
against \.
$(M,j,A_{<j},S=S_{\textup{line}},\vv b)$, which satisfies:
$$\problem{Multi(NZSource,\ts{}Excess)}(\inst_\SI(\vv s)) \, = \, |\vv s| \ \quad
\text{for all} \quad \ \vv s \ \text{ with } \ |\vv s|\in S.
$$
\end{theorem}

\smallskip

The rest of this section is devoted to proving this theorem.
We will first define a \emph{creator} whose \emph{creations}
will be the instantiations in the end,
but the creator is not limited to a single creation for each set.
We will then define the set-instantiator from the creator, by
picking for every subset~$\ts \vv s$ \ts with \ts $|\vv s|\in S$,
just any one of its creations for~$\vv s$.

\subsubsection{Creations}
Having \ts $2^{j-1}$ \ts many bits available, we can encode in a standard way a tuple of two predecessor functions \.
$\psi_{\textup{pred}_1}:[2^n]\to[2^n]$, \.
$\psi_{\textup{pred}_2}:[2^n]\to[2^n]$,
and a successor function \. $\psi_{\textup{succ}}:[2^n]\to[2^n]$,
such that $n$ and $j$ are polynomially related.
We ignore any extra bits that are not needed for this encoding.

We interpret \. $\psi=\bigl(\psi_{\textup{pred}_1},\ts\psi_{\textup{pred}_2},\ts\psi_{\textup{succ}}\bigr)$ \.
as a directed graph (with some additional information,
because the map from the set of function triples to the set of digraphs is not injective).
Here an edge from \ts $x \in [2^n]$ \ts to \ts $y\in[2^n]$ \ts is present if and only if
\ts $\psi_{\textup{succ}}(x)=y$, and either \ts
$\psi_{\textup{pred}_1}(y)=x$ \ts or \ts
$\psi_{\textup{pred}_2}(y)=x$.
Each node in the resulting graph has indegree \ts $\leq 2$ \ts and outdegree \ts $\leq 1$.

A node whose indegree exceeds its outdegree is called an \defn{excessive node}.
A node with indegree~0 and outdegree~1 is called a \defn{source}.
A node with indegree~2 and outdegree~0 is called a \defn{double sink}.
We assume that it the encoding is made in a way that the~0 node is always at source.
A source that is not the 0 node is called a \defn{nonzero source}.

Let \ts $\textsu{nzsource}(\psi)$ \ts denote the set of nonzero sources.
Let \ts $\textsu{dsink}(\psi)$ \ts denote the set of double sinks.
For \ts $\vv a \in S$, $\psi$ \ts is called an \defn{$\vv a$-creation} \ts
if \ts $|\textsu{nzsource}(\psi)| = a_1$ \ts and \ts
$|\textsu{dsink}(\psi)| = a_1$.
Let \ts $\textsu{$a$-creations}$ \ts denote the set of all $a$-creations.
Let \. $\textsu{$a^-$-creations} := \bigcup_{1 \leq c \leq a}\textsu{$c$\,-creations}$.

\subsubsection{Lucky creators}
In this section we introduce the concept of a \defn{creator}.
A creator is a slightly less restrictive version of a set-instantiator.
For each blueprint a creator outputs a creation, where a blueprint is slightly more general than the set $\vv s$ for a set-instantiator, but serves the same purpose.

For two disjoint sets $X$ and~$Y$,
let \. $\binom{X}{2}_\textup{ordered}\times Y$ \.
denote the set of triples, where the first and second element
are distinct elements from~$X$, and the third element is from~$Y$.
We say that two triples are \defn{disjoint} \ts if they share no common element.
Let \ts $T(X,Y,c)$ \ts denote the set $c$-tuples of pairwise disjoint triples.

For \ts $\vv b\in S$, a \emph{$\vv b$-creator} \ts $\xi$ \ts is a
triple consisting of:
\begin{enumerate}
 \item a length \ts $b_1$ \ts ordered list \ts $\textsu{nzsources}(\xi)\subseteq[2^n]$,
 \item a length \ts $b_2$ \ts ordered list \ts $\textsu{dsinks}(\xi)\subseteq[2^{n}]$ \ts of even numbers, and
 \item a map \ts $\textsu{straight}(\xi):[2^n]\to[2^n]$ \ts such that all nodes have indegree at most~$2$,
 the 0~node is a source, the set of all other sources is exactly \ts $\textsu{nzsources}(\xi)$, the set of all double sinks is exactly \ts $2\textsu{dsinks}(\xi)$,  the set of all loops is exactly \ts $2\textsu{dsinks}(\xi)+1$, all other nodes hat indegree $=$ outdegree $=1$, \ts $\textsu{straight}(\xi)(x)\neq x+1$ for $x$ even,
 \ts  $\textsu{straight}(\xi)(x+1)\neq x$ \ts for $x$ even,
 and \ts $\textup{pred}_2(x)=x$ \ts for all nodes that are not double sinks.
\end{enumerate}
This means that each double sink comes with a loop at the next position, and no node maps to its paired neighbor.

Given a $\vv b$-creator \ts $\xi$ \ts and an $\vv L \in T\big([b_1]\times\{1\},[b_2]\times\{2\},c\big)$,
we obtain a $(b_2-c)$-creation \ts $\psi:[2^n]\to[2^n]$ \ts
by performing the following step for each triple \ts $(x_1,x_2,y) \in \vv L$~: \ts
set \.
$$\aligned
& \psi_\textup{succ}\big(\textsu{dsinks}(\xi)_y\big)\, := \, \textsu{nzsources}(\xi)_{x_1}\,, \quad
\psi_\textup{succ}\big(\textup{pred}_1(\textsu{dsinks}(\xi)_y)\big)\, :=\, \textsu{dsinks}(\xi)_y\ts +\ts 1\ts, \\
& \hskip2.75cm \text{and} \quad \psi_\textup{succ}\big(\textsu{dsinks}(\xi)_y+1\big)\, := \, \textsu{nzsources}(\xi)_{x_2}\..
\endaligned
$$
In other words, remove the double sink and the loop and redirect the two paths to the two sources.

\smallskip

\begin{definition}
We call \ts $\xi \in \textsu{$b$-creators}$ \ts \defn{lucky} \ts if
for all \ts $c$ \ts and for all \.
$\vv L\in T\big([b_1]\times\{1\},[b_2]\times\{2\},c\big)$,
all accepting paths of the computation \ts $h_i^{\xi_{\vv L}}(0^j)$ \ts
do not access the oracle at any point in \.
$\textsu{nzsources}(\xi)\sm \textsu{nzsource}(\xi_{\vv L})$,
nor any point in \.
$\textsu{dsinks}(\xi)\sm \textsu{dsink}(\xi_{\vv L})$,
nor any point in \.
$(\textsu{dsinks}(\xi)+1)\sm(\textsu{dsink}(\xi_{\vv L})+1)$ \.
or any directly adjacent nodes.
\end{definition}

\smallskip

This says that the accepting paths do not access the oracle at positions where sources or double sinks (or their loops) are not, but could potentially be.

\smallskip

\begin{claim}\label{cla:PPADwhpsuccess}
Let \ts $\xi$ \ts be sampled from \. $U_\textup{\textsu{$b$-creators}}$\..
Then \ts $\xi$ \ts is lucky  w.h.p.
\end{claim}

\begin{proof}
It is crucial to observe that for a fixed \. $\vv L\in T\big([b_1]\times\{1\},[b_2]\times\{2\},c\big)$,
we have that the $\vv L$-creation \ts $\xi_{\vv L}$ \ts is uniformly distributed from
\mbox{$(L_2-c)$\textsu{-creations}}.
We show that for a fixed \. $\vv L\in T\big([b_1]\times\{1\},[b_2]\times\{2\},c\big)$, we have \ts $\xi$ \ts
is lucky  w.h.p.
Since there are only constantly many \ts $\vv L$,
the claim immediately follows from the union bound.  Thus,
for the rest of the proof, we fix \ts $\vv L$.

We have that \ts $\xi_{\vv L}$ \ts is uniformly distributed.
The, the probability of an oracle access picking one of the forbidden positions is \.
$\leq \frac{8 (L_1+L_2)}{2^{n-1}} \ts \leq \ts\frac{1}{2^{n-2}}$\..
By Bernoulli's inequality we have \. $(1-\frac{1}{2^{n-2}})^{t_i(j)} \geq 1-\frac{t_i(j)}{2^{n-2}}$\ts.
Since~$n$ and~$j$ are polynomially related, this proves the claim.
\end{proof}

\subsubsection{Defining the set-instantiator}
For a $\vv a$-creation \ts $\psi$ \ts and a computation path \ts $\tau \in \{0,1\}^*$ \ts
of a computation \ts $h_i^\psi(0^j)$, let \.
$\textsu{perception}_1(\tau) \subseteq \textsu{nzsource}(\psi)$ \.
denote the set of accessed oracle positions that are nonzero sources or adjacent.
Let \. $\textsu{perception}_2(\tau) \subseteq \textsu{dsink}(\psi)$ \.
denote the set of accessed oracle positions that are double sinks or adjacent
(if the loop is accessed, this counts as an access of the double sink).

Let $\xi$ be a lucky $\vv b$-creator.
We interpret the list \ts
$\textup{\textsu{nzsources}}(\xi)$ \ts
as a bijection \. $\bij_1 : \mP(b_1)\to\textup{\textsu{nzsources}}(\xi)$.
Similarly, we interpret the list \ts
$\textup{\textsu{dsinks}}(\xi)$ \ts as a bijection \.
$\bij_2 : \mP(b_2)\to\textup{\textsu{dsinks}}(\xi)$.
Let \ts $S = S_{\textup{line}}$.
We set
\[
\inst_\SI(\vv s) \ := \ \xi_{\vv s}\..
\]
Finally, for \ts $\tau\in\{0,1\}^*$ \ts we set
\[
\textsu{perc}_\SI(\tau) \ := \
\big(\textsu{bij}_1^{-1}(\textsu{perception}_1(\tau)),\textsu{bij}_2^{-1}(\textsu{perception}_2(\tau))\big)
\]
if \ts $\tau$ \ts is an accepting path for the computation \ts $h_i^{\inst_\SI(\vv b)}(0^j)$,
and \ts $\textsu{perc}_\SI(\tau):=\top$, otherwise.

\smallskip

We can now prove that \ts $\SI$ \ts satisfies the requirements of Definition~\ref{def:setinstantiator}, which then proves Theorem~\ref{thm:PPADsetinstantiator}, because clearly \. $\problem{Multi(Source,\ts{}Excess)}(\inst_\SI(\vv s))=|\vv s|$.
Formally, we have:

\smallskip

\begin{proposition}
\label{pro:PPADsetinstantiatorcorrect}
For all \. $\vv s \in \mP(\vv b)_S$, we have:
$\tau\in\{0,1\}^*$ \. is an accepting path for the computation \.
$h^{\inst_\SI(\vv s)}(0^j)$ \. if and only if \. $\textup{\textsu{perc}}_\SI(\tau) \subseteq \vv s$.
\end{proposition}

\begin{proof}
Since \ts $\tau$ \ts is an accepting path of \ts $h^{\inst_\SI(\vv s)}(0^j)$,
and since \ts $\xi$ \ts is lucky, we conclude that \ts $\tau$ \ts
is also an accepting path of the computation \ts $h^{\inst_\SI(\vv b)}(0^j)$.
This implies that \.
$\textsu{perc}_\SI(\tau) = \textsu{bij}^{-1}(\textsu{perception}(\tau))$.
Clearly \. $\textsu{perception}(\tau) \subseteq \textsu{bij}(\vv s)$,
since otherwise \ts $\tau$ \ts would not even be a computational path of
\ts $h^{\inst_\SI(\vv s)}(0^j)$ \ts because of its oracle answers when
querying lonely nodes in \.
$\textsu{perception}(\tau) \sm \textsu{bij}(\vv s)$.
We conclude that \. $\textup{\textsu{perc}}_\SI(\tau) \subseteq \vv s$.

The argument is reversible:
let \. $\textup{\textsu{perc}}_\SI(\tau) \subseteq \vv s$, in particular $\textup{\textsu{perc}}_\SI(\tau) \neq \top$.
Then \ts $\tau$ \ts is an accepting path of the computation \ts $h^{\inst_\SI(\vv b)}(0^j)$.
Since \ts $\tau$ \ts is an accepting path of \ts $h^{\inst_\SI(\vv b)}(0^j)$,
and since \ts $\xi$ \ts is lucky, we conclude that \ts $\tau$ \ts
is also an accepting path of the computation \ts $h^{\inst_\SI(\vv s)}(0^j)$.
\end{proof}

\smallskip

\subsection{Combining set-instantiators for handling {\normalfont$\#\CLS-1$}}
\label{sec:CLSSetInstantiator}
Let \. $S_{\textup{CLS}} := \bigl\{\vv f \in \IN^3 \.:\. 2f_2=f_1+1, \. f_2\geq 1, \. f_3\geq 1\bigr\}$.

\smallskip

\begin{theorem}
\label{thm:CLSsetinstantiator}
Given a polynomial time nondeterministic Turing machine $M$.
For every $\vv b \in \IN^3$
there exists a threshold $j_0\in \IN$ such that for every $j \geq j_0$
and every $A_{<j}\in\{0,1\}^{<j}$
there exists a
set-instantiator $\SI$ against $(M,j,A_{<j},S=S_{\textup{CLS}},\vv b)$ such that
$\forall \vv s \text{ with } |\vv s|\in S: \problem{Multi(NZSource,\ts{}Excess,\ts{}Iter)}(\inst_\SI(\vv s)) = |\vv s|$.
\end{theorem}

\begin{proof}
The construction is basically a
combination of the set-instantiator for \.
\problem{Multi(NZSource,\ts{}Excess)},
and the set-instantiator for \ts
$\problem{Iter}$.
We can ignore the subtlety of the interaction between both problems that was introduced to be able to have instances of value~1, and for this set-instantiator we just care about instances that have value at least~2, which is reflected in the set \ts $S$ \ts by having \ts $f_2\geq 1$ \ts and \ts $f_3\geq 1$. This is similar to choosing \ts $S_{\textup{line}}$ \ts in~$\S$\ref{sec:SourceOrExcessSetInstantiator}: it already gives instances that are difficult enough.

We create a set-instantiator that creates pairs of instances that are glued together in
a non-sophisticated way by having the \ts \problem{Multi(NZSource,\ts{}Excess)} \ts
instance in the first half and the \ts $\problem{Iter}$ \ts instance in the second half.
Both set-instantiators were obtained by proving that a uniformly random creator is lucky
with high probability, because it just has to satisfy a finite number of constraints.
We put both sets of constraints together and readily obtain a creator with a finite
number of constraints that also with high probability is lucky. As usual, the
set-instantiator is then obtained by picking any single instance from the creator.
\end{proof}

\smallskip

\subsection{The halving separation}
\label{sec:halvingseparation}

In this section we prove the following result:

\smallskip

\begin{theorem}
\label{thm:halvingseparation}
We have \. $\sharpP \subseteq \sharpCOUNTALL[-PPA]\problemp{Leaf}/2$ \.
via a parsimonious relativizing reduction.
Moreover, there exists an oracle \ts $A$ \ts with respect to which \.
$\sharpCOUNTALL[-PPA]\problemp{Leaf}^A/2 \not\subseteq \sharpP^A$.
\end{theorem}

\begin{proof}
The inclusion \. $\sharpP^A \subseteq \sharpCOUNTALL[-PPA]\problemp{Leaf}^A/2$ \.
is of the same level of difficulty as the considerations in~$\S$\ref{sec:syntacticcountingsubclasses}.
Given a \ts \problem{CircuitSat} \ts instance we construct a
\ts \problem{AllLeaves} \ts instance of as follows:
Wherever the circuit yields True, we add two vertices that are connected by an edge.
We add no other vertices.
We obtain an instance of value~$2f$. We are done, because \ts $(2f)/2 = f$.

For the other direction, we use the
Diagonalization Theorem~\ref{thm:diagonalization}
and the set-instantiator for \ts \problem{AllLonely} \ts (see Theorem~\ref{thm:COUNTALLPPAsetinstantiator}),
with
\[
k \. = \. 1\., \quad S \. = \. 2\ts\IN\ts.
\]
Set \ts $\ka=k=1$ \ts and have no functions \ts $\zeta$.

We can now verify the preconditions of the Diagonalization Theorem~\ref{thm:diagonalization}:
\begin{enumerate}
\item $Z$ \ts contains the integer point~$0$,
\item $C'_S = \IQ_{\geq 0}$ \. lies Zariski-dense in \ts $\IQ$, and
\item the last point is a vacuous truth.
\end{enumerate}
We fix \. $\problem{Multiplicities} = \problem{AllLonely}$.
The necessary set-instantiators are given by Theorem~\ref{thm:COUNTALLPPAsetinstantiator}.
We set \ts $\varphi(f_1) = f_1/2$.
Clearly \ts $\varphi+I = f_1/2+\langle 0 \rangle$ \ts is \ts binomial-bad.
From the Diagonalization Theorem~\ref{thm:diagonalization}, it follows that for every~$M$
we have \. $p_A(0^j)\neq \#\acc_{M^A}(0^j)$.
Hence \. $\sharpCOUNTALL[-PPA]\problemp{Lonely}^A \neq \sharpP^A$.
The statement follows from the relativizing parsimonious equivalence between \ts
\problem{Lonely} \ts and \ts \problem{Leaf}, once we observe that \. $p_A \in
\sharpCOUNTALL[-PPA]\problemp{Lonely}^A/2$.

Let \ts $C_\bot$ \ts be an input to \ts \problem{AllLonely} \ts
(i.e., a circuit\footnote{Note that in the case of \problem{Lonely}
it is naturally a single circuit, not a list of circuits,
cf.~\cref{footnote:singlecircuit}.}) that yields output $\vv t$.
Let \ts $\ell$ \ts be the number of inputs of the circuit \ts $C_\bot$\ts.
Let \ts $\nu_j(w) = \textup{True}$ \ts if and only if \ts the string~$w$
has only zeros at positions \ts $j+1,j+2,\ldots$ \ts
Finally, let \ts $C_j$ \ts be an input to \ts $\problem{AllLonely}^A$ \ts
that describes the $(j-1)$-input circuit that consists of a single arity~$j$
oracle gate that takes a constant 1 into its first input, takes all inputs
into its remaining $(j-1)$ inputs, and forwards its output directly to the
circuit output.

We construct a circuit \. $\alpha(x) := D_{|x|}$ \.
with \ts $\max\{\ell,|x|\}$ \ts many inputs and
with oracle gates defined as follows \. $\big(w\in\{0,1\}^{\max\{\ell,|x|\}}\big)$:
\begin{eqnarray*}
D_j(w) & \, := \, & \Big((A(0^{|w|+1})=0) \textsf{ and }
\nu_\ell(w)
\textsf{ and }
C_\bot(w_1,\ldots,w_\ell)\Big)
\\
& \, \textsf{ and } \, &
\Big( (A(0^{|w|+1})=1)
\textsf{ and }
\nu_{|x|}(w)
\textsf{ and }
C_{|x|+1}(w_1,\ldots,w_{|x|})
\Big)
\end{eqnarray*}
We define the polynomially balanced relation \ts $R$ \ts via \ts $R(w,y)$ \.
if and only if \. $\problem{rAllLonely}(\alpha(w),y)$.
We set \ts $\beta$ \ts to be the identity function.
By definition, \ts $R \in \ComCla{rCOUNTALL-PPA}$.
The following function is thus in \ts $\sharpCOUNTALL[-PPA]$:
\[
\sum_{y \in\{0,1\}^*}R(x,y) \ = \
\sum_{y \in\{0,1\}^*}\problem{rAllLonely}(\alpha(x),y)
\]
which equals
\[
\begin{cases}
\problem{AllLonely}(D_{|x|})
& \textup{ if } \ A(0^{|x|+1})\ts =\ts 1,
\\
t & \textup{ otherwise.}
\end{cases}
\]
This function equals \ts $2p_A(w)$, which finishes the proof.
\end{proof}

\smallskip

\subsection{A set-instantiator for {\normalfont\sharpCOUNTALL{-\PPA}\problemp{Leaf}
}}
\label{sec:CountallLeafSetInstantiator}

\begin{theorem}
\label{thm:COUNTALLPPAsetinstantiator}
Let \ts $M$ \ts be a polynomial time nondeterministic Turing machine.
For every \ts $b \in \IN$,
there exists a threshold \ts $j_0\in \IN$, such that for every \ts
$j \geq j_0$ \ts and every \ts $A_{<j}\in\{0,1\}^{<j}$,
there exists a set-instantiator \ts $\SI$ \ts against \.
$\big(M,j,A_{<j},S\!=\!2\ts\IN,b\big)$ \. such that \.
$\problem{AllLonely}(\inst_\SI(s)) = |s|$ \. for all \ts $s$ \ts with \ts
$|s|\in S$.
\end{theorem}
Intuitively, this says that for all our purposes $\problem{AllLonely}$ behaves exactly as $\sharpP_{\textup{even}}$.
The rest of this section is devoted to proving this theorem.
We will first define a \emph{creator} whose \emph{creations} will be the instantiations in the end, but the creator is not limited to a single creation for each set. We will then define the set-instantiator from the creator by picking for every subset $\vv s$ with $|\vv s|\in 2\IN$ just any one of its creations for~$\vv s$.

We crucially use the parsimonious polynomial-time reductions from \problem{AllLonely} to \problem{AllLeaves} and back, so we can focus on \problem{AllLonely} here.
The reductions are gadget-based (with extremely simple gadgets), but \problem{AllLonely} fits perfectly for our proof technique, while it is slightly more trouble to work with \problem{AllLeaves}.

\subsubsection{Creations}
Having \ts $2^{j-1}$ \ts many bits available, we can encode in a standard way a function \.
$\psi:[2^{n}]\to[2^{n}]$, where $n$ and~$j$ are polynomial related.
We ignore any extra bits that are not needed for this encoding.
A \defn{paired node} in \ts $\psi$ \ts is an \ts $x\in[2^n]$ \ts for which \ts $y \in[2^n]\sm\{x\}$ \ts
exists such that \ts $\psi(x)=y$ \ts and \ts $\psi(y)=x$.
A \defn{lonely node} is a node that is not paired.
Since \ts $2^n$ \ts is even, clearly the number of lonely nodes is always even.

For \ts $a \in 2\IN$, a function \ts $\psi$ \ts is called an \defn{$a$-creation}  \ts
if it consists of exactly $a$~many lonely nodes, and for each lonely node~$x$ we have \ts $\psi(x)=x$.
Let \ts $\textsu{$a$-creations}$ \ts denote the set of all $a$-creations.
Let \. $\textsu{$a^-$-creations} \ts := \ts \bigcup_{1 \leq c \leq a}\textsu{$c$\,-creations}$.
Finally, let \ts $\textsu{lonely}(\psi)\subseteq[2^n]$ \ts denote the set of all lonely nodes.

\subsubsection{Lucky creators}
In this section we introduce the concept of a \defn{creator}.
A creator is a slightly less restrictive version of a set-instantiator.
For each blueprint a creator outputs a creation, where a blueprint is slightly more general than the set $\vv s$ for a set-instantiator, but serves the same purpose.

For a set~$X$, let \ts $\binom{X}{2 \ts, \ts \ldots \ts , \ts 2}$ \ts denote the set of all set partitions of $X$ into pairs.
We call these elements \defn{pairings}.
For \ts $b\in 2\IN$, a \defn{$b$-creator} \. $\xi$ \ts is a
finite ordered set \ts $\textsu{lonelies}(\xi)$ \ts of~$b$ many distinct elements
of \ts $[2^n]$ \ts together with a pairing \.
$\textsu{pairing}(\xi)\in\binom{[2^n]\sm \textsu{lonelies}(\xi)}{2 \ts, \ts \ldots \ts , \ts 2}$.
Let $\textsu{$b$-creators}$ denote the set of all $b$-creators.

Given a $b$-creator \ts $\xi$, a subset \ts $L\subseteq[b]$ \ts with \ts $|L|$ \ts even,
and a pairing \. $o \in \binom{L}{2 \ts, \ts \ldots \ts , \ts 2}$, we obtain a $(b-|L|)$-creation \ts $\xi_{L,o}$ \ts
by taking the pairing \ts $\textsu{pairing}(\xi)$, enlarging it pairing up the nodes as indicated by~$o$ (i.e., if $a$ and $b$ are paired in $o$, then \ts $\textsu{lonelies}(\xi)_a$ \ts and \ts $\textsu{lonelies}(\xi)_b$ \ts get paired), and setting \ts $\xi_{L,o}(x)=x$ \ts for all nodes that are still unpaired.

\smallskip

\begin{definition}
We call \ts $\xi \in \textsu{$b$-creators}$ \ts \defn{lucky} \ts if
for all \ts
$L\subseteq[b]$ \ts with \ts $|L|$ \ts even,
and all pairings \. $o \in \binom{L}{2,\ldots,2}$, we have:
all accepting paths of the computation \ts $h_i^{\xi_{L,o}}(0^j)$ \ts
do not access the oracle at any point in \.  $\textsu{lonelies}(\xi)\sm \textsu{lonely}(\xi_{L,o})$.
\end{definition}

\smallskip

In other words, this says that the accepting paths do not access the oracle at positions
where lonely nodes are not, but could potentially be.

\smallskip

\begin{claim}
If \ts $\xi$ \ts is sampled from \ts $U_\textup{\textsu{$b$-creators}}$\ts, then \ts $\xi$ \ts is lucky w.h.p.
\end{claim}

\begin{proof}
We describe a way of sampling from \ts $U_\textup{\textsu{$b$-creators}}$.
Let \ts $\binom{X}{b}_\textup{ordered}$ \ts
denote the set of length~$b$ lists of distinct elements from~$X$.
We first sample \ts $\textsu{lonelies}(\xi)$ \ts from \ts $U_{\binom{[2^n]}{b}_\textup{ordered}}$.
We then sample \ts $\textsu{pairing}(\xi)$ \ts from \ts
$U_{\binom{[2^n]\sm \textsu{lonelies}(\xi)}{2,\ldots,2}}$.
This process samples a random variable \ts $\xi$ \ts
from \ts $U_\textup{\textsu{$b$-creators}}$.

Observe that for a fixed \ts $(L,o)$, where \ts $L\subseteq [b]$ \ts and \ts $o\in\binom{L}{2,\ldots,2}$,
we have that \ts $\xi_{L,o}$ \ts is uniformly distributed from \ts \mbox{$(b-|L|)$\textsu{-creations}}.
It remains to show that for such fixed \ts $(L,o)$, we have that \ts $\xi$ \ts is lucky w.h.p.
Note that since there are only constantly many such \ts $(L,o)$, the claim immediately follows
from the union bound.
Hence, for the rest of the proof fix $(L,o)$.

We have that \ts $\xi_{L,o}$ \ts is uniformly distributed.
Therefore, the set of \ts $\textsu{lonelies}(\xi)\sm\textsu{lonely}(\xi_{L,o})$ \ts
is uniformly distributed over a set of size \ts $\binom{2^n-b+|L|}{|L|}$.
By the union bound, an oracle query accesses such a position with probability \. $\leq \frac{|L|}{2^n-b+|L|}$.
By Bernoulli's inequality, we have
$$\left(1-\frac{|L|}{2^n-b+|L|}\right)^{t_i(j)} \ \geq \ 1 \. - \, \frac{|L| \ts \cdot \ts t_i(j)}{2^n-b+|L|}\..
$$
Since $n$ and $j$ are polynomially related, this proves the claim.
\end{proof}

\subsubsection{Defining the set-instantiator}
For an $a$-creation \ts $\psi$ \ts and a computation path \ts $\tau \in \{0,1\}^*$ \ts
of a computation \ts $h_i^\psi(0^j)$, let \.
$\textsu{perception}(\tau) \subseteq \textsu{lonely}(\psi)$ \.
denote the set of accessed oracle positions that are lonely nodes.

Let \ts $\xi$ \ts be a lucky $b$-creator.
We interpret the list \ts $\textup{\textsu{lonelies}}(\xi)$ \ts
as a bijection \. $\bij : [b]\to\textup{\textsu{lonelies}}(\xi)$.
Let \ts $S = 2\ts\IN$.
For every \ts $\vv s \in \mP(\vv b)_S$
(actually, \ts $\vv s = s$ \ts is univariate),
we let \ts $L_s := [b] \sm s$ \ts
and we choose a \ts
$o_{s} \in \binom{L}{2,\ldots,2}$.
We set
\[
\inst_\SI(s) \. := \. \xi_{L_s,o_{s}}\..
\]
Finally, for $\tau\in\{0,1\}^*$, let
\[
\textsu{perc}_\SI(\tau) \ := \ \begin{cases}
\textsu{bij}^{-1}(\textsu{perception}(\tau)) & \text{ if \ts $\tau$ \ts is an accepting path of the computation \ts $h_i^{\inst_\SI(\vv b)}(0^j)$\ts,}
\\
\top & \text{ otherwise.}
\end{cases}
\]

The rest of this section is devoted to proving that \ts $\SI$ \ts satisfies the requirements of Definition~\ref{def:setinstantiator}, which then proves Theorem~\ref{thm:COUNTALLPPAsetinstantiator}, because clearly \ts $\problem{Lonely}(\inst_\SI(\vv s))=|\vv s|$.
Formally, we have the following result.

\begin{proposition}
For all \ts $\vv s \in \mP(\vv b)_S$, we have: \ts
$\tau\in\{0,1\}^*$ \ts is an accepting path for the computation \ts
$h^{\inst_\SI(\vv s)}(0^j)$ \. if and only if \. $\textup{\textsu{perc}}_\SI(\tau) \subseteq \vv s$.
\end{proposition}

\begin{proof}
Since \ts $\tau$ \ts is an accepting path of \ts $h^{\inst_\SI(\vv s)}(0^j)$,
and since \ts $\xi$ \ts is lucky, we conclude that \ts $\tau$ \ts
is also an accepting path of the computation \ts $h^{\inst_\SI(\vv b)}(0^j)$.
This implies that \.
$\textsu{perc}_\SI(\tau) = \textsu{bij}^{-1}(\textsu{perception}(\tau))$.
Clearly \ts $\textsu{perception}(\tau) \subseteq \textsu{bij}(\vv s)$,
because otherwise \ts $\tau$ \ts would not even be a computational path of
\ts $h^{\inst_\SI(\vv s)}(0^j)$ \ts because of its oracle answers when
querying lonely nodes in \ts $\textsu{perception}(\tau) \sm \textsu{bij}(\vv s)$.
We conclude \. $\textup{\textsu{perc}}_\SI(\tau) \subseteq \vv s.$

The argument is reversible.  Indeed, let \.
$\textup{\textsu{perc}}_\SI(\tau) \subseteq \vv s$, so in particular \ts
$\textup{\textsu{perc}}_\SI(\tau) \neq \top$.
Then \ts $\tau$ \ts is an accepting path of the computation \ts
$h^{\inst_\SI(\vv b)}(0^j)$. Since \ts $\tau$ \ts is an accepting path
of \ts $h^{\inst_\SI(\vv b)}(0^j)$, and since \ts $\xi$ \ts is lucky,
we conclude that \ts $\tau$ \ts is also an accepting path of the
computation \ts $h^{\inst_\SI(\vv s)}(0^j)$, as desired.
\end{proof}

\smallskip

\subsection{The unbalanced flow separation}
\label{sec:unbalancedflowseparation}
In this section we stud the \ts \problem{BipartiteUnbalance} \ts problem, and prove that
$$\sharpCOUNTGAP\problemp{BipartiteUnbalance}^A \, \not\subseteq \, \sharpCOUNTALL[-PPA]\problemp{Leaf}^A/2\..
$$
For clarity of exposition, in notation of $\S$\ref{subsubsec:BipartiteUnbalance},
the vertices in $V_+$ are called \defn{white} and the vertices in $V_-$ are called \defn{dark}.

\smallskip

While all other classes in Figure~\ref{fig:inclusions} are contained in
$$\PolynP_{\textup{univariate}}\ := \ \big\{\varphi(\sharpP) \, : \, \varphi \textup{ univariate}\big\},
$$
it is not clear if the problem \problem{BipartiteUnbalance} lies in that class.
It does lie in
$$\GapP \, \subseteq \, \PolynP_{\textup{bivariate}} \ := \ \big\{\varphi(\vv\sharpP) \,:\, \varphi \textup{ bivariate}\big\}.
$$
To put it into the complexity class inclusion diagram in Figure~\ref{fig:inclusions}, the definition of its corresponding complexity class \ts $\sharpCOUNTGAP\problemp{BipartiteUnbalance}$ \ts is done in analogy to the definitions in~$\S$\ref{sec:countingclassesandTFNPINTRO}.

Consider the two relations
$$\problem{rBipartiteUnbalanceWhite} \quad \text{and} \quad
\problem{rBipartiteUnbalanceDark}\ts,
$$
defined as follows.
Let \ts $(C,w)\in\problem{rBipartiteUnbalanceWhite}$ \ts if and only if \ts $C$ \ts describes a graph in which~$w$ is a white vertex.
Let the relation \ts \problem{rBipartiteUnbalanceDark} \ts be defined analogously.

Let \ts $\ComCla{rCOUNTGAP}\problemp{BipartiteUnbalance}$ \ts be the set of
pairs of polynomially balanced relations \ts $R=(R_1,R_2)$ \ts for which a pair
\ts $(\alpha,\beta)$ \ts of polynomial-time computable maps exists with
$(C,\beta(x)) \in R_1$ \ts if and only if \ts $(\alpha(C),x)\in \problem{rBipartiteUnbalanceWhite}$
and \ts
$(C,\beta(x)) \in R_2$ \ts if and only if \ts $(\alpha(C),x)\in \problem{rBipartiteUnbalanceDark}$.
In addition, we require that \ts $(C,\beta(x))\in R_1$ \ts and \ts $(C,\beta(y))\in R_1$ \ts implies \ts $x=y$,
and we require that \ts $(C,\beta(x))\in R_2$ \ts and \ts $(C,\beta(y))\in R_2$ \ts implies \ts $x=y$.

Let \ts $\ComCla{rP}\times\ComCla{rP}$ \ts denote the set of pairs of polynomially balanced relations that can be evaluated in polynomial time.
By definition, \. $\ComCla{rCOUNTGAP}\problemp{BipartiteUnbalance} \subseteq \ComCla{rP}\times\ComCla{rP}$.
We attach an oracle completely analogous to~\S\ref{sec:countingclassesandTFNPINTRO},
to obtain \. $\ComCla{rCOUNTGAP}\problemp{BipartiteUnbalance}^A \subseteq \ComCla{rP}^A\times\ComCla{rP}^A$.
We define the corresponding counting class $\sharpCOUNTGAP\problemp{BipartiteUnbalance}^A$ via
\begin{eqnarray*}
 && f \in \sharpCOUNTGAP\problemp{BipartiteUnbalance}^A \qquad \textup{ if and only if }
 \\
 &&
\exists R\in\ComCla{rCOUNTGAP}\problemp{BipartiteUnbalance}^A \ : \ f(w) \ = \ \sum_{y\in\{0,1\}^*} \big(R_1(w,y)-R_2(w,y)\big).
\end{eqnarray*}
By definition and Proposition~\ref{p:flow-graph} we have \.
$\sharpCOUNTGAP\problemp{BipartiteUnbalance}^A\subseteq \GapP_{\geq 0}^A$.
The main result of this subsection (Proposition~\ref{pro:unbalancedflowseparation}) follows from the following theorem.

\smallskip

\begin{theorem}
\label{thm:unbalancedflowseparation}
There exists \ts $A\subseteq\{0,1\}^*$ \ts with \.
$$2\ts\sharpCOUNTGAP\problemp{BipartiteUnbalance}^A \not\subseteq \sharpP^A\ts.
$$
\end{theorem}
\smallskip

The following proposition is the inclusion and separation shown in Figure~\ref{fig:inclusions}.

\smallskip

\begin{proposition}
\label{pro:unbalancedflowseparation}
$\sharpCOUNTALL[-PPA]\problemp{Leaf}/2 \subseteq
\sharpCOUNTGAP\problemp{BipartiteUnbalance}$
via a relativizing parsimonious reduction.
Moreover, there exists an oracle $A$ with respect to which
$\sharpCOUNTGAP\problemp{BipartiteUnbalance}^A \not\subseteq \sharpCOUNTALL[-PPA]\problemp{Leaf}^A/2$.
\end{proposition}

\begin{proof} 
We prove the first inclusion
via a relativizing reduction that preserves the function value (in the same way a parsimonious reduction preserves the function value).
Given a \ts \problem{Leaf} \ts instance, replace each vertex that has nonzero degree by a white vertex and replace each edge by a black vertex connecting the corresponding two white vertices.
The value of the resulting \ts \problem{BipartiteUnbalance} \ts instance is exactly the number of connected components in the original instance.

The (non-)inclusion in the second part follows directly from dividing
Theorem~\ref{thm:unbalancedflowseparation} by~2.  Formally,
for all oracles $A$, we have \.
$\sharpCOUNTALL[-PPA]\problemp{Leaf}^A \subseteq \sharpP_{\textup{even}}^A$\ts.
Therefore, we have \. $
\sharpCOUNTALL[-PPA]\problemp{Leaf}^A/2 \subseteq \sharpP_{\textup{even}}^A/2$\ts.

We use Theorem~\ref{thm:unbalancedflowseparation} to obtain an oracle $A$ such that
$$2\.\sharpCOUNTGAP\problemp{BipartiteUnbalance}^A \. \not\subseteq \. \sharpP^A\ts.
$$
In particular, we have
$$2\.\sharpCOUNTGAP\problemp{BipartiteUnbalance}^A\. \not\subseteq \. \sharpP^A_{\textup{even}}\ts.
$$
We divide by 2:
$$\sharpCOUNTGAP\problemp{BipartiteUnbalance}^A \. \not\subseteq \.  \sharpP^A_{\textup{even}}/2\ts.
$$
But if
$$\sharpCOUNTGAP\problemp{BipartiteUnbalance}^A \. \subseteq \. \sharpCOUNTALL[-PPA]\problemp{Leaf}^A/2$$
were true,
then
$$\sharpCOUNTGAP\problemp{BipartiteUnbalance}^A \. \subseteq \. \sharpP^A_{\textup{even}}/2$$
were also true, which we know is false.
\end{proof}

\smallskip

\begin{proof}[Proof of Theorem~\ref{thm:unbalancedflowseparation}]
We use the
Diagonalization Theorem~\ref{thm:diagonalization}
and the set-instantiator for \.
\problem{Multi(1\ts{}Source,\ts{}3\ts{}Sink)}, see Theorem~\ref{thm:BipUnbsetinstantiator}),
with
\[
k=2, \quad S \. = \. \big\{\vv f \in \IN^2 \. : \. f_1=3f_2+6\big\}.
\]
We switch the positions to get the set-instantiator for \.
\problem{Multi(3\ts{}Sink,\ts{}1\ts{}Source)},
with
\[
k=2, \quad S \. = \. \big\{\vv f \in \IN^2 \. : \. f_2=3f_1-6\big\}.
\]
We set \ts $\ka=1$ \ts and \ts $\zeta_2(v_1) = 3v_1-6$, so that \ts $I = \langle  3v_1-6-f_2\rangle$ \ts and \ts
$Z = \big\{\vv f \in \IQ^2 \.:\. 3f_1-6-f_2=0\big\}$.
We verify the preconditions of the Diagonalization Theorem~\ref{thm:diagonalization}:

\begin{enumerate}
\item $Z$ contains the integer point \ts $(0,6)$,
\item $C'_S$ \ts lies Zariski-dense in \ts $\IQ$ because \ts $(f_1,3 f_1-6) \in S$ \ts for all \ts $f_1 \in \IN_{\geq 2}$\ts, and
\item $\vv v = (1)$ \ts satisfies \ts $3v_1>0$.
\end{enumerate}
We fix \ts $\problem{Multiplicities} =
\problem{Multi(3\ts{}Sink,\ts{}1\ts{}Source)}$.
The necessary set-instantiators are given by Theorem~\ref{thm:BipUnbsetinstantiator}.
We set \ts $\varphi(f_1,f_2) = 8-4f_1+2f_2$, which is the value of a \ts \problem{2BipartiteUnbalance} \ts
instance with one 6-source, \ts $f_1$ \ts many 3-sinks, and \ts $f_2$ \ts many 1-sources,
as created by the set-instantiator in~$\S$\ref{sec:BipartiteUnbalanceSetInstantiator}.

We verify that \ts $\varphi+I = -4f_1+2f_2+8+\langle 3f_1-f_2-6 \rangle$ \ts
is binomial-bad by using the Polyhedron Theorem~\ref{thm:polyhedron}.
In the notation of the Polyhedron Theorem~\ref{thm:polyhedron}, we have \ts
$\varphi' = 2f_1 -4$, \ts $\delta=1$, so that \ts $\IO_\delta$ \ts
has only 2 elements.  Hence, the polyhedron \ts $\mathcal P(\varphi,\zeta)$ \ts
will be defined using 2 equations.

Let $(1,v_1)$ be the order of basis vectors of basis vectors of $\IQ[\vv v]_{\leq 1}$. We express the 3 polynomials $y_{\vv e}$ over this basis and obtain coefficient vectors $(1,0)$, $(0,1)$, $(-6,3)$.
The polyhedron given by
\[
\begin{array}{rrrrrr}
 x_0 &     &  -6x_2    & \. = \. &  -4 \\
     & x_1 &  +3x_2    &\. = \. & 2 \\
 \multicolumn{3}{r}{x_0, \. x_1, \. x_2}  &\. \geq \. & 0 \\
\end{array}
\]
has no integer points.  Therefore, \ts $\varphi+I$ \ts is binomial-bad by the Polyhedron Theorem~\ref{thm:polyhedron}.\footnote{Note that the LP relaxation has solutions, for example \ts $(0,0,\frac 2 3)$, so a if a computation path could not just accept or reject, but accept with a fractional contribution, then this would work out.}
From the Diagonalization Theorem~\ref{thm:diagonalization} it follows that for every $M$ we have \. $p_A(0^j)\neq \#\acc_{M^A}(0^j)$.
Hence \ts $p_A\notin\sharpP^A$.
It remains to show that \.
$p_A \ts \in \ts
2\ts\sharpCOUNTGAP\problemp{BipartiteUnbalance}^A$.

Let $C_\bot$ be an input to \ts \problem{Multi(3\ts{}Sink,\ts{}1\ts{}Source)} \ts
(i.e., a circuit merged from multiple circuits that describe the exponentially large graph,
cf.~\cref{footnote:singlecircuit}) that yields output~$\ts\vv t$.
Let \ts $\ell$ \ts be the number of inputs of the circuit \ts $C_\bot$.
Let \ts $\nu_j(w) = \textup{True}$ \ts if and only if \ts
the string \ts $w$ \ts has only zeros at positions \ts $j+1,j+2,\ldots$

Finally, let \ts $C_j$ \ts be an input to \ts $\problem{Multi(3\ts{}Sink,\ts{}1\ts{}Source)}^A$ \ts
that describes the $(j-1)$-input circuit that consists of a single arity $j$ oracle gate that
 takes a constant~1 into its first input, takes all inputs into its remaining \ts $j-1$ \ts inputs,
 and forwards its output directly to the circuit output.

We construct a circuit \ts $\alpha(x) := D_{|x|}$ \ts
with \ts $\max\{\ell,|x|\}$ \ts many inputs, and
with oracle gates as follows \. $(w\in\{0,1\}^{\max\{\ell,|x|\}})$:
\begin{eqnarray*}
D_j(w) & := & \Big((A(0^{|w|+1})=0) \textsf{ \ \  and \ \  }
\nu_\ell(w)
\textsf{ \ \ and \ \ }
C_\bot(w_1,\ldots,w_\ell)\Big)
\\
&\textsf{ \ \  and \ \ }&
\Big( (A(0^{|w|+1})=1)
\textsf{ \ \ and \ \ }
\nu_{|x|}(w)
\textsf{ \ \  and \ \ }
C_{|x|+1}(w_1,\ldots,w_{|x|})
\Big)
\end{eqnarray*}
We define the polynomially balanced relations \ts $R=(R_1,R_2)$ \ts as follows: \ts
let \ts $R_1(w,y)$ \ts if and only if \ts $\problem{rBipartiteUnbalanceWhite}(\alpha(w),y)$.
Similarly, let \ts $R_2(w,y)$ \ts if and only if \ts $\problem{rBipartiteUnbalanceDark}(\alpha(w),y)$.
We set $\beta$ to be the identity function.

By definition, $R \in \ComCla{rCOUNTGAP\problemp{BipartiteUnbalance}}$.
Therefore, the following function is in \. $\sharpCOUNTGAP\problemp{BipartiteUnbalance}$:
\[\aligned
& \sum_{y \in\{0,1\}^*}R(x,y) \\
& = \ \sum_{y \in\{0,1\}^*}(\problem{rBipartiteUnbalanceWhite}-
\problem{rBipartiteUnbalanceDark})(\alpha(x),y) \\
& = \
\begin{cases}|V_+|(D_{|x|})-|V_-|(D_{|x|})
& \textup{ if \ \ } A(0^{|x|+1})=1,
\\
t_1-t_2 & \textup{ otherwise.}
\end{cases}
\endaligned
\]
Thus, this function equals \ts $p_A(w)+1$, as desired. This finishes the proof.\footnote{We remark that the
graph gadgets that make this separation work must be carefully chosen, and
for several other choices of gadgets we do not get a separation.}
\end{proof}

\smallskip

\subsection{A set-instantiator for {\ts{}}{\normalfont$2\ts$\sharpCOUNTGAP\problemp{BipartiteUnbalance}
}}
\label{sec:BipartiteUnbalanceSetInstantiator}

\begin{theorem}
\label{thm:BipUnbsetinstantiator}
Let \ts $M$ \ts be a polynomial time nondeterministic Turing machine. Fix $k=2$.
Then, for every \ts $\vv b \in \IN^2$,
there exists a threshold \ts $j_0\in \IN$, such that for every \ts $j \geq j_0$
and every \ts $A_{<j}\in\{0,1\}^{<j}$, there exists a set-instantiator \ts $\SI$ \ts
against \. $(M,j,A_{<j},S=S_{\textup{BUline}},\vv b)$, so that \.
 $\problem{Multi(1\ts{}Source,\ts{}3\ts{}Sink)}(\inst_\SI(\vv s)) \ts = \ts |\vv s|$ \.
 for all \ts $\vv s$ \ts with \ts $|\vv s|\in S$.
\end{theorem}

\begin{proof}
The set-instantiator is constructed analogously to the construction for \problem{Multi(NZSource,\ts{}Excess)} in Theorem~\ref{thm:PPADsetinstantiator}, but with the source at zero, while the other sources, the directed edges
and the double sinks are replaced by graph gadgets.
The rest of the proof then works analogously.

We have slightly different parameters in this case.  Let \.
$$S_{\textup{BUline}} \ := \ \big\{\ts (f_1,f_2) \in \IN^2 \, :\, f_1-3f_2+6 \ts = \ts 0\ts \big\}.
$$
The graph gadgets are as follows, where the dark nodes have degree at least as high as all adjacent white nodes.

\smallskip

First, a \emph{directed edge} is replaced by a \emph{double edge}, where each edge connects a white node of degree~2 or~3,
and a dark node of degree~3.  An \emph{indegree} $=$ \emph{outdegree}~1 node is replaced by
the following gadget:
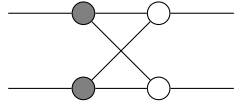
\begin{center}
\begin{tikzpicture}[every node/.style={draw,circle,inner sep=0pt,minimum size=0.3cm}]
\node[fill=gray] (a) at (0,0) {};
\node[fill=gray] (b) at (0,1) {};
\node            (c) at (1,0) {};
\node            (d) at (1,1) {};
\draw ($(a)+(-1,0)$) -- (a);
\draw ($(b)+(-1,0)$) -- (b);
\draw ($(c)+(+1,0)$) -- (c);
\draw ($(d)+(+1,0)$) -- (d);
\draw (a) -- (c);
\draw (a) -- (d);
\draw (b) -- (c);
\draw (b) -- (d);
\end{tikzpicture}
\end{center}
\noindent Note that \. $2\ts{}(\#\textup{white nodes}-\#\textup{dark nodes})=0$ \. in this case.

\smallskip

Next, the \emph{zero source} is replaced by the following gadget (in the original directed problem this would be a 6-fold source):
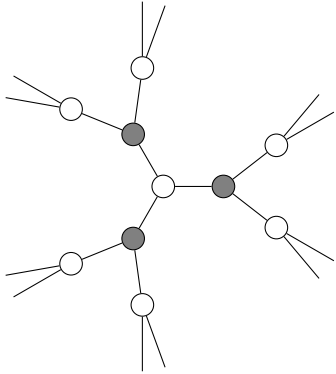
\begin{center}
\begin{tikzpicture}[scale=0.8, every node/.style={draw,circle,inner sep=0pt,minimum size=0.3cm}]
\node            (v) at (0,0) {};
\node[fill=gray] (v1) at (  0:1) {};
\node[fill=gray] (v2) at (120:1) {};
\node[fill=gray] (v3) at (240:1) {};
\node             (w11) at ( 20:2) {};
\node             (w12) at (340:2) {};
\node             (w21) at (100:2) {};
\node             (w22) at (140:2) {};
\node             (w31) at (220:2) {};
\node             (w32) at (260:2) {};
\node[draw=none]  (x111) at ($(w11)+(50:1.3)$) {};
\node[draw=none]  (x112) at ($(w11)+(30:1.3)$) {};
\node[draw=none]  (x121) at ($(w12)+(310:1.3)$) {};
\node[draw=none]  (x122) at ($(w12)+(330:1.3)$) {};
\node[draw=none]  (x211) at ($(w21)+(90:1.3)$) {};
\node[draw=none]  (x212) at ($(w21)+(70:1.3)$) {};
\node[draw=none]  (x221) at ($(w22)+(170:1.3)$) {};
\node[draw=none]  (x222) at ($(w22)+(150:1.3)$) {};
\node[draw=none]  (x311) at ($(w31)+(210:1.3)$) {};
\node[draw=none]  (x312) at ($(w31)+(190:1.3)$) {};
\node[draw=none]  (x321) at ($(w32)+(290:1.3)$) {};
\node[draw=none]  (x322) at ($(w32)+(270:1.3)$) {};
\draw (v) -- (v1);
\draw (v) -- (v2);
\draw (v) -- (v3);
\draw (v1) -- (w11);
\draw (v1) -- (w12);
\draw (v2) -- (w21);
\draw (v2) -- (w22);
\draw (v3) -- (w31);
\draw (v3) -- (w32);
\draw (w11) -- (x111);
\draw (w11) -- (x112);
\draw (w12) -- (x121);
\draw (w12) -- (x122);
\draw (w21) -- (x211);
\draw (w21) -- (x212);
\draw (w22) -- (x221);
\draw (w22) -- (x222);
\draw (w31) -- (x311);
\draw (w31) -- (x312);
\draw (w32) -- (x321);
\draw (w32) -- (x322);
\end{tikzpicture}
\end{center}
\noindent Note that \. $2\ts(\#\textup{white nodes}-\#\textup{dark nodes})=8$ \. in this case.
The \emph{source} is replaced by the following simple gadget:
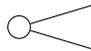
\begin{center}
\begin{tikzpicture}[every node/.style={draw,circle,inner sep=0pt,minimum size=0.3cm}]
\node (a) at (0,0) {};
\draw ($(a)+(1,0.3)$) -- (a);
\draw ($(a)+(1,-0.3)$) -- (a);
\end{tikzpicture}
\end{center}
\noindent Note that $2\ts (\#\textup{white nodes}-\#\textup{dark nodes})=2$.

\smallskip

Finally, there are no double sink nodes, but for a triple sink we use the following gadget:
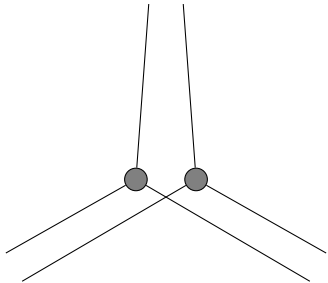
\begin{center}
\begin{tikzpicture}[every node/.style={draw,circle,inner sep=0pt,minimum size=0.3cm}]
\node[fill=gray] (a) at (-0.4,0) {};
\node[fill=gray] (b) at (0.4,0) {};
\node[draw=none] (x1) at (85:2.5) {};
\node[draw=none] (x2) at (95:2.5) {};
\node[draw=none] (x3) at (205:2.5) {};
\node[draw=none] (x4) at (215:2.5) {};
\node[draw=none] (x5) at (325:2.5) {};
\node[draw=none] (x6) at (335:2.5) {};
\draw (x1) -- (b);
\draw (x2) -- (a);
\draw (x3) -- (a);
\draw (x4) -- (b);
\draw (x5) -- (a);
\draw (x6) -- (b);
\end{tikzpicture}
\end{center}
\noindent Note that \. $2\ts (\#\textup{white nodes}-\#\textup{dark nodes})=-4$ \. in this case.
This finishes the proof.
\end{proof}

\medskip

\section{Open problems}\label{s:open}
\newcounter{openproblemscounter}

\subsection{Counting subgraphs} \label{ss:open-log-concavity}

Let \ts $G=(V,E)$ \ts be a simple graph.  Denote by \ts $m_k=m_k(G)$ \ts
the number of \defn{$k$-matchings}, i.e.\ sets of $k$ edges with disjoint
vertices.  As we mentioned in the introduction, the Heilmann--Lieb theorem
states~\cite{HL72}:
$$
\de(k,G) \, := \, m_k^2 \. - m_{k+1}\ts m_{k-1} \. \ge \. 0 \quad \text{for all} \ k\ge 1,
$$
where we assume \ts $m_0=1$. Clearly, $\de\in \GapPP$\ts.

\smallskip

\begin{proposition}[\cite{Pak19}]
\label{prop:mathing}
In the notation above, $\de\in \SP$\ts.
\end{proposition}

\smallskip

The proof in~\cite{Pak19} is an easy adaptation of the proof
by Krattenthaler~\cite{Kra96}.  Many other log-concavity problems
remain open.  Proposition~\ref{prop:log-concavity} suggests some of
them might not be in~$\SP$.

\smallskip

\begin{theorem}[{\cite{AHK18}}]
Let \ts $G=(V,E)$ \ts be a simple graph, and denote by \ts
$\tau_k=\tau_k(G)$ \ts the number of spanning forests in~$G$
with $k$ edges.  Then:
$$
\tau_k^2 \. \ge \tau_{k+1}\ts \tau_{k-1} \quad \text{for all} \ k\ge 1.
$$
\end{theorem}

\smallskip

This is a celebrated result by Adiprasito, Huh and Katz, which holds
for all matroids, not just truncations of the graphical matroids.
We refer to~\cite{Huh18} for a survey of the algebraic approach,
and to~\cite{CP21,CP22} for an elementary approach using linear algebra.

\smallskip

\noindent\refstepcounter{openproblemscounter}$(\theopenproblemscounter)$ \.
Is it true that \. $\tau_k^2 \. \geqslant_\# \tau_{k+1}\ts \tau_{k-1}$?
This remains open and is a major challenge in the area, see e.g.~\cite{Pak19} and \cite[$\S$17.17]{CP21}.

\smallskip

For many other unimodality and log-concavity results, see \cite{Bra15,Huh18,Sta80}.
Following \cite{Pak19} and the approach in this paper, most of these results
can be rephrased as questions about corresponding counting functions in~$\SP$.

\smallskip

\subsection{Counting linear extensions} \label{ss:open-LE}
Let \ts $P=(X,\prec)$ \ts be a poset on \ts $|X|=n$ \ts elements.
A \defn{linear extension} is a order-preserving bijection \ts $f: X\to [n]$,
i.e.\ \ts $f(x)<f(y)$ \ts for all \ts $x\prec y$.  As in \ts $\S$\ref{ss:intro-basics-SP}\eqref{motivational:LE},
denote by \ts $\EE(P)$ \ts and \ts $e(P):=|\EE(P)|$ \ts the set and the number
of linear extensions, respectively.  The problem of computing $e(P)$ is famously
$\SP$-complete~\cite{BW91}, even in special cases of posets of height two,
or of width two~\cite{DP18}.

There are some surprising combinatorial inequalities for the numbers of linear
extensions.  For each element \ts $x \in X$, let \.
$B(x) := \big\{y \in X \. : \. y \succcurlyeq x  \big\}$ \.
be the \defn{upper order ideal} generated by~$x$, and let \. $b(x):=|B(x)|$.

\smallskip

\begin{theorem}[{\defn{Bj\"orner--Wachs inequality}~\cite{BW89}}]\label{t:BW-ineq}
In the notation above, we have:
\begin{equation*}
e(P) \,\. \prod_{x \in X} b(x) \ \geqslant_\#  \  \. n!
\end{equation*}
\end{theorem}

\smallskip

The Bj\"orner--Wachs inequality is an equality for ordered forests, where it is called the
\defng{hook-length formula}, see e.g.~\cite{SY89,Sta12}.  The original
proof in~\cite{BW89} uses an explicit injection which is easily computable in polytime,
see e.g.~\cite{CPP22}.

\smallskip

For an element \ts $x\in X$ \ts and integer \ts $1\le a \le n$, let \ts
$\EE(P,x,a)$ \ts be the set of linear extensions \ts $f\in \EE(P)$ \ts such
that \ts $f(x)=a$.  Denote by \. $e(P, x,a):=\bigl|\EE(P, x,a)\bigr|$ \.
the number of such linear extensions.
\smallskip

\begin{theorem}[{\rm \defn{Stanley inequality}~\cite{Sta81}}]\label{t:Sta}
In the notation above, we have:
\begin{equation*}
e(P, x,a)^2 \ \ge \ e(P, x,a+1) \,\. e(P, x,a-1).
\end{equation*}
\end{theorem}

\smallskip

The original proof in~\cite{Sta81} uses the Alexandrov--Fenchel inequalities applied
to \defng{order polytopes} (see~$\S$\ref{sec:completesquaresandnonmonotone}),
suggesting that a direct combinatorial proof might not exist.  Recently, Stanley's
inequality has been the subject of intense recent investigation, see~\cite{CPP21,SvH20}.
Notably, an elementary linear algebraic proof of the inequality is given in~\cite{CP21}.

\smallskip

\noindent\refstepcounter{openproblemscounter}$(\theopenproblemscounter)$ \.
Is it true that \. $e(P, x,a)^2 \ts \geqslant_\# \ts e(P, x,a+1) \, e(P, x,a-1)$?

\smallskip

\begin{theorem}[{\defn{XYZ inequality}~\cite{She82}}]
Let $P=(X,\prec)$, $x,y,z\in X$ incomparable elements.  Let \.
$P_{xy}:= P\cup \{x\prec y\}$, \. $P_{xz}:= P\cup \{x\prec z\}$ \. and \. $P_{xyz}:= P\cup \{x\prec y, x\prec z\}$.
Then:
\begin{equation*}
e\bigl(P\bigr) \, e\bigl(P_{xyz}\bigr) \ \ge \
e\bigl(P_{xy}\bigr) \, e\bigl(P_{xz}\bigr).
\end{equation*}
\end{theorem}

\smallskip

This correlation inequality is derived by Shepp from the \defng{FKG inequality}
mentioned in~$\S$\ref{sec:ahlswededaykinkleitman},
again suggesting that a direct combinatorial proof might not exist.

\smallskip

\noindent\refstepcounter{openproblemscounter}$(\theopenproblemscounter)$ \.
Is it true that \. $e\bigl(P\bigr) \, e\bigl(P_{xyz}\bigr) \ts \geqslant_\# \ts
e\bigl(P_{xy}\bigr) \, e\bigl(P_{xz}\bigr)$?

\smallskip

The closest to resolving the
problem in the positive is the combinatorial (but not injective!)
proof in~\cite{BT02}, which involves an
application of {\sc BipartiteUnbalance}, see \ts $\S$\ref{ss:TFNP-Background}
and~$\S$\ref{sec:unbalancedflowseparation}.

\smallskip

\subsection{Ranks of matrices}
In~\cite{GKPT16}, the authors show the $\sharpP$-hardness of the dimension $\dim\partial^* f$ of the
space of partial derivatives of a polynomial~$f$,
when the input is given in the sparse presentation (i.e., as a list of monomials that have nonzero coefficients).
They explicitly state as an open question whether or not \ts $\dim\partial^* f$ \ts or \ts $\dim\partial^{=k} f$ \ts are in~$\ts\sharpP$.
When $f$ is homogeneous, then \ts $\partial^{=k} f$ \ts is a special case of a so-called \defng{Young flattening}.

For a partition~$\la$, let \ts $S^\la \IC^n$ \ts denote the irreducible $\GL_n$-representation of type~$\la$, i.e.,
for \ts $\la=(n)$, we have \. $S^{(d)} \IC^n = \IC[x_1,\ldots,x_n]_d$.
Similarly, for \ts $\la=(1^d)$,  we have \.
$S^{(1^d)} \IC^n = \bigwedge^d \IC^n$.
Pieri's rule states that if the Young diagrams of two partitions
\ts $\la \subseteq \mu$ \ts differ by $d$ boxes, at most by~1 in each column,
and if they both have least $n$ many rows,
then there is a unique nonzero $\GL_n$-equivariant map \. $S^{(d)}\IC^n\otimes S^\la \IC^n \to S^\mu \IC^n$.
For every \ts $f \in S^{(d)}\IC^n$ \ts this induces a linear map \.
$\mathcal F_{\la,\mu}:S^{(d)} \to \textup{End}(S^\la\IC^n,S^\mu\IC^n)$,
see e.g.~\cite{Sam09, HI21}.
The rank of this map has been used to derive lower bounds in algebraic complexity theory, see e.g.\ \cite{LO15, Lan15, Far16}.

\smallskip

\noindent\refstepcounter{openproblemscounter}$(\theopenproblemscounter)$ \.
Is the map \. $(\la,\mu,d,n,f) \mapsto \textup{rank}(\mathcal F_{\la,\mu}(f))$ \. in \ts $\sharpP$,
where~$f$ is given in the sparse presentation?

\smallskip

In \cite{BIMPS20}, a very similar lower bounds technique, based on \cite{HL16},
is used to show that a polynomial can be computed by an algebraic branching program
of a certain format \emph{only when approximations are allowed}.

\smallskip

\subsection{Orbit closures and representation theoretic multiplicities}
\label{ss:open-orbits}
The Kronecker and the plethysm coefficients
have interpretations as \defng{dimensions of highest weight vector spaces} as follows.
For a $\GL_m$-representation $V$ a \emph{weight vector} is a vector $v$
that is rescaled by the action of the algebraic torus: \.
$\diag(\alpha_1,\ldots,\alpha_m) \ts v \ts = \ts \alpha_1^{\la_1}\cdots\alpha_n^{\la_m}\ts v$.
The list \ts $\la\in\IZ^m$ \ts is called the \emph{weight} of~$v$.

A \emph{highest weight vector} is a weight vector that is fixed under the action of a maximal unipotent subgroup $U$ (for example the upper triangular matrices with 1s on the main diagonal).
In this case, $\la$ is always a partition.
Highest weight vectors of weight $\la$ form a vector space that we denote by $\HWV_\la(V)$. In general, we have $\dim\HWV_\la(V)=\mult_\la(V)$, where $\mult_\la(V)$ is the number of irreducible representations of type $\la$ in any decomposition of $V$ into irreducibles (this number does not depend on the actual decomposition).  We have:
\[
p_\la(d,n) \ = \ \dim\HWV_\la\big(S^{(d)}(S^{(n)}\IC^m)\big) \ = \ \mult_\la\big(S^{(d)}(S^{(n)}\IC^m)\big)
\]
for $m$ large enough. In fact, the dimension is zero if the number of rows of $\la$ is larger than $m$, and it is a fixed number otherwise.
An analog argument can be done for the triple product $\GL_m \times \GL_m \times \GL_m$ on the space of order 3 tensors $\tensor^3\IC^m$:
\[
g(\la,\mu,\nu) \ = \ \dim\HWV_{(\la,\mu,\nu)}\big(S^{(d)}(\tensor^3\IC^m)\big) \ = \ \mult_{(\la,\mu,\nu)}\big(S^{(d)}(\tensor^3\IC^m)\big).
\]
For a $\GL_m$-orbit closure\footnote{indeed, for any $\GL_m$-variety that is closed under rescaling, i.e., a projective variety} $Z\subseteq S^{(n)}\IC^m$, the plethysm coefficient decomposes:
\[
p_\la(d,n) \ = \ \mult_\la\big(I(Z)_d\big) \, + \, \mult_\la\big(\IC[Z]_d\big),
\]
We write \. $p_\la(d,n)_{Z^+} := \mult_\la(\IC[Z]_d)$ \.
and \. $p_\la(d,n)_{Z^-} := \mult_\la(I(Z)_d)$.
Analogously for the tensor setting, for a $\GL_m^3$-orbit closure $Z$, with $d$ being the number of boxes in $\la$ and also in $\mu$ and $\nu$,
we have:
\[
g(\la,\mu,\nu) \ = \ \mult_{(\la,\mu,\nu)}\big(I(Z)_d\big) \, + \, \mult_{(\la,\mu,\nu)}\big(\IC[Z]_d\big)
\]
We write \. $g(\la,\mu,\nu)_{Z^+} := \mult_{(\la,\mu,\nu)}(\IC[Z]_d)$ \.
and \. $g(\la,\mu,\nu)_{Z^-} := \mult_{(\la,\mu,\nu)}(I(Z)_d)$.
These numbers are zero when $m$ is too small, but are constant otherwise.

If $Z$ is an orbit closure of a point $q$, then
an upper bound for
$p_\la(d,n)_{Z^+}$ (and for \ts $g(\la,\mu,\nu)_{Z^+}$ \ts in the tensor setting),
is given by the dimension of the invariant space of the irreducible
$\GL_m$-representation $V_\la$ of type~$\la$
under the action of the stabilizer $H$ of~$q$, i.e.
$p_\la(d,n)_{Z^+} \leq \dim V_\la^H$ \. or \.
$g(\la,\mu,\nu)_{Z^+} \leq \dim(V_\la\otimes V_\mu\otimes V_\nu)^H$, see \cite{BLMW11}.
We denote this quantity by \ts $u_\la(q)$.

The \defng{geometric complexity} \ts paper
\cite{MS08} pioneered studying a subset of the following multiplicities,
where input partitions and input numbers are always given in unary.
Let
$${\det}_{n}  \ := \ \sum_{\pi\in\aS_n} \sgn(\pi) \. \prod_{i=1}^n x_{i\ts \pi(i)} \quad \text{and let}  \quad
\per_n \ := \ \sum_{\pi\in\aS_n} \. \prod_{i=1}^n x_{i\ts \pi(i)}\..
$$

\smallskip

\noindent\refstepcounter{openproblemscounter}\label{open:detI}$(\theopenproblemscounter)$ \.
For \. $Z = \overline{\GL_{n^2}\det_n}$,
is \. $(d,n,\la)\mapsto p_\la(d,n)_{Z^\pm}$ \. in $\ts\sharpP$? \,
Is \. $(d,n,\la)\mapsto u_\la(\det_n)$ \. in \ts $\sharpP$?
These problems are studied in \cite{Kum15, IP17, BIP19}.

\smallskip

\noindent\refstepcounter{openproblemscounter}$(\theopenproblemscounter)$ \.
For \. $Z = \overline{\GL_{n^2}\per_n}$,
is $(d,n,\la)\mapsto p_\la(d,n)_{Z^\pm}$ in $\ts\sharpP$?
Is \. $(d,n,\la)\mapsto u_\la(\per_n)$ \. in $\ts\sharpP$?
See \cite[eq~(5.5.2)]{BLMW11}.

\smallskip

Let
\[
\textup{IMM}_{m}^{(n)} :=
\textup{tr} \. \left[
\left(\begin{smallmatrix}
x_{1, 1}^{(1)} & \cdots & x_{1, m}^{(1)} \\
\vdots & \ddots & \vdots \\
x_{m, 1}^{(1)} & \cdots & x_{m, m}^{(1)}
\end{smallmatrix}\right)
\cdots
\left(\begin{smallmatrix}
x_{1, 1}^{(n)} & \cdots & x_{1, m}^{(n)} \\
\vdots & \ddots & \vdots \\
x_{m, 1}^{(n)} & \cdots & x_{m, m}^{(n)}
\end{smallmatrix}\right)
\right]
\quad \text{and let}  \quad
\textup{Pow}_{m}^{(n)} :=
\textup{tr} \. \left[
\left(\begin{smallmatrix}
x_{1, 1}^{(1)} & \cdots & x_{1, m}^{(1)} \\
\vdots & \ddots & \vdots \\
x_{m, 1}^{(1)} & \cdots & x_{m, m}^{(1)}
\end{smallmatrix}\right)^n
\right]
\]

\smallskip

\noindent\refstepcounter{openproblemscounter}$(\theopenproblemscounter)$ \.
For \. $Z = \overline{\GL_{m^2n}\textup{IMM}_m^{(n)}}$,
is \. $(d,n,m,\la)\mapsto p_\la(d,n)_{Z^\pm}$ \. in $\sharpP$?
Is $(d,n,m,\la)\mapsto u_\la(\textup{IMM}_m^{(n)})$ \. in~$\ts\sharpP$?

\smallskip

\noindent\refstepcounter{openproblemscounter}$(\theopenproblemscounter)$ \.
For \. $Z = \overline{\GL_{m^2}\textup{Pow}_m^{(n)}}$,
is \. $(d,n,m,\la)\mapsto p_\la(d,n)_{Z^\pm}$ \. in $\sharpP$?
Is \. $(d,n,m,\la)\mapsto u_\la(\textup{Pow}_m^{(n)})$ \. in~$\ts\sharpP$?
This last quantity has been studied in \cite{GIP17}.

\smallskip

\noindent\refstepcounter{openproblemscounter}\label{open:powersumI}$(\theopenproblemscounter)$ \.
For \. $Z = \overline{\GL_{m}(x_1^n+x_2^n+\cdots+x_m^n)}$,
is \. $(d,n,m,\la)\mapsto p_\la(d,n)_{Z^\pm}$ \. in $\sharpP$?
Is \. $(d,n,m,\la)\mapsto u_\la(x_1^n+x_2^n+\cdots+x_m^n) \in \ts\sharpP$?

\smallskip

\noindent\refstepcounter{openproblemscounter}\label{open:productI}$(\theopenproblemscounter)$ \.
For \. $Z = \overline{\GL_{n}(x_1\cdots x_n)}$,
is \. $(d,n,\la)\mapsto p_\la(d,n)_{Z^\pm}$ in $\ts\sharpP$?

\smallskip

The multiplicities in \eqref{open:powersumI} and \eqref{open:productI} are used in \cite{DIP20, IK20}, see also \cite[Thm.~22.4.1]{BI18}.
The multiplicities in \eqref{open:productI} play a central role in the Foulkes--Howe approach to the \defng{Foulkes conjecture},
see e.g.~\cite{McK08,CIM17}.
For a partition \ts $\la$ \ts with \ts $|\la|=nd$, we actually have \ts $u_\la(x_1\cdots x_n)=p_\la(n,d)$.

\smallskip

Let \ts $e_i\in\IC^m$ \ts denote the $i$-th standard basis vector, and let \.
$\{e_{i,j} \. : \. 1 \leq i,j\leq m\}$ \. be a basis of~$\IC^{m^2}$.
 Consider the unit tensor and the  \defn{matrix multiplication tensor}:
 $$
E_m \ := \ e_1^{\otimes 3}+\cdots+e_m^{\otimes 3}
\quad\text{ and }\quad
M_m \  := \ \sum_{i,j,k=1}^m \. e_{i,j} \otimes e_{j,k} \otimes e_{k,i}\.
$$

\noindent\refstepcounter{openproblemscounter}\label{open:unittensorI}$(\theopenproblemscounter)$ \.
For \. $Z = \overline{\GL_{m}^3(E_m)}$,
is \. $(d,m,\la,\mu,\nu)\mapsto g(\la,\mu,\nu)_{Z^\pm}$ \. in $\sharpP$?
Is $(d,m,\la,\mu,\nu)\mapsto q_{(\la,\mu,\nu)}(E_m) \in \ts \sharpP$?

\smallskip

\noindent\refstepcounter{openproblemscounter}\label{open:mamutensorI}$(\theopenproblemscounter)$ \.
For \. $Z = \overline{\GL_{m^2}^3(M_m)}$,
is \. $(d,m,\la,\mu,\nu)\mapsto g(\la,\mu,\nu)_{Z^\pm}$ \. in $\sharpP$?
Is $(d,m,\la,\mu,\nu)\mapsto q_{(\la,\mu,\nu)}(M_m) \in \ts\sharpP$?

\smallskip

The coefficients in open problems \eqref{open:unittensorI} and
\eqref{open:mamutensorI}
were used in \cite{BI11, BI13STOC} to study the \defng{complexity of matrix multiplication}.
For \eqref{open:unittensorI}, see also \cite[Thm.~22.4.3]{BI18}.
See \cite{BILPS21} for more tensor setting examples.

In~\cite{BHIM22}, the authors give software to compute $p_\la(d,n)_{Z^\pm}$
in small cases.
The papers \cite{BIJL18} and \cite{GIMOWW20} point out the hardness of computing these quantities in general.
Notably, Grochow \cite{Gro15} reinterprets numerous lower bounds results in \defng{algebraic complexity theory} \ts
as group orbit (closure) separations, and thus they fall into the category of this paragraph,
but the focus in \cite{Gro15} is not directly on the multiplicities,
but on the more general approach of finding separating modules.

\smallskip

\subsection{Algorithms for the binomial basis on an affine variety}
Expressing a polynomial in its binomial basis can be done using a greedy algorithm starting from the top total degrees and proceeding to the lower degrees. Hence, given a polynomial $\varphi$ in its monomial basis, one can check in polynomial time whether or not $\varphi$ is binomial-good. We have seen that it is important to understand when a coset $\varphi+I$ contains a binomial-good element (we say that $\varphi+I$ is binomial-good), where $I$ is an ideal.

\smallskip

\noindent\refstepcounter{openproblemscounter}$(\theopenproblemscounter)$ \.
If we assume that a set of generating elements of $I$ is presented in the binomial basis, what is the complexity of determining whether or not $\varphi+I$ is binomial-good?
In which cases do we get efficient algorithms?


\bigskip

\subsection*{Acknowledgements}
We are grateful to Swee Hong Chan, Nikita Gladkov and Greta Panova for
numerous helpful discussions and remarks on the subject.  We thank
Joshua Grochow and Greg Kuperberg for many inspiring conversations
on computational complexity over the years.  
We learned about \cite{HVW95} and \cite{Bei97} only 
after the paper was written; we thank Lane Hemaspaandra 
who kindly for pointed out these references.
We thank Markus Bl\"aser and Paul Goldberg for helpful comments on a draft version of this paper.

We are grateful to the Oberwolfach Research Institute for Mathematics
where we started this project in the innocent times of February of 2020.
The first author was partially supported by the DFG grant IK 116/2-1 and the EPSRC grant EP/W014882/1.
The second author was partially supported by the NSF grant CCF-2007891.

\end{document}